\documentclass[%
reprint,
superscriptaddress,
amsmath,amssymb,
aps,
pra,
longbibliography,
]{revtex4-2}

\usepackage{amsthm}
\usepackage{physics}
\usepackage{amsmath}
\usepackage{graphicx}
\usepackage{qcircuit}
\usepackage{ragged2e} 

\usepackage{booktabs}
\usepackage{multirow}
\usepackage[english]{babel}

\usepackage{hyperref}
\hypersetup{colorlinks=true, citecolor=blue, urlcolor=blue, linkcolor=blue}

\newcommand{\superproj}[1]{|#1\rangle\rangle\langle\langle#1|}

\newtheorem{definition}{Definition}
\newtheorem{theorem}{Theorem}

\newtheorem{lemma}{Lemma}

\begin{document}

\preprint{APS/123-QED}

\title{Characterization of Unlearnable Noise with Mid-Circuit-Measurement-Based Cycle Benchmarking 
}

\author{M. H. Cheng}
 \email{mhc2617@ic.ac.uk}
\affiliation{Blackett Laboratory, Imperial College London, London, SW7 2AZ, UK\\
}
\affiliation{
 Fraunhofer Institute for Industrial Mathematics, Fraunhofer-Platz 1, 67663 Kaiserslautern, Germany\\
}

\author{Stefano Mangini}
\affiliation{Algorithmiq Ltd, Kanavakatu 3C 00160 Helsinki, Finland\\}
\affiliation{Department of Physics, University of Helsinki, P.O. Box 43, FI-00014 Helsinki, Finland.}

\author{V. Bartsch}
\affiliation{Fraunhofer Center for Maritime Logistics, Blohmstrasse 32, 21079 Hamburg, Germany\\}

\author{A. C. Medina}
\affiliation{
 Fraunhofer Institute for Industrial Mathematics, Fraunhofer-Platz 1, 67663 Kaiserslautern, Germany\\
}

\author{Sergey N. Filippov}
\affiliation{Algorithmiq Ltd, Kanavakatu 3C 00160 Helsinki, Finland\\}

\author{Matteo A. C. Rossi}
\affiliation{Algorithmiq Ltd, Kanavakatu 3C 00160 Helsinki, Finland\\}

\author{M. S. Kim}
\affiliation{QOLS, Blackett Laboratory, Imperial College London, London, SW7 2AZ, UK\\
}
\date{\today}

\begin{abstract}
    Noise characterization of multi-qubit entangling Clifford operations is a key practical bottleneck for quantum error mitigation and for the calibration, validation, and optimization of quantum error-correction protocols, especially in the presence of state preparation and measurement (SPAM) errors. Although cycle benchmarking can isolate some Pauli error components, it cannot resolve the problem of coupled error parameters, which leads to unlearnable degrees of freedom even in simple noisy gates, not to mention general $n$-qubit Clifford gates. Here we introduce mid-circuit-measurement-based generalized cycle benchmarking, a framework that makes otherwise unidentifiable Pauli fidelities and non-Markovian noise learnable via repeated measurements and classical post-processing. Applying the deferred feed-forward principle to generalized cycle benchmarking, we show that an insertion of mid-circuit measurements can reverse Pauli cycles induced by a general Clifford gate. This fact enables us to reveal a Pauli-noise learnability condition for Clifford gates. Assuming sufficient state preparation quality, we numerically demonstrate the feasibility of characterizing the previously unlearnable noise components. We implement the protocol on superconducting quantum processing units and validate its effectiveness in disambiguating the coupled noise components, benchmarked against conventional tomography. Finally, we observe consistent measurement-induced bit-flip bias and non-Markovian correlations, which define a range of applicability for the Pauli noise model and the proposed noise-characterization protocol.
\end{abstract}

\maketitle

Quantum noise is a fundamental obstacle to the practical application of quantum computers, and noise reduction schemes have become central to approaching quantum advantage in solving classically intractable problems~\cite{Kim2023utility, RobledoMoreno2025chemistry, Cai2023error}. Quantum noise tomography therefore emerges as an important tool for scaling up quantum hardware. This will allow hardware precision benchmarking, error mitigation~\cite{temme_2023prob, Filippov2023tensor, Filippov2024scalability, Mangini2024, Fischer2026}, and the advancement of efficient quantum error correction (QEC)~\cite{Bennett_1996qmix-state-entanglement}. State preparation and measurement (SPAM) errors remain a challenge for reliable noise isolation~\cite{Onorati2023fittingquantumnoise, Jayakumar_2024}. Cycle benchmarking (CB) and gate set tomography were subsequently introduced to account for unwanted extra noise couplings during tomography~\cite{Erhard2019-CBog, Nielsen_2021gatesettomography}.

The underlying principle of CB is simple: to isolate the SPAM error, one can magnify the quantum gate error via a repeated gate operation in between state preparations and measurements. The varying gate noise magnitude can be fitted away from the constant contribution of the SPAM errors. CB excels in characterizing a Pauli channel, where the channel's eigenvalues can be unambiguously amplified, and the collection of noise coefficients is consistent for each cycle. However, CB is limited by its Markovianity assumption and the types of gates that it can effectively characterize~\cite{Erhard2019-CBog}. A general $n$-qubit Clifford gate induces an irreducible cycle of Pauli propagation, which cannot be canceled with single-qubit Clifford gates~\cite{Chen_2023learnability}. For a general $n$-qubit Clifford gate, the associated Pauli transfer graph grows exponentially with $n$, which can lead to an exponentially large number of unlearnable (gauge) degrees of freedom. As a result, standard CB becomes increasingly difficult to apply as a full noise-characterization tool for large multi-qubit Clifford gates on near-term devices. The bottleneck to CB characterization is the assumption that single-qubit Clifford gates are the only noiseless tomographic primitives, but the rapid improvement of quantum hardware now offers a new option. Mid-circuit measurements (MCMs) have recently emerged as an important single-qubit capability on multiple quantum computing platforms~\cite{Bluvstein2024, qfvd-93lw, Bengtsson2024willowmcm}. Although known for their role in QEC syndrome measurement, MCMs have also found applications in quantum tomography~\cite{Zhuravlev2020}, non-Hermitian Hamiltonian dynamics~\cite{Luchnikov2017, Grimaudo2020}, and quantum algorithms through dynamic circuits~\cite{PhysRevLett.127.100501, lemelin2025mcmprimitive}. Error characterization schemes for MCMs have been proposed in learning composite Clifford gates and MCMs~\cite{Chen_2023generalized, Hines_2025}, and as fidelity assessment for hardware MCMs~\cite{Hothem2025errorratenatcomm}.

In this paper, we present a comprehensive study of MCMs as tomographic primitives for generalized CB. The learnability of a given Pauli fidelity depends on how the Clifford gate rearranges the operator's Pauli weight structure. Based on this, we introduce the \textit{Pauli weight mismatch} --- a combinatorial quantity capturing the difference in support between the input and output Pauli operators. Building on Ref.~\cite{Chen_2023learnability}, we show that the MCMs cost of learning all Pauli fidelities of an entangling $n$-qubit Clifford gate, including the previously unlearnable ones, is fully characterizable by this single quantity. The Pauli weight mismatch plays a dual role: it identifies which Pauli fidelities are learnable in standard CB, and it specifies precisely where MCMs must be inserted in generalized CB. Hence, the Hamming weight of the mismatch is the MCM resource cost of breaking a given Clifford cycle. The technical insight that enables MCMs-based CB is the deferred feed-forward principle, a method that allows the implementation of classical feed-forward of MCMs in post-processing. Our method generalizes the post-processing techniques introduced in Refs.~\cite{Chen_2023generalized, Hines_2025}. The complete generalized CB protocol is summarized in Figure~\ref{fig:benchmarking_and_graph}. The Pauli weight mismatch framework also admits graph-theoretic pictures that generalize Ref.~\cite{Chen_2023learnability}, illustrated in Figure~\ref{fig:pauli_transfer_plot}.

We study the requirements on noisy MCMs for reliable Clifford noise learning and the situation where the Clifford gate consists of local noise models. The learnability condition of the Pauli weight mismatch holds when MCMs realize high-fidelity post-measurement projections onto $\ket{0}$ and $\ket{1}$ in the $z$-basis. All noise components associated with both MCMs and Clifford gates are learnable, including the classical readout error. As a corollary, we show that Pauli fidelities of the CNOT gate remain approximable even when only the preparation of $\ket{0}$ is effectively noiseless. Meanwhile, we numerically demonstrate how to exploit the noise locality constraints on the Pauli weight mismatch to reduce the MCMs resource cost of learning the local noise models of the Clifford gate.

Beyond its role as a resource for resolving unlearnable noise, MCMs serve a second function in our framework: as a diagnostic for the Markovianity assumption that underlies CB. We show that under separable Markovian noise, the flip statistics of repeated MCM readouts follow a binomial distribution governed by a single Pauli fidelity. Any deviation from this distribution is therefore a direct, quantitative signature of non-Markovian dynamics --- invisible to standard CB, which collapses MCM information into a single final measurement. This binomial analysis turns MCMs into a probe of the very assumption on which the rest of the protocol relies, defining the regime in which the Pauli noise model is operationally valid.

We validate both roles experimentally on \texttt{ibm-aachen} and \texttt{ibm-pittsburgh}. For the resource role, we extract unlearnable Pauli fidelity pairs of CNOT gates and achieve within-CPTP-bound characterization on both \texttt{ibm-aachen} and \texttt{ibm-pittsburgh}, surpassing the precision of direct fidelity measurements. For the diagnostic role, the binomial analysis reveals a measurement-induced bit-flip bias attributable to classical readout asymmetry, and a leakage-induced exponential tail consistent with non-Markovian memory in the readout process. Together, our results demonstrate the potential for integrating MCMs with CB for the characterization of larger Clifford gates, combining the sampling efficiency of CB with partial SPAM robustness and an explicit handle on the breakdown of Markovianity.

In Section~\ref{section: Framework}, we review the existing CB framework, introduce quantum instruments with classical feed-forward (QICF), and expose the learnability of MCM noise. We establish the theoretical basis for generalized CB with MCMs for $n$-qubit entangling Clifford gates, encompassing a framework for Pauli noise decoupling and the detection of non-Markovianity. In Section~\ref{section: algorithm}, we condense these theoretical tools into a set of algorithmic primitives that implement the quantum circuit and post-processing pipeline depicted in Figure~\ref{fig:benchmarking_and_graph}. In Section~\ref{section: numerics}, we perform numerical noise characterizations of CNOT gates under varying noise constraints on the MCMs, and demonstrate the scalability of the protocol for noise learning of CNOT ladders. Finally, in Section~\ref{section: Experiment}, we perform generalized CB with MCMs on CNOT gates implemented on the \texttt{ibm-aachen} and \texttt{ibm-pittsburgh} quantum processors. Beyond demonstrating a partially SPAM-robust noise characterization of unlearnable noise pairs, we perform an analysis of non-Markovian deviations associated with leakage errors.

\section{Theoretical Framework}
\label{section: Framework}
We briefly define the conventions and superoperators used throughout this paper~\cite{greenbaum2015introductionquantumgateset}. An \textit{unspecified} $n$-qubit Pauli is labeled with $\alpha, \beta$. An \textit{unspecified but single-qubit} Pauli is labeled with $i,j,k$. We reserve the indices $l, m, n$, with the exception of Figure~\ref{fig:pauli_transfer_plot}, to respectively denote the number of measured qubits, unmeasured qubits, and total qubits in a multi-qubit MCM process. We use $p, q, r, s$ to denote the qubit locations on which the \textit{specified} Pauli operators act. We avoid superoperator representation on Pauli observables to keep the presentation clear. The bit-flip channel of a single-qubit Pauli-$X$ at site $p$ is written as $\mathcal{X}_p(\cdot) = X_p(\cdot)X_p$. For binary vectors $\vec{v} \in \mathbb{Z}_2^N$ and $\vec{k} \in \mathbb{Z}_2^l$, we adopt the vectorized notation $\mathcal{X}^{\vec{v}}(\cdot) = X^{\vec{v}}(\cdot)X^{\vec{v}}$, where $X^{\vec{v}} = \prod_{p=1}^N X_p^{\vec{v}[p]}$ is a product of Pauli-$X$ operators over the $p$-th qubits. The vector sign is removed when $\vec{v}$ is one-bit. The projection superoperator onto the computational basis state $\ket{\vec{k}}$ is denoted $\superproj{\vec{k}} = \dyad{\vec{k}}(\cdot)\dyad{\vec{k}}$. A Pauli channel is denoted as $\Lambda(\cdot)$ and indexed if the channel is conditional.

\subsection{Cycle Benchmarking}

In quantum noise characterization, the noise dynamics of each quantum gate are typically assumed to be mutually independent. A noisy $n$-qubit Clifford gate can be modeled as the composition of an ideal Clifford gate $G$ followed by a $4^n$-dimensional noise channel $\Omega$, decomposed on the Pauli basis $P_\alpha, P_\beta \in \mathcal{P}^n = \{I, X, Y, Z\}^{\otimes n}$:
\begin{equation}
    (\Omega \circ \mathcal{G})(\cdot) = \sum_{P_\alpha, P_\beta \in \mathcal{P}^n}c_{\alpha, \beta}P_{\alpha}G(\cdot)G^\dagger P_{\beta},
\end{equation}
where $c_{\alpha, \beta}$ is a positive semi-definite scalar, and $\mathcal{G}(\cdot) = G(\cdot)G^\dagger$. The goal of noise characterization is to extract information about $c_{\alpha, \beta}$ so that an inverse channel $\Omega^{-1}$ can be constructed, which finds application in, for example, quantum error mitigation~\cite{RevModPhys.95.045005,temme_2023prob,Kim2023utility,Filippov2023tensor,Filippov2024scalability,Fischer2026}.

Given the limited $4^n$ Pauli measurement basis, a SPAM-free characterization of a complete noise channel with $4^n\times4^n$ degrees of freedom is extremely challenging, if not impossible~\cite{Erhard2019-CBog}. Pauli twirling is introduced to obtain a simplified Pauli channel from the original~\cite{PhysRevA.88.012314}:
\begin{equation}
    (\mathcal{T}(\Omega) \circ \mathcal{G})(\cdot) = \frac{1}{4^n}\sum_{T_\alpha \in \mathcal{P}^n}G(T_\alpha)^\dagger(\Omega \circ G)\{T_\alpha(\cdot)T_\alpha\}G(T_\alpha),
    \label{eqn: twirled Clifford gate}
\end{equation}
where the channel is sandwiched and averaged over $4^n$ Pauli operators, canceling the coherent degrees of freedom. Given a random sequence of Paulis $\{T_1, T_2, \ldots, T_n\}$, Pauli twirling can be approximated when noise channels are local:
\begin{equation}
    (\mathcal{T}_N(\Omega) \circ \mathcal{G})(\cdot) = \frac{1}{N}\sum_{\alpha = 1}^NG(T_\alpha)^\dagger(\Omega \circ G)\{T_\alpha(\cdot)T_\alpha\}G(T_\alpha),
    \label{eqn: twirled approx Clifford gate}
\end{equation}
yielding the effective channel $\mathcal{T}(\Omega) = \Lambda$. A Pauli channel is convenient to characterize because its eigenbasis consists of the Pauli operators, with eigenvalues known as Pauli fidelities:
\begin{align}
    \Lambda(\cdot) &= \sum_{P_\alpha \in\mathcal{P}^n}c_{\alpha, \alpha}P_{\alpha}(\cdot)P_{\alpha}, \label{eqn: pauli channel}\\
    \text{Tr}(P_\beta\Lambda(P_\alpha)) &= \lambda_{P_\alpha}\delta_{\alpha, \beta},
    \label{eqn: PTM}
\end{align}
where $\lambda_{P_\alpha}\delta_{\alpha,\beta}$, together with a Kronecker delta, is the diagonal Pauli transfer matrix (PTM). This structure allows stabilizer states to serve as probes that accumulate noise eigenvalues during quantum state tomography, after which the Pauli channel is reconstructed via the Walsh--Hadamard transform of the PTM~\cite{10.1145/3408039}.

A Pauli noise representation is essential for eliminating SPAM contamination. In typical noise channel tomography, the Pauli noise channel of a quantum gate is sandwiched between state preparation $\mathcal{S}(\cdot)$ and measurement $\mathcal{M}(\cdot)$ errors~\cite{Chen_2023learnability}:
\begin{equation}
    \mathcal{M}(\Lambda(\mathcal{S}(P_{\alpha}))) = (A\lambda_{P_\alpha}+B)P_{\alpha},
\end{equation}
for some SPAM errors $A$ and $B$, so that the factor $(A\lambda_{P_\alpha}+B)$ prevents unbiased extraction of $\lambda_{P_\alpha}$. CB addresses this by introducing an additional degree of freedom --- the cycle --- defined as the sequential application of the channel $\Lambda$. By fitting the observable $f(N)$ as a function of the number of cycles~\cite{Jayakumar_2024, greenbaum2015introductionquantumgateset}:
\begin{equation}
    f(N) = \mathcal{M}(\Lambda^N(\mathcal{S}(P_{\alpha}))) = (A\lambda_{P_\alpha}^N+B)P_{\alpha},
\end{equation}
the contributions $A$, $B$, and $\lambda_{P_\alpha}$ are decoupled, enabling effective characterization of the Pauli fidelities $\lambda_{P_\alpha}$.

In the special case of Pauli channel characterization, CB is sufficient for complete SPAM-free noise tomography. However, when CB is applied to a noisy Clifford gate that admits the gate-noise decomposition $\Lambda\circ\mathcal{G}$, distinct Pauli channel coefficients are actively coupled together. For example, in the CNOT gate, a weight-1 Pauli eigenvalue couples with a weight-2 Pauli eigenvalue. Let $i, j, k$ index single-qubit Paulis in $\{I, Z, X, Y\}$, so that $P_{0i} = I\otimes P_i$ and $P_{jk} = P_j\otimes P_k$. The CB of CNOT yields:
\begin{equation}
    \mathcal{M}((\Lambda\circ \mathcal{G})^{2N} (\mathcal{S}(P_{0i}))) = (A(\lambda_{0i}\lambda_{jk})^N+B)P_{0i}.
    \label{eqn: dilemma}
\end{equation}
Such pairs are termed \emph{unlearnable noise}: the Pauli eigenvalues are identifiable only up to a multiplicative gauge factor $g$, since $\lambda_{0i}\lambda_{jk} = (\lambda_{0i}g) \cdot (\lambda_{jk}/g)$ along Clifford-induced Pauli cycles. Due to the pair structure, $\lambda_{0i}\lambda_{jk}$ is learned in $2N$ cycles, reflecting the fact that CB of CNOT gates is implemented in pairs. In general, let us define the Pauli weight mismatch.
\begin{definition}[Pauli Weight Mismatch]
    Let $P_\alpha$ and $P_\beta$ be two $n$-qubit Paulis. Let $\mathrm{pt}(P_\alpha)\in\mathbb{Z}_2^n$ be the Pauli weight vector that records the locations of non-identity Paulis in $P_\alpha$. The Pauli weight mismatch is the bitwise summation between the Pauli weight vectors of $P_\alpha$ and $P_\beta$:
    \begin{equation}
        \delta(P_\alpha, P_\beta):=\mathrm{pt}(P_\beta)\oplus\mathrm{pt}(P_\alpha).
        \label{eqn: mismatch}
    \end{equation}
\end{definition}
Standard CB, illustrated as the repeated operation of circuit block (a) in Figure~\ref{fig:benchmarking_and_graph}, has the following limitation when it comes to characterizing $n$-qubit Clifford gates $\mathcal{G}$~\cite{Chen_2023learnability}.
\begin{theorem}[Learnability of Pauli Fidelity in CB]
    \label{theorem: learnability}
    Let $P_\alpha$ and $P_\beta$ be two $n$-qubit Paulis, and $\mathcal{G}$ an $n$-qubit Clifford gate such that $\mathcal{G}(P_\alpha)=P_\beta$. When single-qubit Clifford gates are assumed to be noiseless, the Pauli fidelity of $P_\beta$, $\lambda_{P_\beta}$, is learnable if there is no Pauli weight mismatch between $P_\alpha$ and $P_\beta$:
    \begin{equation}
        \delta(P_\alpha, P_\beta)=\vec{\mathbf{0}}.
        \label{eqn: learnability}
    \end{equation}
\end{theorem}
For an arbitrary $n$-qubit Clifford gate $\mathcal{G}$, Theorem~\ref{theorem: learnability} asserts that most Pauli fidelities are in fact not learnable, since $\delta(P_\alpha, P_\beta)\neq \vec{\textbf{0}}$ in the general case.
This learnability structure admits a natural graph-theoretic interpretation~\cite{Chen_2023learnability}. The Pauli transitions induced by $\mathcal{G}$ form a cycle graph, whose edges decompose
into two complementary subspaces.
Pauli operators $P_\beta$ satisfying Eq.~(\ref{eqn: learnability}) span the \emph{loop-space}, corresponding to closed loops in the graph, while those violating it span the \emph{cut-space},
corresponding to open lines.
Figure~\ref{fig:pauli_transfer_plot} illustrates these cycle graphs for the CNOT and SWAP gates
on the left-hand side of the large arrows.

As a consequence of this theorem, CB has found only limited applications,
typically restricted to characterizing a small set of two-qubit entangling Clifford gates.
Moreover, the broader significance of $\delta(P_\alpha, P_\beta)$ beyond delimiting
CB learnability has remained unclear. Introducing MCMs as a tomographic primitive resolves both bottlenecks. This observation motivates the general theory of CB with MCMs developed in the following section.

\subsection{Generalized Cycle Benchmarking with MCMs}
\label{section: generalized cycle benchmarking}

\begin{figure*}
    \centering
    \includegraphics[width=1\linewidth]{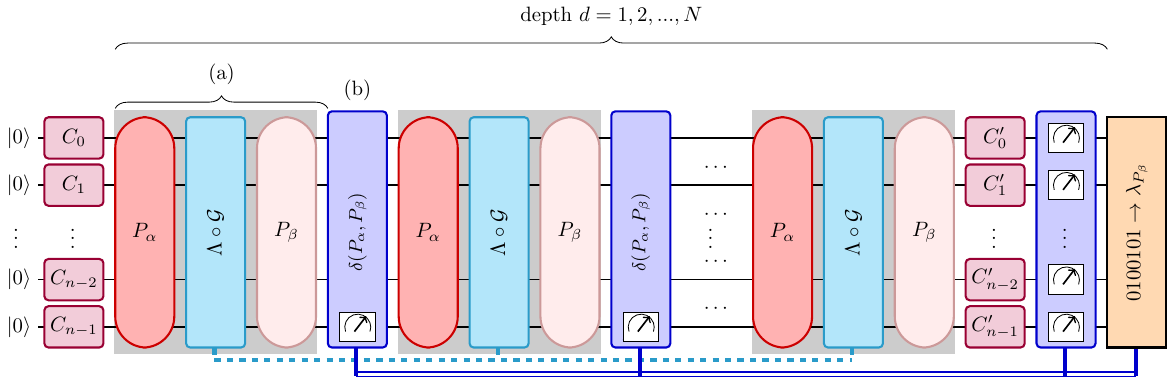}
    \caption{\justifying
Pauli noise learning of $\lambda_{P_\beta}$ of noisy Clifford gate $\Lambda\circ\mathcal{G}$ (cyan) with MCMs with classical feed-forwards (purple). (a) Traditional CB protocol for learning $\Lambda$ suffers from unintended noise coupling induced by $\mathcal{G}$ that transforms $n$-qubit Pauli operator $P_\alpha$ (red) to $P_\beta$ (light pink) as $P_\alpha$ propagates to the measurements for final post-processing (orange). (b) By inserting single-qubit Clifford gates and partial MCMs in between $\Lambda\circ\mathcal{G}$, arbitrary $P_\beta$ is reset back to $P_\alpha$. The Pauli weight mismatch $\delta(P_\alpha,P_\beta)$ of Eq.~(\ref{eqn: mismatch}) computes the location of the partial MCMs. Additionally, non-Markovian processes (dashed cyan) can be characterized by computing the deviation from Markovianity. Classical feed-forwards can be deferred to classical post-processing (purple line) to realize CTRL-$M$ in Eq.~(\ref{eqn: QICF with cerr}). In practice, noise of MCMs can be characterized with (b) alone to enhance the characterization quality.}
    \label{fig:benchmarking_and_graph}
\end{figure*}
\begin{figure*}
    \centering
    \includegraphics[width=1.0\linewidth]{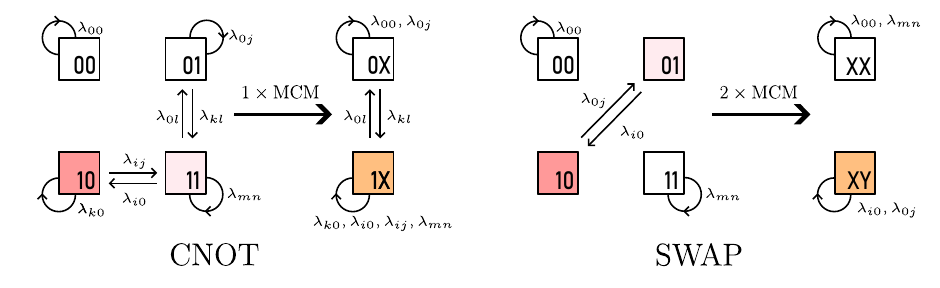}
    \caption{\justifying CNOT/SWAP gates pattern transfer graph under CB with MCMs. A Pauli transfer graph captures the structure of coupled Paulis in terms of their Pauli weight~\cite{Chen_2023learnability}. Coupled Pauli weights between $P_\alpha, P_\beta$ (red, light pink as in Figure~\ref{fig:benchmarking_and_graph}) generate coupled noise coefficients that lead to unintended noise coupling, which limits tomographic precision in CB. For CNOT, the example coupling noise pairs have weight $\text{pt}(P_\alpha)= [1,0], \text{pt}(P_\beta)= [1,1]$. By inserting MCMs to the second qubit $\text{pt}(P_\alpha) \oplus \text{pt}(P_\beta) = [0,1]$, the Pauli noise pair $\{\lambda_{i0}, \lambda_{ij}\}$ following the convention in Eq.~(\ref{eqn: dilemma}) can be decoupled in classical post-processing (orange). Paulis with weight $\text{pt}(P) = [1,X]$ for $X = 0, 1$ become learnable. For SWAP, the coupling noise pairs have weight $\text{pt}(P_i)= [1,0], \text{pt}(P_j)= [0,1]$. MCMs inserted at both qubits $\text{pt}(P_\alpha) \oplus \text{pt}(P_\beta) = [1,1]$ can decouple the noise pair $\lambda_{0j},\lambda_{k0}$. Paulis with weight $\text{pt}(P) = [X,Y]$ for $X = 0, 1$, $Y = 1,0$ become learnable.}
    \label{fig:pauli_transfer_plot}
\end{figure*}

MCMs, a key diagnostic tool in QEC, have recently
been identified as algorithmic primitives~\cite{lemelin2025mcmprimitive}. They have applications to dynamic circuits~\cite{PhysRevLett.127.100501}, but their other roles in quantum computation remain an open question. In this work, we identify that MCMs can also serve as an algorithmic primitive for generalized CB. Specifically, MCMs can be treated as a resource for state preparation that deciphers noise parameters
otherwise unidentifiable through standard CB, even when the MCMs themselves are
susceptible to noise. Meanwhile, MCMs provide an additional degree of freedom to probe the very assumption underlying a CB protocol: Markovianity.

In this section, we establish the theoretical foundations
of this idea, identify realistic conditions under which unlearnable noise of MCMs
and Clifford gates becomes accessible, and motivate two protocols that form the
basis of the algorithms presented later.

\subsubsection{Uniform Stochastic Instruments}

In recent years, uniform stochastic instruments (USIs) have been introduced as a standardized quantum model for noisy MCMs~\cite{beale2023randomizedcompilingsubsystemmeasurements}. Similarly to twirled generic quantum gates, a USI is a noisy MCM under randomized compilation, assuming Markovianity between events of MCMs. In a quantum circuit of $n$ qubits, a single round of USI operation divides the circuit into $m$ unmeasured qubits and $l$ measured qubits ($n = m+l$). Building on the original definition of USI~\cite{mclaren2023stochasticerrorsquantuminstruments}, we define a general form of USI as follows:
\begin{align}
    \mathcal{M}_{\textbf{e}}(\cdot) &=  \sum_{\vec{r}, \vec{k}\in\mathbb{Z}_2^l}q_{\vec{r}}\mathcal{M}_{\vec{k}}(\cdot)\otimes\mathbf{e}_{\vec{k}+\vec{r}}, \label{eqn: USI with cr-err.}\\
    \mathcal{M}_{\vec{k}}(\cdot) &= \sum_{\vec{a}, \vec{b} \in \mathbb{Z}_2^l}p_{\vec{a}, \vec{b}}\left[\Lambda_{\vec{a}, \vec{b}}\otimes \mathcal{X}^{\vec{b}}\superproj{\vec{k}}\mathcal{X}^{\vec{a}}\right](\cdot)
    \label{eqn: standard USI},
\end{align}
where the subscript $\textbf{e}$ is used to denote the classical read-out bits. In this set-up, the classical read-out assignment error in Eq.~(\ref{eqn: USI with cr-err.}) is isolated from the purely quantum errors in Eq.~(\ref{eqn: standard USI}). In other words, $\mathcal{M}_{\mathbf{e}}$ encodes noisy \textit{classical} events $\{\vec{r}\}$ that occur with probability $q_{\vec{r}}$ and contaminate the classical register $\mathbf{e}_{\vec{k}+\vec{r}}$. In the classical noise model, independent classical read-out assignment events are assumed:
\begin{equation}
    q_{\vec{r}} = q_{\vec{r}[0]}\cdots q_{\vec{r}[l]}.
\end{equation}
Each $q_{\vec{r}[i]}$ is the read-out error rate for \textit{both} the $0\rightarrow 1$ and $1\rightarrow 0$ error processes. Randomized compilation naturally averages over the two types of read-out error events (see the derivation in Appendix~\ref{appendix: Read-out error derivation}).

Meanwhile, the purpose of $\mathcal{M}_{\vec{k}}$ is to describe the \textit{quantum} events acting on the $(m+l)$ qubits given an $l$-qubit projection $\superproj{{\vec{k}}}$. Specifically, $\mathcal{M}_{\vec{k}}$ consists of events $\{\Lambda_{\vec{a}, \vec{b}}, \mathcal{X}^{\vec{a}/\vec{b}}\}$ that occur with normalized probability $p_{\vec{a},\vec{b}}$ ($\sum_{\vec{a},\vec{b} \in \mathbb{Z}_2^l}p_{\vec{a}, \vec{b}} = 1$). $\Lambda_{\vec{a}, \vec{b}}$ are Pauli channels (Eq.~(\ref{eqn: pauli channel})), and $\mathcal{X}^{\vec{a}/\vec{b}}$ are bit-flip events. In Figure~\ref{fig:benchmarking_and_graph}, a USI is represented as purple blocks without classical post-processing. Randomized compilation decouples the dependency of error events $\vec{a}, \vec{b}$ from read-outs $\vec{k}+\vec{r}$.

Given the structure of USI, we may derive a stochastic Pauli propagation rule when averaged over measurements:
\begin{theorem}[Pauli Propagation for USIs]
\label{theorem: USI pauli propagation.}
A uniform stochastic instrument (USI) measuring a single qubit $p$, when averaged over all measurement outcomes, acts as a Pauli channel on $Z$-type observables and commutes with bit-flip operations. Formally, define $\mathcal{M}(\cdot) = \sum_{k\in \mathbb{Z}_2}\mathcal{M}_{k}(\cdot)$ as the outcome-averaged USI. For $c\in \mathbb{Z}_2$:
\begin{align}
    \mathcal{M}(Z_p^cZ_q) &= \delta_{Z_p^cZ_q}\,Z_p^cZ_q,
    \label{eqn: stochastic Pauli channel}\\
    \mathcal{M}(\mathcal{X}_p^{c_1}\mathcal{X}_q^{c_2}(\cdot)) &= \mathcal{X}_p^{c_1}\mathcal{X}_q^{c_2}(\mathcal{M}(\cdot)),
    \label{eqn: stochastic channel prop}
\end{align}
where the fidelity coefficient is:
\begin{align}
    &\delta_{Z_p^cZ_q} = p_{00}\lambda_{00}+(-1)^cp_{10}\lambda_{10}+(-1)^{c}p_{01}\lambda_{01}+p_{11}\lambda_{11}, \nonumber\\
    &\Lambda_{ab}(Z^c_pZ_q) = \lambda_{ab}Z^c_pZ_q. \nonumber
\end{align}
\end{theorem}
\begin{proof}
    See Section~\ref{subpendix: pauli propagation} in Appendix~\ref{appendix: all proofs}
\end{proof}
Notice that we provide details on both the Pauli propagation of the Pauli observable in Eq.~(\ref{eqn: stochastic Pauli channel}) and the Pauli bit-flip channel in Eq.~(\ref{eqn: stochastic channel prop}). This is a necessary distinction because measurement projects the anti-commutation relation between the Pauli-$X$ and Pauli-$Z$ operators onto a computational hierarchy. Reversing the roles of Pauli-$X$ and Pauli-$Z$ will result in a loss of quantum information, namely:
\begin{align}
    \mathcal{M}(X_p Z_q) &= 0,\\
    \mathcal{M}(\mathcal{Z}_p^{c_1}\mathcal{X}_q^{c_2}(\cdot))&=\mathcal{I}_p^{c_1}\mathcal{X}_q^{c_2}(\mathcal{M}(\cdot)).
\end{align}
This unique property of USI will allow us to connect classical feed-forward and post-processing while maintaining the invariance of Pauli fidelity.

The complete characterization of USI noise is the prerequisite for a characterization of an $n$-qubit entangling Clifford gate. Refs.~\cite{Chen_2023generalized, Hines_2025} have introduced the learnability theory of USI noise. Here, we develop a generalized criterion on state preparation in USI in order to learn all noise components of a USI when the number of measured qubits $l = 1$. In particular:
\begin{theorem}[Exact Noise Extraction Principle]
\label{theorem: exact}
All Pauli fidelities of a single-qubit USI are exactly learnable from measurement data, provided that no quantum errors occur \emph{after} the measurement projection. Formally, given $\mathcal{M}(\cdot)$ that obeys Eq.~(\ref{eqn: stochastic Pauli channel})--(\ref{eqn: stochastic channel prop}), all Pauli fidelities are learnable if:
\begin{equation}
    p_{01} = p_{11} = 0.
\end{equation}
Classical readout assignment errors $q_{\vec{r}}$ do not affect learnability.
\end{theorem}
\begin{proof}
    See Section~\ref{subpendix: learnable dof} in Appendix~\ref{appendix: all proofs}.
\end{proof}
The intuition behind Theorem~\ref{theorem: exact} is as follows. State preparation errors are the fundamental obstacle to unbiased noise tomography: when the quantum state entering the measurement projection is contaminated, the pre- and post-projection noise parameters become entangled in the measurement record and cannot be individually resolved. The condition $p_{01} = p_{11} = 0$ eliminates all post-projection quantum errors, decoupling the noise parameters and making each Pauli fidelity independently accessible. The independence from the number of unmeasured qubits $m$ follows directly: those qubits do not participate in the measurement event and therefore do not contribute to the SPAM coupling.

A particularly striking consequence is that classical readout assignment errors $q_{\vec{r}}$ play no role in the learnability condition. A classical readout error is a mislabelling of the measurement outcome --- it does not disturb the quantum state and therefore does not create the noise coupling that obstructs tomography. Theorem~\ref{theorem: exact} is therefore more precisely stated as: all measurement noise is learnable provided the residual state preparation imperfection is of classical readout type. This condition is naturally satisfied in trapped-ion platforms, where fluorescence-based state detection is highly accurate and residual errors are dominated by classical photon-counting noise rather than quantum back-action~\cite{PhysRevLett.100.200502}. When the condition holds, the \emph{learnable} noise triplets are:
\begin{equation}
    \{p_{\vec{a}, \vec{0}},\; \Lambda_{\vec{a}, \vec{0}},\; q_{\vec{r}}\}.
\end{equation}
The extraction procedure is given in the algorithms of Section~\ref{subsection: subroutine 1}, and a numerical validation that extracts MCM and CNOT gate noise is shown in Figure~\ref{fig:noise_characterization_top}.

Meanwhile, different devices suffer from different hardware noise constraints. In superconducting platforms, the read-out of the $\ket{1}$ state suffers from leakage, energy relaxation, and measurement back-action~\cite{Bengtsson2024willowmcm}. However, in this special case, the measurement noise and the noise of a two-qubit Clifford gate can be approximated with a noisy USI:
\begin{theorem}[Approximate Noise Extraction Principle]
\label{theorem: approx}
With high-fidelity preparation of $\ket{0}$ alone, the Pauli fidelities of CNOT gate can still be approximated. Formally, given a USI that obeys Eq.~(\ref{eqn: stochastic Pauli channel})--(\ref{eqn: stochastic channel prop}) with $p_{01} \approx 0$, the Pauli fidelity pairs of CNOT gate are approximable via post-selection on the $\ket{0}$ outcome.
\end{theorem}
\begin{proof}
    See Appendix~\ref{appendix: scheme 2 proof}.
\end{proof}
In Appendix~\ref{appendix: scheme 2 proof}, we provide a detailed account of the implementation, but let us outline the procedure for Pauli noise approximation of a two-qubit Clifford gate when we have access only to a clean $\ket{0}$ state in noisy MCMs. This idea relies on the fact that when we post-select the MCMs read-outs to $\ket{0}$, given the unlearnable pairs $\{\lambda_{0i},\lambda_{jk}\}$, besides the geometric mean in Eq.~(\ref{eqn: dilemma}) between the pair, we also have access to the approximate value of the arithmetic mean $\frac{\lambda_{0i}+\lambda_{jk}}{2}$, since measurement projection creates equivalence between the coupling Paulis $P_{0i}$ and $P_{jk}$ defined around Eq.~(\ref{eqn: dilemma}). Therefore, MCMs provide an additional degree of freedom from which we may mathematically isolate $\lambda_{0i}$ and $\lambda_{jk}$ from the geometric and arithmetic means. Figure~\ref{fig:scheme2} provides a validation of the special case noise learning on the CNOT gate. This approximation is useful when state preparation by projection produces a higher-fidelity $\ket{0}$ state than state preparation by cooling.

\subsubsection{Quantum Instruments with Classical Feed-Forward}

In addition to generalizing the Pauli noise learning theory of USI in Refs.~\cite{Chen_2023generalized, Hines_2025}, we take a step further to employ USI as a noise tomographic tool. In QEC and quantum algorithms, noisy MCMs are active components of quantum circuits that contribute to error correction and quantum computation through classical feed-forward~\cite{Quek2024-limitationqec, lemelin2025mcmprimitive}. Here, we will show that feed-forward also plays a crucial role in breaking gauge noise in generalized CB. To establish this point, we introduce QICF:

\begin{align}
    \text{Ctrl-}M(\cdot) &= \sum_{\vec{r},\vec{k}\in \mathbb{Z}_2^n} q_{\vec{r}}\mathcal{C}_{\vec{k}, \vec{r}}(\cdot)\otimes\mathbf{e}_{\vec{k}+\vec{r}}\label{eqn: QICF with cerr}, \\
    \mathcal{C}_{\vec{k}, \vec{r}}(\cdot) &= \mathcal{M}_{\vec{k}}\left[\left(\mathcal{X}^{M(\vec{k}+\vec{r})}\otimes \mathcal{I}^{\otimes l} \right)(\cdot)\right]. \label{eqn: QICF}
\end{align}
A QICF inherits all basic noise compositions in Eq.~(\ref{eqn: USI with cr-err.}) and (\ref{eqn: standard USI}). In addition, $\mathcal{C}_{\vec{k}, \vec{r}}$ implements a classical feed-forward on the $m$ unmeasured qubits via an $m\times l$ linear transformation $M$. The matrix converts classical read-out instructions $\textbf{e}_{\vec{k}+\vec{r}}$ into a conditional bit-flip $\mathcal{X}^{M(\vec{k}+\vec{r})}$. $\mathcal{C}_{\vec{k}, \vec{r}}$ also inherits the classical read-out assignment error $q_{\vec{r}}$ through feed-forwards. Operationally, $\text{Ctrl-}M$ is a control gate operator from the $l$-registers to the $m$-registers. It acts like a Clifford gate and preserves some Clifford transformations. However, a measurement projection inevitably destroys some coherence. For Pauli noise propagation, this means that some noise information propagates through $\text{Ctrl-}M$ as if it is propagating through a Clifford gate, while the basis choice $\superproj{\vec{k}}$ destroys noise information with incompatible bases. We note that a QICF can generalize to non-linear Boolean logic (as opposed to $\text{Ctrl-}M$) such as the Toffoli gate, but Eq.~(\ref{eqn: QICF}) suffices for the scope of our investigation. To see how a QICF transforms Pauli operators, let us consider the stochastic CNOT gate for $l = 1$:
\begin{align}
    \text{CX}_{\text{sto}, p, q; \textbf{e}}(\cdot)&=\sum_{r, k\in \mathbb{Z}_2}q_r\mathcal{CX}_{k, r; p, q}(\cdot)\otimes\mathbf{e}_{k+r}, \\
    \mathcal{CX}_{k, r; p, q}(\cdot)&= \mathcal{M}_{k}\left[\left(\mathcal{X}^{k+r}_q\otimes \mathcal{I}_p\right)(\cdot)\right]
    \label{eqn: stochastic CNOT}.
\end{align}
We remove the vector sign for $k$ and $r$ as binary variables. The term $\mathcal{X}^{k+r}_q$ emphasizes that the feed-forwards act on qubit $q$ and $\mathcal{I}_p$ is where the measurement projection occurs. $\textbf{e}$ denotes the classical read-out bits and $p, q$ denote the control and target qubits.
\begin{theorem}[Pauli Propagation for Stochastic CNOT]
\label{theorem: 1}
A stochastic CNOT, formed by tracing out the classical register of the QICF, replicates the Pauli transformation of a coherent CNOT gate on $Z$-type observables, up to a scalar pre-factor from the classical readout error. Formally, tracing out the classical register gives:
\begin{equation}
    \text{CX}_{\text{sto}, p, q}(\cdot) = \sum_{r, k\in \mathbb{Z}_2}q_r\,\mathcal{CX}_{k, r; p,q}(\cdot).
\end{equation}
For $c \in \mathbb{Z}_2$, the stochastic CNOT preserves the Pauli transformation structure of a coherent CNOT in the $z$-basis:
\begin{align}
    \text{CX}_{\text{sto}, p, q}(Z_p^cZ_q) &= (1-2q_{1})\,\omega_{Z_p^{c+1}Z_q}\,Z_p^{c+1}Z_q,
    \label{eqn: ZZ propagation}\\
    \text{CX}_{\text{sto}, p, q}(\mathcal{X}_p\mathcal{X}^c_q(\cdot)) &= \mathcal{X}_p\mathcal{X}^{c+1}_q\bigl(\text{CX}_{\text{sto}, p, q}(\cdot)\bigr),
\end{align}
where the classical readout error $(1-2q_1)$ is factored out independently of the Pauli structure, and:
\begin{align}
    &\omega_{Z_p^{c+1}Z_q} = p_{00}\lambda_{00}+(-1)^cp_{10}\lambda_{10}
    -\bigl((-1)^{c}p_{01}\lambda_{01}+p_{11}\lambda_{11}\bigr), \nonumber\\
    &\Lambda_{ab}(Z^c_pZ_q) = \lambda_{ab}Z^c_pZ_q. \nonumber
\end{align}
\end{theorem}
\begin{proof}
    See Section~\ref{subpendix: pauli propagation} in Appendix~\ref{appendix: all proofs}
\end{proof}
Analogously to Eq.~(\ref{eqn: dilemma}), these rules match those of a coherent CNOT with a Pauli noise channel. The classical read-out error $(1-2q_1)$ is treated separately, as its error propagation is independent of $Z^c_pZ_q$. Importantly, CTRL-$M$ can be implemented as sequential operations of stochastic CNOTs on various qubit pairs. In other words, by performing a few noisy MCMs, we can implement a variety of highly non-local Clifford gates that execute the desired Pauli propagation rules. Furthermore, the costs of the noise contributions are strictly limited to these MCMs, because all classical feed-forward can be deferred to classical post-processing.

\subsubsection{Noise Learning with MCMs}

Active choices of measured qubits and classical feed-forward can reverse the transformation induced by a generic Clifford gate. It turns out that we can derive a stronger result due to a form of the deferred feed-forward principle:
\begin{lemma}[Deferred Feed-Forward Principle for Generalized CB]
\label{lemma: generalized post-processing}
In generalized CB interleaved with a Clifford gate and a QICF, no real-time quantum feed-forward is needed. All classical feed-forward operations that manipulate Pauli weights of the propagating Paulis can be deferred to classical post-processing without changing the measurement statistics.
\end{lemma}
\begin{proof}
    Let us give a proof sketch of this claim. Given QICF, its role in the interleaving quantum circuit is to manipulate the Pauli weight given a Pauli $P_\alpha$ that propagates through the quantum circuit to our favor. According to Eq.~(\ref{eqn: ZZ propagation}), the measured qubit is precisely where the Pauli weight change takes place given a stochastic CNOT. Let us introduce a perfect ancillary qubit, so that the Pauli operator is expanded to $Z_0\otimes P_\beta$. By performing the stochastic CNOT between the qubits of your choice and $Z_0$, and suitable operations of single-qubit Clifford gates, we may transform the Pauli weight to our satisfaction, including the reversal map:
    \begin{equation}
        \mathcal{R}_{\delta(P_\alpha, P_\beta)}(Z_0\otimes P_\beta) = \lambda Z_0\otimes P_\alpha. \label{eqn: reversal map}
    \end{equation}
    where $\delta(P_\alpha, P_\beta)$ is the Pauli weight mismatch that computes the locations where MCMs are inserted. When extended to multiple cycles, the ancillary qubit will receive a collection of Pauli-$X$ feed-forward operations $\prod_iX^i\ket{0}_0$. Since they commute with each other locally and are not entangled with the logical process of the CB, the final state of the ancillary qubit for each shot will be invariant if we defer the feed-forward to the final measurement process. Since the ancillary qubit does not entangle with the system, it may be replaced with a classical bit that stores the conditional bit-flip instruction. Statistically, we recover the desired sequence of Clifford-QICF operations.

    See Section~\ref{subpendix: deferred feed-forward principle} in Appendix~\ref{appendix: all proofs}
\end{proof}
In other words, to enforce a reversal of a Clifford transformation of a Pauli operator, we only need to consider a clever selection of MCMs on the quantum circuit. This result simplifies the framework introduced in Refs.~\cite{Chen_2023generalized, Hines_2025} to retrieve learnable information, as now we concentrate the post-processing information on one ancillary qubit. This result is independent of the number of cycles of a generalized CB. QICF unlocks a powerful method for implementing control gate logic, as arbitrary Clifford reversal, whether or not the transformation is non-local, can be implemented with single-qubit Clifford gates, MCMs, and post-processing.
Given this information, we wish to implement as few MCMs as possible. The
following result specifies the minimal number of state preparation and measurement
operations needed to reverse the action of a Clifford gate:

\begin{theorem}[State Preparation Cost from Pauli Weight-Pattern Mismatch]
\label{theorem: main result}
Let $\mathcal{G}$ be an $n$-qubit Clifford gate, and let $P_\alpha$, $P_\beta$
be a pair of Pauli operators satisfying $\mathcal{G}(P_\alpha) = P_\beta$.
Then the Pauli fidelity $\lambda_{P_\beta}$ is learnable provided that USIs
satisfying Theorem~\ref{theorem: exact} are inserted between successive applications
of $\mathcal{G}$ in the CB protocol at the qubit locations specified by:
\begin{equation}
    \delta(P_\alpha, P_\beta).
\end{equation}
The number of MCMs required equals the Hamming weight $|\delta(P_\alpha, P_\beta)|$,
which also quantifies the overhead of generalized CB relative to standard CB.
\end{theorem}
\begin{proof}
    This immediately follows from Lemma~\ref{lemma: generalized post-processing} by performing measurements at the locations where Pauli weights differ.

    See Section~\ref{subpendix: deferred feed-forward principle} in
    Appendix~\ref{appendix: all proofs}.
\end{proof}
Recall that Theorem~\ref{theorem: learnability} asserts learnability when $\delta(P_\alpha,P_\beta)= 0$. Theorem~\ref{theorem: main result} generalizes the learnability condition to an MCMs resource condition when $\delta(P_\alpha,P_\beta)\neq 0$. The Pauli weight mismatch $\delta(P_\alpha, P_\beta)$ serves a dual purpose: it specifies where MCMs must be inserted to reverse the action of the Clifford gate $\mathcal{G}$, and it measures the distance
from standard CB, quantifying the additional measurement resources required for
generalized CB. This has three implications:
\begin{itemize}
    \item \textbf{No full tomography is required.} The identifiability of
    unlearnable parameters depends only on how $G$ changes the Pauli weight
    pattern, rather than on global state preparation over all $n$ qubits as
    required by Ref.~\cite{Chen_2023learnability}.
    \item \textbf{No additional Clifford gates are needed.} Breaking Clifford
    cycles requires only MCMs, which affords partial SPAM-robustness even in
    multi-qubit circuits.
    \item \textbf{No post-selection is required.} All feed-forward corrections
    are deferred to classical post-processing.
\end{itemize}
Figure~\ref{fig:benchmarking_and_graph} summarizes the noise learning protocol of a generalized CB with injection of MCMs at the locations of $\delta(P_\alpha, P_\beta)$. Meanwhile, Figure~\ref{fig:pauli_transfer_plot} shows examples of two-qubit Clifford gates and how Theorem~\ref{theorem: main result} provides a systematic recipe
for the collapse of cut-space to cycle-space, by inserting MCMs at the qubit locations specified in Eq.~(\ref{eqn: mismatch})---shown as
circuit block (b) interleaved with block (a) in Figure~\ref{fig:benchmarking_and_graph}. In Section~\ref{subsection: delta function}, algorithmic procedures are introduced that implement the generalized CB protocol for arbitrary $\mathcal{G}$.

\subsection{Binomial Analysis}
\label{subsection: binomial analysis}

In the previous section, we characterized the cost of MCM state preparation and the
requirements on its quality in order to learn all Pauli fidelities for a general Clifford
gate without gauge errors. These results, like conventional Pauli noise characterization
methods, rest on the assumption of Markovianity. In real quantum devices, this assumption
holds only up to a certain accuracy threshold, beyond which the noise dynamics can no
longer be captured by a Markovian model. Ultimately, Pauli noise learning techniques will
encounter a fundamental noise floor due to non-Markovian dynamics. Noisy MCMs as tomographic primitives open new avenues for
diagnosing these effects.

In standard CB, the probability distribution sampled for a Pauli operator $P_\beta$ is expected to take the form $p(\vec{y}\in \mathbb{Z}_2^n) = p(y_1,y_2,\cdots,y_n)$. MCM-based CB extends to the intermediate readouts through the quantum circuit of Figure~\ref{fig:benchmarking_and_graph}:
\begin{equation}
    \prod_{i=1}^N\left[\prod_{j\in \delta (P_\alpha,
    P_\beta)}\mathcal{M}^{\mathrm{basis}}_{k_{j, i}}(\mathcal{G}(\rho))\right].
    \label{eqn: circuit configuration}
\end{equation}
where $\mathcal{M}^{\mathrm{basis}}_{k_{j, i}}$ is a USI model with single-qubit Clifford basis rotations that are designed to reverse $\mathcal{G}$. This introduces $N$ auxiliary MCMs and, hence, an additional sequence of read-outs $\vec{k}_{p} = [k_{p,0}, k_{p,1}, \ldots, k_{p,N}]$ to each qubit $p$.

Suppose that all USIs are post-processed according to Lemma~\ref{lemma: generalized post-processing}. The final probability distribution $p(\vec{y})$ then admits a decomposition $p(\vec{y};\vec{k}_p)$. This gives the following expression for the expectation value of $P_\beta$ in Figure~\ref{fig:benchmarking_and_graph}:
\begin{equation}
    \langle P_\beta\rangle = \sum_{\vec{y}\in \mathbb{Z}_2^n, \vec{k_p}\in
    \mathbb{Z}_2^N}p(\vec{y};\vec{k}_p)(-1)^{\mathrm{pt}(P_\beta)\cdot \vec{y}}.
    \label{eqn: extended probability}
\end{equation}
According to this summation, $\langle P_\beta\rangle$ can be expanded as a collection of functions of $\vec{k}_p$, through which time correlations can be detected.

\subsubsection{Flip Variable}

To test statistical independence between successive MCM readouts, we introduce the
\textit{flip variable}:
\begin{equation}
    f_{p,i} = k_{p, i-1}\oplus k_{p, i}.
    \label{eqn: the flip variable}
\end{equation}
$f_{p,i}$ takes value 1 when consecutive readouts in qubit $p$ disagree and 0 otherwise,
regardless of the absolute outcomes. Under Markovianity, each cycle of
$\Lambda\circ\mathcal{G}$ is statistically independent of its neighbors, so each
$f_{p,i}$ depends only on the local process in cycle $i$ --- the $\{f_{p,i}\}$ are
independent and identically distributed. Any deviation from this independence structure
signals non-Markovian time correlations between cycles. Computing the flip variable is equivalent to measuring in the MCMs read-out:
\begin{equation}
    \hat{f}_p = Z_{p-1}Z_p.
    \label{eqn: flip observable}
\end{equation}

For an important class of Clifford gates --- the control gates such as the CNOT and CZ gates --- it is possible to further derive a closed-form measurement statistic from the flip variables when Markovianity holds. Suppose that for a non-identity Pauli $P_p$ that constitutes the single-qubit Pauli of either of $\{P_\alpha, P_\beta\}$ at qubit $p$, corresponding to the MCMs that have not yet been post-processed according to Lemma~\ref{lemma: generalized post-processing}, the following identity holds:
\begin{equation}
    \Lambda\circ\mathcal{G}(P_p) =P_p.
    \label{label: invariance}
\end{equation}
Then, when $\Lambda\circ\mathcal{G}$ combines with the incomplete reversal channel of $\mathcal{M}^{\mathrm{basis}}_{k_{j, i}}$, they induce effective Clifford transformations between the pair of Pauli operators:
\begin{equation}
    \{P_{\gamma}\otimes I_{p}, P_\gamma\otimes
P_{p}\},
\end{equation}
with one Pauli weight difference.

The prerequisite for the detection of Markovianity is a sufficiently clean MCMs resource. However, a typical USI model contaminates the statistical independence of $\Lambda\circ\mathcal{G}$, because the read-out also contains information that couples the pre-collapse and post-collapse errors. Therefore, the correct quantum instrument for detecting Markovianity requires a restriction to the \textit{separable USI} model, in which the pre-measurement and post-measurement effects decouple:
\begin{align}
    \Lambda_{ab} &= \Lambda_a\Lambda^\prime_b, \label{eqn: separable 1}\\
    p_{ab} &= p_ap^\prime_b. \label{eqn: separable 2}
\end{align}
In an experimental setting, this restriction works in our favor since correlations in USI read-outs are now rendered detectable.

\subsubsection{Binomial Distribution}
Denoting the
corresponding Pauli fidelities of $\{P_{\gamma}\otimes I_{p}, P_\gamma\otimes
P_{p}\}$ as $\{\lambda_{P_\gamma\otimes I_p},
\lambda_{P_\gamma\otimes P_p}\}$, and defining $f = \sum_i f_{p,i}$ as the total number
of flip events, it can be shown (Appendix~\ref{appendix: Read-out error derivation}) that
the expectation value of $P_{\gamma}$ expands as an unnormalized binomial distribution
in $f$:
\begin{align}
    &\lambda_{P_\gamma}(f) = \nonumber\\
    &\binom{N}{f}\left(\frac{\lambda_{P_\gamma\otimes I_p}+\lambda_{P_\gamma\otimes
    P_p}}{2}\right)^{N-f}\left(\frac{\lambda_{P_\gamma\otimes
    I_p}-\lambda_{P_\gamma\otimes P_p}}{2}\right)^{f},
    \label{eqn: exp of binomial distribution}
\end{align}
through which $\lambda_{P_{\gamma}\otimes I_p}^N$ and $\lambda_{P_{\gamma}\otimes P_p}^N$ are computed. In the special case where we consider only the statistics of $f$ for a single-qubit MCM
sequence, that is, $P_\gamma = I^{\otimes(n-1)}$, the flip distribution reduces to a
binomial:
\begin{align}
    p(f) &= \binom{N}{f}\left(\frac{1+\lambda_{P_p}}{2}\right)^{N-f}
    \left(\frac{1-\lambda_{P_p}}{2}\right)^f,
    \label{eqn: binom}
\end{align}
through which $\lambda_{P_p}^N$ is computed. This is the key diagnostic result of this section: under a Markovian separable USI model, $p(f)$ follows a clean binomial distribution. Any deviation from this
binomial signature therefore constitutes direct evidence of non-Markovian dynamics.
Different types of non-Markovian noise produce characteristically distinct
deviations, making the flip distribution a sensitive and interpretable diagnostic tool.
For example, when the separable USI assumption is relaxed to the general definition of USI, the binomial distribution becomes asymmetric, favoring even flip counts over odd ones, as illustrated in Figure~\ref{fig:asymmetry}. In Section~\ref{section: Experiment}, the closed-form expression of the binomial analysis provides the analytical tool to probe the non-Markovian behavior of entangling gates.

\begin{figure}
    \centering
    \includegraphics[width=\linewidth]{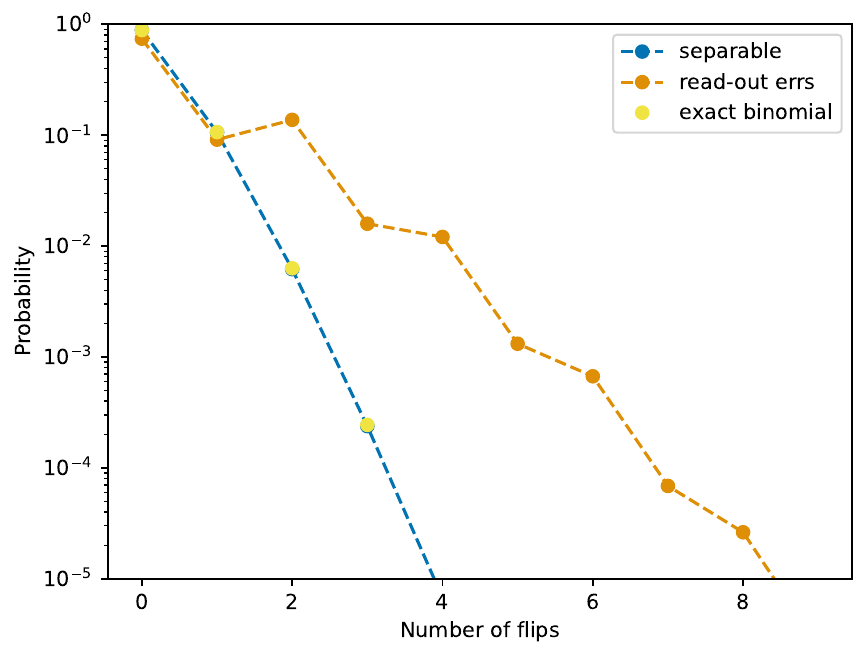}
    \caption{\justifying Deviation from binomial distribution in the presence of classical read-out assignment error. The separable (blue) USI is simulated with pre-measurement error channel and post-measurement error probability $p_1 = p_1^\prime =0.001$ from Eq.~(\ref{eqn: separable 2}), corresponding to the value $\lambda_{P_p} = 0.996004$. The theoretical model (yellow) is plotted to verify the consistency of the simulation. The presence of classical read-out assignment error $p_{r1}= 0.003$ (orange) statistically induces decay from the $f \in \{0, 1\}$ statistics to $f \in \{2, 3\}, \{4, 5\}, \ldots$. Therefore, a zig-zag linear-like statistical deviation emerges.}
    \label{fig:asymmetry}
\end{figure}

\section{Algorithms}

\label{section: algorithm}
In the previous section, we gave the formulation of generalized CB in the presence of noisy MCMs. In this section, we provide procedures that connect all of the above theorems into a generalized noise characterization protocol. In particular, the noise learning method, as shown in Figure~\ref{fig:benchmarking_and_graph}, is divided into learning the noise of the MCMs (b) and the noise of the Clifford-MCMs circuit (a) + (b) in the benchmarking protocol. In both cases, one classical ancillary bit represented as double lines is introduced to absorb a conditional post-processing. Post-processing converts a USI into a QICF noise model in the quantum circuit, through which we may manipulate the Pauli weight, a feature that is inaccessible to conventional CB~\cite{Chen_2023learnability}. Eventually, we may decouple the MCMs noise contributions from the Clifford gate's Pauli fidelity, and extract the noise profile of a given Clifford gate. In the following, we give the step-by-step implementation of the CB protocol with MCMs, and discuss the compatibility of the protocol with the landscape of Pauli noise learning.

\subsection{Noise Learning of MCMs}
\label{subsection: subroutine 1}

Before extracting the Clifford gate noise, a procedure is given to extract MCMs noise provided that the state preparation quality of MCMs is guaranteed (Theorem~\ref{theorem: exact}). At this stage, the quantum circuit from which we wish to sample consists of purely MCMs at the qubit locations $\delta(P_\alpha, P_\beta)$, illustrated as CB of (b) in Figure~\ref{fig:benchmarking_and_graph}. Naively, extracting the same information through an independent implementation of classical feed-forward will take $|\delta(P_\alpha, P_\beta)|+1$ distinct quantum circuits to process all quantum noise from MCMs. However, Lemma~\ref{lemma: generalized post-processing} shows that the same quantum information is encoded in CB of (b) in Figure~\ref{fig:benchmarking_and_graph}. In one characterization circuit, it is possible to learn all of the relevant noise information that involves the noisy MCMs: $P_\beta$, $\delta_{P_\beta}$, and the classical read-out error $\lambda_{q_p}$ at the measured qubit $p$.

Given the sampled data from the CB of (b), the extraction of $\delta_{P_\beta}$ is straightforward as it involves the preparation and measurement of $\langle P_\beta\rangle$. This step is tractable because USIs behave like a Pauli channel in the $z$-basis (Theorem~\ref{theorem: USI pauli propagation.}). Meanwhile, to extract the classical read-out errors, an iterative implementation of classical post-processing is required for each qubit location $p$ in the non-zero locations of $\delta(P_\alpha,P_\beta)$. Specifically, we incorporate the classical feed-forward $\left(\mathcal{X}^{k+r}_q\otimes \mathcal{I}_p\right)(\cdot)$ of Eq.~(\ref{eqn: QICF}) by post-processing to convert each USI into a QICF. The feed-forward effectively introduces a stochastic CNOT gate, which couples $P_\beta$ with an adjacent $P_{\gamma_p}$ whose Pauli weight differs from $P_\beta$ at $p$.

As a result of generalizing Eq.~(\ref{eqn: stochastic Pauli channel}) and~(\ref{eqn: ZZ propagation}), the fidelity extracted with feed-forward becomes $(1-2q_{p})\sqrt{\delta_{P_{\gamma_p}}\delta_{P_\beta}}$, where $q_{p}$ denotes the single-qubit classical readout assignment error rate of qubit $p$.

Combining both steps, we perform exponential fitting over the following fidelity decay curves for each qubit location $p$:
\begin{align}
    f_{\mathrm{USI}}(d) &= A\delta_{P_{\gamma_p}/P_\beta}^d+B,\label{eqn: CB USI}\\
    f_{\mathrm{QICF}}(d) &= A \left((1-2q_{p})^2\delta_{P_{\gamma_p}}\delta_{P_\beta}\right)^{d/2}+B, \label{eqn: CB QICF}
\end{align}
from which the classical and quantum noise parameters $\{\delta_{P_{\gamma_p}}, \delta_{P_{\beta}}, q_p\}$ for each $p$ can be jointly extracted. This two-step procedure separates the classical readout assignment error $q_p$ from the quantum Pauli fidelities $\delta_{P_\beta}$, which is the central operational result of this framework.

\subsection{Noise Learning of Multi-Qubit Clifford Gates}
\label{subsection: delta function}
Given the pre-measurement error channel and the collections of classical read-out errors for each MCM, we may now leverage Theorem~\ref{theorem: main result}, and the noise extraction procedure becomes relatively straightforward. Here, $P_\beta$ is characterized through the (a) + (b) circuit in Figure~\ref{fig:benchmarking_and_graph}. The independent components of (a) and (b) contribute to the following fidelities:
\begin{align}
    (\Lambda\circ\mathcal{ G})(P_\alpha) &= \lambda_{P_{\beta}}P_{\beta},\\
    R_{\delta(P_\alpha, P_\beta)}(P_{\beta})&=\prod_{p \in \delta(P_\alpha,P_\beta)}(1-2q_p)\delta_{P_{\beta}}P_{\alpha},
\end{align}
where the product of $(1-2q_p)$ emerges due to the fact that all MCMs in $R_{\delta(P_\alpha, P_\beta)}$ are now turned into QICFs acting on the ancillary classical bit. When combined, an (a) + (b) circuit results in the following decay curve:
\begin{equation}
    f_{\mathrm{GCB}}(d) = A\left[\prod_{p \in \delta(P_\alpha,P_\beta)}(1-2q_p)\delta_{P_{\beta}}\lambda_{P_{\beta}}\right]^d + B. \label{eqn: GCB}
\end{equation}
\subsection{Consistent Noise Learning of Multi-Qubit Clifford Gates}

So far, we have given two primitives that decouple unlearnable quantum noise provided that single-qubit MCMs exhibit sufficiently good $z$-basis projection quality, comparable to single-qubit gate fidelities. Depending on the type of quantum device, there exist regimes in which $z$-basis projections are no longer guaranteed. In this situation, we cannot rely on the assumption of Theorem~\ref{theorem: exact}. However, the two primitives in fact remain compatible with the consistent Pauli noise learning framework. Consequently, this significantly extends the applicability of CB to arbitrary Clifford circuits.

The consistent Pauli noise learning framework leverages gate set tomography to unify the unlearnable noise gauge of a gate set with Pauli noise, providing a self-consistent characterization of noise across a circuit~\cite{chen2025disambiguatingpaulinoisequantum}. A key insight of our work is that an MCM with classical feed-forward effectively acts as a Clifford gate, inheriting the same noise structure and gauge properties that Ref.~\cite{chen2025disambiguatingpaulinoisequantum} is designed to handle. As a consequence, MCMs can be incorporated into the consistent Pauli noise learning framework without breaking its theoretical guarantees. Specifically, the Pauli propagations in Eq.~(\ref{eqn: CB USI}), (\ref{eqn: CB QICF}), and (\ref{eqn: GCB}) can be embedded within the design matrix formalism of Ref.~\cite{chen2025disambiguatingpaulinoisequantum}. This result enables consistent CB of arbitrary Clifford circuits, provided that the noise is characterized consistently using single-qubit MCMs. Adopting the USI gate assumption, our framework effectively absorbs SPAM errors into a consistently learnable Pauli noise framework that leverages the noise magnification of CB. By establishing this compatibility, we expand the reach of CB beyond purely unitary Clifford gate sets to include the broader class of operations available in modern quantum processors that support mid-circuit measurement and real-time classical feed-forward.

\section{Numerics}
\label{section: numerics}

\begin{figure*}[t]
    \centering
    \begin{minipage}{0.45\linewidth}
        \centering
        \includegraphics[width=\linewidth, height=0.77\linewidth]{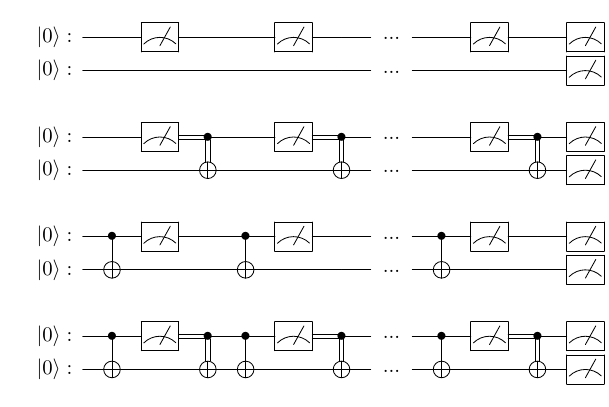}
        \label{fig:philosophy}
    \end{minipage}
    \hfill
    \begin{minipage}{0.5\linewidth}
        \centering
        \includegraphics[width=\linewidth, height=0.7\linewidth]{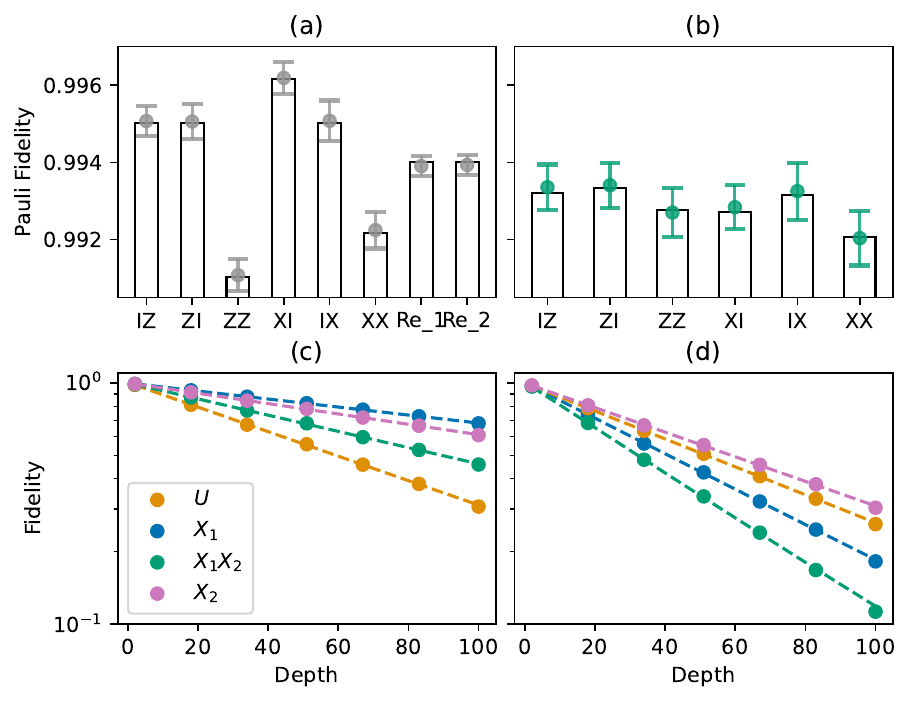}
        \label{fig:benchmark}
    \end{minipage}

    \caption{\justifying Unlearnable noise characterization of noisy CNOT gates using mid-circuit measurements. Characterization with $R = 1000, S=1000$ per data point in (c) and (d) using Markovian noisy USI satisfying Theorem~\ref{theorem: exact}. We simulate the extraction of unlearnable pairs $Z_2, Z_1Z_2$ and $X_1, X_1X_2$ of a noisy CNOT gate with the \texttt{fake-torino} noise model. \textbf{Left}: The four possible and effective MCMs and CNOT-MCMs circuits that one may implement through post-processing via Lemma~\ref{lemma: generalized post-processing}. These diagrams serve to show that MCMs can cancel the actions of coherent CNOT gates, thus breaking the unlearnable couplings. \textbf{Right}: CNOT gate noise characterization demonstration. In this characterization round, we perform unlearnable noise learning in the $z$/$x$-basis, using two types of MCMs. \texttt{Re$_{1/2}$} denotes the classical read-out errors extracted the MCMs used for the $z/x$-basis characterisation, while $U$ denotes the unlearnable noise pair extracted from regular CB. The characterization begins with MCMs noise learning with result displayed in panel (a). The learning process requires CB on MCMs that produces exponential decay of fidelity displayed in panel (c). Next, the unlearnable noise for the CNOT gate is extracted, displayed in panel (b). The characterisation of interleaving CNOT-MCMs quantum circuits produces the signature exponential decay of fidelity in (d).}
    \label{fig:noise_characterization_top}
\end{figure*}

\begin{figure*}[t]
    \centering
    \includegraphics[width=\linewidth, height=0.3\linewidth]{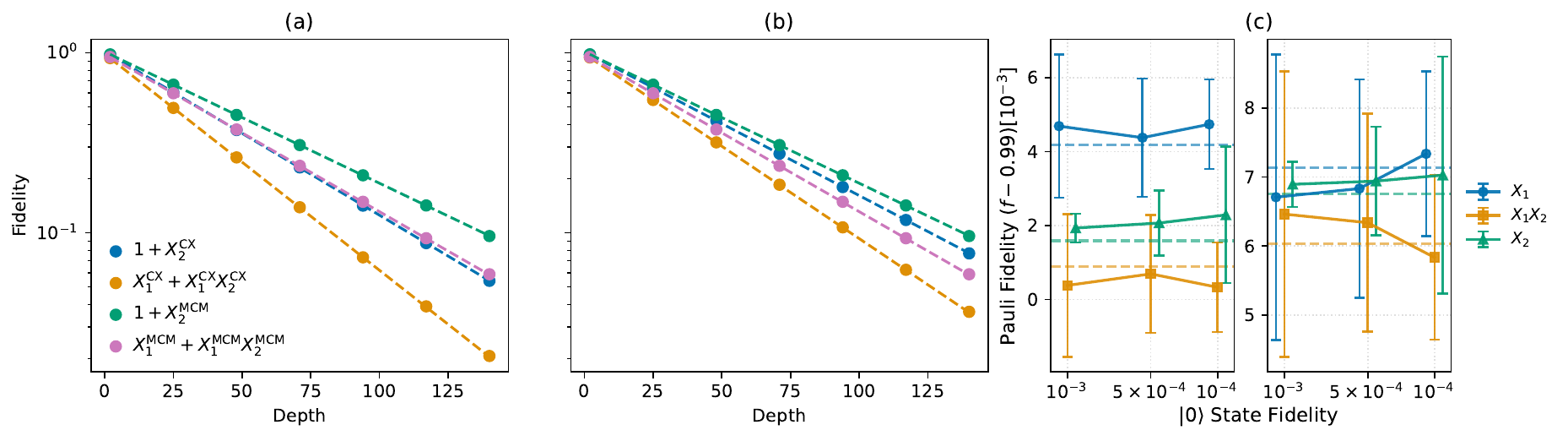}
    \caption{\justifying Unlearnable noise approximation of noisy CNOT gates using MCMs with unreliable $\ket{1}$ preparation. Characterization with $R=1000, S = 10000$ per data point in (a) and (b) using Markovian noisy MCMs satisfying Theorem~\ref{theorem: approx} in the $x$-basis, on qubit pairs $[20,21]$ and $[40,34]$ of \texttt{fake-torino}. (a) and (b) illustrate CB with $\ket{0}$ state projections performed on interleaving CNOT and MCMs on $[20, 21]$ and $[40, 34]$, respectively. In these examples, a single-qubit $y$-rotation $\lambda_r = \cos 14^\circ$ is performed prior to each MCM to implement Eq.~(\ref{eqn: arithmetic and geometric mean}). (c) summarizes the quality of Pauli noise extraction of the unlearnable pairs $X_1, X_1X_2$ (blue, orange) and the learnable coefficient $X_2$ (green) of $[20,21]$ (left) and $[40,34]$ (right), when the state preparation infidelity decreases. Dashed lines indicate exact solutions. When the unlearnable noise difference is pronounced, the left panel of (c) shows that noise fidelity can be effectively separated with the approximate technique, despite of the state-preparation fidelity by each MCM. Otherwise, the right panel of (c) shows that better state-preparation quality can improve on the accuracy of unlearnable noise detection. It is possible to tailor the single-qubit $y$-rotation that maximizes the difference of fidelities of Eq.~(\ref{eqn: arithmetic and geometric mean}) when performing the approximation, according to the $\ket{0}$ state fidelity. They are chosen to be, from left to right, $\cos 14^\circ, \cos 18^\circ, \cos 22^\circ$. }
    \label{fig:scheme2}
\end{figure*}

\begin{figure*}[t]
    \centering
    \includegraphics[width=1\linewidth]{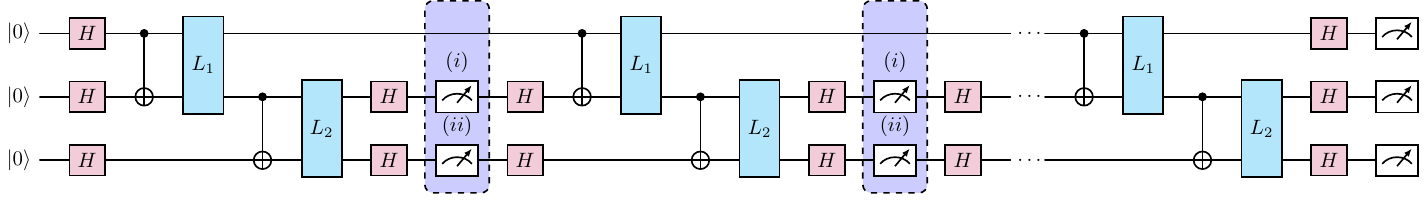}

    \vspace{1em}
    \includegraphics[width=1\linewidth]{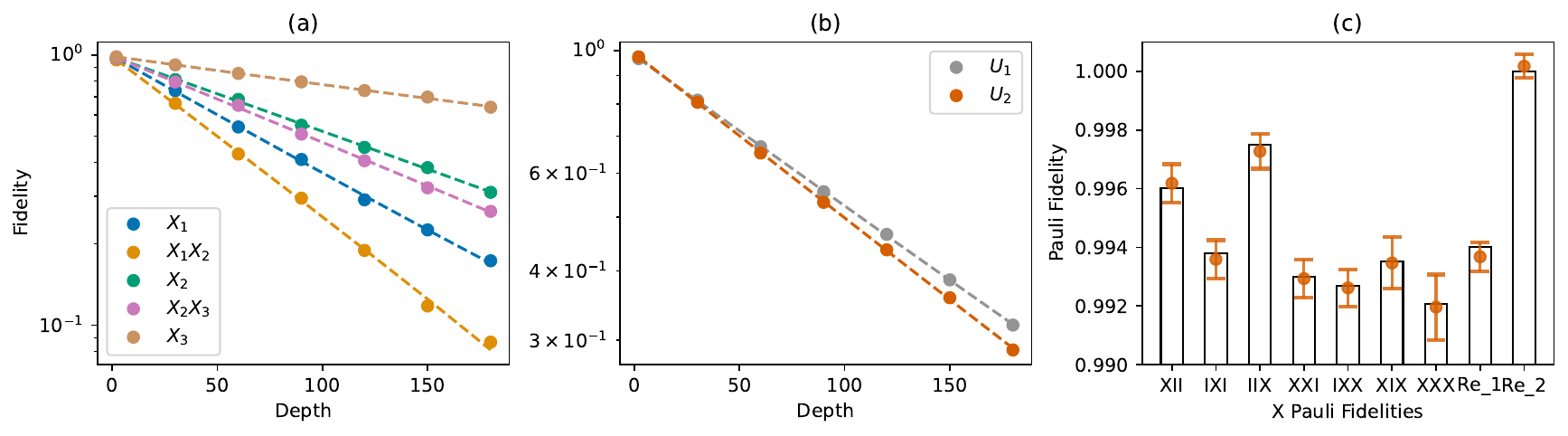}
    \caption{\justifying Three-qubit ladder CNOT gate coefficient characterization with $R = 250, S = 100$ per data point in (a) and (b) using Markovian noisy MCMs satisfying Theorem~\ref{theorem: exact}. CNOT gates are simulated with the noise model of \texttt{fake-torino}. \textbf{Top}: Circuit diagram for the characterization of the CNOT ladder in the $x$-basis in the presence of a local noise model. When noise locality is not guaranteed, Theorem~\ref{theorem: main result} asserts that MCMs must be inserted at (i) and (ii) simultaneously to learn all learnable Pauli fidelities. However, knowledge of the local noise structure $L_1$ and $L_2$ allows a separate implementation of (i) and (ii), reducing the MCMs overhead for noise tomography.
    \textbf{Bottom}: (a) demonstrates the information extracted from post-processing of interleaving CNOT-ladder and MCMs measurements; (b) demonstrates the coupled unlearnable Pauli fidelities $U_1$ and $U_2$ with standard CB when MCMs are not available. $U_1$ denotes the unlearnable pairs of Eq.~(\ref{eqn: cycle 2}) while $U_2$ denotes the unlearnable quadruple of Eq.~(\ref{eqn: cycle 1}); (c) summarizes the extracted noise information that is originally coupled due to Pauli cycles. Classical read-out errors are additionally demonstrated.}
    \label{fig:ladder-cnot}
\end{figure*}

We numerically estimate Pauli fidelities using the protocols of Theorems~\ref{theorem: exact} and~\ref{theorem: approx}. Subsection~\ref{subsection: CNOT gate} benchmarks the CNOT gate to demonstrate robustness against MCMs noise, and Subsection~\ref{subsection: CNOT ladder} extends to a 3-qubit CNOT ladder to demonstrate scalability. Within Subsection~\ref{subsection: CNOT gate}, we further separate \textit{exact} learning, which assumes the USI model of Theorem~\ref{theorem: exact}, from \textit{approximate} learning, which relaxes the USI model to allow amplitude damping of the $\ket{1}$ state and a noisy $\ket{0}$ preparation with infidelity $F(\ket{0})=1-\bra{0}\rho\ket{0}$. The CNOT gate noise model is taken from the \texttt{fake-torino} calibration snapshot dated \texttt{2024-02-28}. For all simulations, we implement the state preparation error as a depolarizing error:
\begin{equation}
    E(\rho) = (1 - \lambda) \rho + \lambda \text{Tr}[\rho] \frac{I}{2^n}.
\end{equation}
with $\lambda =0.005$.

Meanwhile, for both this section and Section~\ref{section: Experiment}, $R$ denotes the number of randomizations and $S$ denotes the number of shots per twirl across the plots, so that the total number of shots per cycle is $R\times S$. If we denote the total number of cycles as $C$, the total number of shots in the experiment corresponds to $C\times R\times S$. $[a,b]$ is used to specify $a$ as the control and $b$ as the target of a CNOT or a CZ gate.

\subsection{CNOT Gate}
\label{subsection: CNOT gate}
\subsubsection{Exact Learning}

In this benchmark, the unlearnable coefficients of the CNOT gate are extracted through the use of USIs that satisfy the condition of Theorem~\ref{theorem: exact}, summarized in Figure~\ref{fig:noise_characterization_top}. Specifically, the task is to disentangle the unlearnable noise pairs $\{\lambda_{x0},\lambda_{xx}\}$ and $\{\lambda_{0z}, \lambda_{zz}\}$, where both are only learnable as a product in standard CB. The learning process follows the procedure in Section~\ref{section: algorithm}. The noise of MCMs is first characterized, followed by the noise learning of the CNOT gate.

To emulate a realistic simulation for quantum gates, we prepare a noisy CNOT gate with a noise model taken from qubit $[0, 1]$ of \texttt{fake-torino}. The USI noise model consists of consecutive amplitude damping errors $T_1 = 50\,\mu\mathrm{s}$, $T_2 = 70\,\mu\mathrm{s}$, and gate time duration $t = 200\,\mathrm{ns}$, and a depolarizing error of $\lambda = 0.001$. Meanwhile, we assume that measurement errors are contributed by asymmetric classical read-out errors $\epsilon_{0} =0.001, \epsilon_1 = 0.005$, where $\epsilon_0$ denotes the read-out error event $\ket{0}\rightarrow \ket{1}$ and vice versa. Pauli twirling is applied independently to both the CNOT gate and the noisy MCMs. The twirling process converts both the CNOT and MCMs noise to a Pauli channel, and symmetrizes the classical read-out error to $\frac{\epsilon_0+\epsilon_1}{2} = 0.003$. Additionally, the classical read-out errors also contribute to the SPAM error in the final measurements.

The numerical results of Figure~\ref{fig:noise_characterization_top} can be interpreted as a two-step noise processing procedure as in Section~\ref{section: algorithm}. The left-hand side of the figure set illustrates the types of quantum circuit that we can implement in post-processing through the deferred feed-forward principle, and the right-hand side demonstrates the fitted results. The learnable MCMs noise parameters, in the $x$- and $z$-basis, are shown in panel (a), with the exponential fitting demonstrated in (c). They are obtained by implementing CB of the top two quantum circuits on the left-hand side through the deferred classical feed-forward principle (Lemma~\ref{lemma: generalized post-processing}). With the MCMs noise parameters, the unlearnable noise pairs $\{\lambda_{x0},\lambda_{xx}\}$ and $\{\lambda_{0z}, \lambda_{zz}\}$ become learnable in panel (b) with fittings in (d), through implementing CB of the bottom two quantum circuits on the left-hand side.

\subsubsection{Approximate Learning}

In approximate noise learning, the problem of interest is whether noisy MCMs may still be a reliable state preparation resource for approximating unlearnable noise pairs, when the state-projection quality of the $\ket{1}$ state is no longer guaranteed. With this premise, the exact learning protocol that we have been carefully developing is no longer applicable for isolating the measurement errors from the quantum gate errors. Nonetheless, the availability of the $\ket{0}$ state in the MCMs resource still provides a window for isolating unlearnable noise. Specifically, let us consider the propagation of the coupling Pauli pairs coupled by a CNOT gate:
\begin{equation}
    \frac{\lambda_{i0}P_{i0}+\lambda_{ij}P_{ij}}{2},
\end{equation}
with the corresponding unlearnable fidelities $\{\lambda_{i0}, \lambda_{ij}\}$. A projection on the second qubit in the $j$-basis alters the shape of the fidelity:
\begin{equation}
    \frac{(\lambda_{i0}+\lambda_{ij})(P_{i0}+P_{ij})}{2}.
\end{equation}
When CB is performed over this process and $P_i$ is measured on the first qubit, we obtain the fidelity estimate:
\begin{equation}
    f_{\text{approx.}}(d) = A\left(\frac{\lambda_{i0}+\lambda_{ij}}{2}\right)^d+B.
\end{equation}
When combined with the geometric mean learned from conventional CB, we obtain the system of equations:
\begin{equation}
    \lambda_{\text{am}} = \frac{\lambda_{i0}+\lambda_{ij}}{2}, \quad \lambda_{\text{gm}} = \sqrt{\lambda_{i0}\lambda_{ij}},
\end{equation}
through which the unlearnable fidelity can be solved. In noise characterization, $\lambda_{i0}\approx\lambda_{ij}$ implies that $\lambda_{\text{am}}\approx\lambda_{\text{gm}}$, but this is overcome with the trick that a single-qubit rotation, known as a regulator in our work, can alter $\lambda_{ij}^\prime = \lambda_{ij}\times\lambda_{\text{reg}}$ without altering $\lambda_{i0}$. In principle, $\lambda_{i0}$ can be solved exactly through the system of equations:
\begin{equation}
    \lambda_{\text{am}} = \frac{\lambda_{i0}+\lambda^\prime_{ij}}{2}, \quad \lambda_{\text{gm}} = \sqrt{\lambda_{i0}\lambda_{ij}^\prime}.
    \label{eqn: arithmetic and geometric mean}
\end{equation}
When the $\ket{0}$ state projection process is accompanied by noise effects prior to each MCM measurement, we can instead approximate $\lambda_{\text{am}}$ (Appendix~\ref{appendix: scheme 2 proof}), hence the name approximate learning. In the following, we describe a noise environment that emulates harsher noise conditions compared to exact learning.

In this environment, the amplitude damping errors before measurement become $T_1 = 20\,\mu\mathrm{s}$ and $T_2 = 30\,\mu\mathrm{s}$ with gate time duration $t = 200\,\mathrm{ns}$, while the consecutive depolarizing error remains $\lambda = 0.001$. In addition, a single-qubit amplitude damping error with $T_1 = 90\,\mu\mathrm{s}$ and $T_2 = 90\,\mu\mathrm{s}$ and gate time duration $t = 500\,\mathrm{ns}$ is prepared to prevent the use of the $\ket{1}$ resource state. We implement the classical read-out error rates $\epsilon_0 = 0.001, \epsilon_1= 0.005$ in noisy MCMs, and also $\epsilon_0 = 0.005, \epsilon_1= 0.01$ on the final measurements to simulate SPAM errors. Pauli twirling is applied to the CNOT gates as before. However, since the $\ket{1}$ state is no longer a useful quantum resource in this characterization set, full Pauli twirling cannot be applied to MCMs, except for their non-measuring neighbor. Through this set-up, the approximate learning benchmark simulates an environment with available $\mathcal{M}_0$ of Eq.~(\ref{eqn: standard USI}) but not $\mathcal{M}_1$ for $l=1$.

Figure~\ref{fig:scheme2} illustrates whether the noise approximation protocol is sensitive to the noise difference of unlearnable noise parameters. For this purpose, two noisy CNOT gates $[20,21]$ and $[40, 34]$ of \texttt{fake-torino}, with stronger and weaker noise differences respectively, are tested. Since this simulation is not used to demonstrate noise learning in different bases, we limit our investigation to decoupling the unlearnable pairs $\{\lambda_{x0}, \lambda_{xx}\}$. Additionally, we also approximate the learnable fidelity $\lambda_{0x}$ as a control to test the reliability of the approximation. In the benchmark, three values of the $\ket{0}$ state preparation fidelity are simulated through the depolarizing channel $F(\ket{0}) \in \{10^{-3}, 5\times10^{-4}, 10^{-4}\}$, spanning the range of hardware quality achievable. Correspondingly, regulators that alter $\lambda_{ij}$: $\lambda_r=\cos 14^\circ, \cos 18^\circ, \cos 22^\circ$ are chosen for each fidelity value $F(\ket{0})$ in the region of highest estimation precision (see Figure~\ref{fig:regulator} in Appendix~\ref{appendix: scheme 2 proof}). The fittings of $[20, 21]$ are shown in panel (a), $[40,34]$ in panel (b), and the comparison between the approximate unlearnable/learnable fidelities and the exact solution is shown in panel (c).

The results confirm the in-principle viability of approximate unlearnable noise differentiation. Two practical limitations emerge: the shot efficiency is substantially lower than in the exact regime, and the approximation degrades with deterioration in the $\ket{0}$ state preparation fidelity, consistent with the theoretical mechanism of Theorem~\ref{theorem: approx}.

\subsection{Characterization with Local Noise Model}
\label{subsection: CNOT ladder}

Following the benchmarks of CNOT noise learning in the presence of different types of MCMs noise, we investigate the learnability of a $3$-qubit CNOT Clifford gate when the quantum gate noise is local, summarized in Figure~\ref{fig:ladder-cnot}. In this benchmark, the USI model of Theorem~\ref{theorem: exact} is assumed, and since MCMs noise under this condition can be isolated prior to characterization of an $n$-qubit Clifford gate, no noise model is employed for MCMs, except for classical read-out errors that take the values $\epsilon_0 = 0.001, \epsilon_1= 0.005$. Meanwhile, the noise model of the CNOT gate is taken from the gate errors of consecutive qubit pairs $[26,27]$ and $[27, 28]$ of \texttt{fake-torino}. For the SPAM error, a 3-qubit version of the SPAM error in the exact learning of the CNOT gate is prepared with exactly the same parameterization.

The purpose of this benchmark is to demonstrate a generalized application of Theorem~\ref{theorem: main result} to MCMs noise learning when the noise model of a Clifford gate is local. Provided that all Pauli fidelities are independent degrees of freedom, the Pauli propagation of the CNOT ladder gate obeys the following transition rules:
\begin{align}
    \mathrm{cycle}_1 &: X_1X_2\rightarrow X_1 \rightarrow X_1X_2X_3\rightarrow X_1X_3\label{eqn: cycle 1},\\
    \mathrm{cycle}_2 &: X_2X_3\rightarrow X_2\rightarrow X_2X_3\rightarrow X_2\label{eqn: cycle 2},\\
    \mathrm{cycle}_3 &: X_3\rightarrow X_3\rightarrow X_3\rightarrow X_3.
\end{align}
The maximum fidelity difference of the CNOT ladder is $\delta(X_1, X_1X_2X_3) = [0, 1, 1]$, and any other fidelity difference lies within the binary vector $[0, 1, 1]$, such as $\delta(X_1X_2, X_1)$. Without noise locality, the generalized CB for characterizing all learnable information in the $x$-basis requires MCMs on both the second and third qubits, corresponding to (i) and (ii) in Figure~\ref{fig:ladder-cnot}.

With gate noise locality, the accumulation of Pauli fidelities can be divided into pairs $(\alpha_{L_1}, \alpha_{L_2})$ that correspond to what the local noise models $L_1$ and $L_2$ observe.
\begin{align}
    X_1 &\rightarrow (X_1, I),\\
    X_3 &\rightarrow (I, X_3),\\
    X_1X_3 &\rightarrow (X_1, X_3),\\
    X_2 &\rightarrow (X_2, X_2),\\
    X_1X_2 &\rightarrow (X_1X_2, X_2).
\end{align}
The pairs are homomorphisms and generators of the Pauli group on the left. Importantly, $X_1X_2X_3$ can be mapped to $(X_1X_2,X_2X_3)$, so $\mathrm{cycle}_1$ of Eq.~(\ref{eqn: cycle 1}) can be further decomposed to:
\begin{equation}
    (X_1X_2,X_2)\rightarrow (X_1, I)\rightarrow (X_1X_2, X_2X_3)\rightarrow (X_1, X_3).
\end{equation}
This changes the landscape of noise learning because the real unlearnable contributions now center on the left side of the pair. The characterization of the CNOT ladder can be decomposed into three sets of characterization that measure at qubit $\delta(X_3,X_3)=[0,0,0]$, $\delta(X_2X_3, X_2) = \delta(X_2, X_2X_3)= [0, 0, 1]$, and $\delta(X_1X_2, X_1) = [0, 1, 0]$. Through Theorem~\ref{theorem: main result}, this asserts the learnability of $X_3, X_2, X_2X_3$, and $X_1$. When performing standard CB on Eq.~(\ref{eqn: cycle 1}), we also extract $X_1X_2$. In the circuit diagram in Figure~\ref{fig:ladder-cnot}, noise locality implies that (i) and (ii) can be implemented separately, thus reducing the overhead of MCMs used for characterization. In the data set of Figure~\ref{fig:ladder-cnot}, the numerical benchmarks in panel (a) capture all of the learnable Pauli fidelities required to construct all fidelities in panel (c). We also demonstrate the computation of the classical read-out errors. Meanwhile, panel (b) contrasts panel (a) with the learnable information in standard CB when MCMs are not available.

This benchmark demonstrates the compatibility of Theorem~\ref{theorem: main result} with the assumption of noise locality, which reduces the number of MCMs per characterization circuit for the CNOT ladder. Pauli fidelities are learned in the $x$-basis via the generalized post-processing scheme, and three types of characterization circuit, excluding MCMs noise characterization, are sufficient for the full extraction of $x$-basis fidelities and classical readout errors for the 3-qubit CNOT ladder benchmark. Under the realistic sparse Pauli--Lindblad noise model~\cite{temme_2023prob}, the generalized CB technique provides a new method for Pauli noise characterization of Clifford gates with limited use of MCMs resources.

\section{Experiment}
\label{section: Experiment}

We present two experimental studies where noisy MCMs are employed as tomography primitives in IBM devices based on the theoretical framework of generalized CB in Subsection~\ref{section: generalized cycle benchmarking} and binomial analysis in Subsection~\ref{subsection: binomial analysis}. In Subsection~\ref{section: unlearnable noise learning}, we measure the unlearnable noise of the CNOT gate through generalized CB, benchmarked against direct tomography of Pauli fidelity. In Subsection~\ref{subsection: experimental binomial analysis}, we investigate the dynamics of non-Markovianity of Pauli fidelities through binomial analysis. Specifically, the study investigates the binomial distribution of $P_p$ in Eq.~(\ref{label: invariance}) which can be expanded to the form of Eq.~(\ref{eqn: binom}). Both analyzes have been conducted on the CNOT gates of \texttt{ibm-aachen} and \texttt{ibm-pittsburgh}. The experiment on \texttt{ibm-aachen} with full twirling are conducted on \texttt{2025-12-28}. Meanwhile, experiments on \texttt{ibm-aachen} and \texttt{ibm-pittsburgh} with no twirling on the MCMs have been conducted on \texttt{2025-12-26}. MCMs were chosen with minimal real-time read-out assignment error rate given by latest calibration data. 

\subsection{Unlearnable Noise Learning}
\label{section: unlearnable noise learning}

\begin{figure*}
    \centering
    \includegraphics[width=\linewidth]{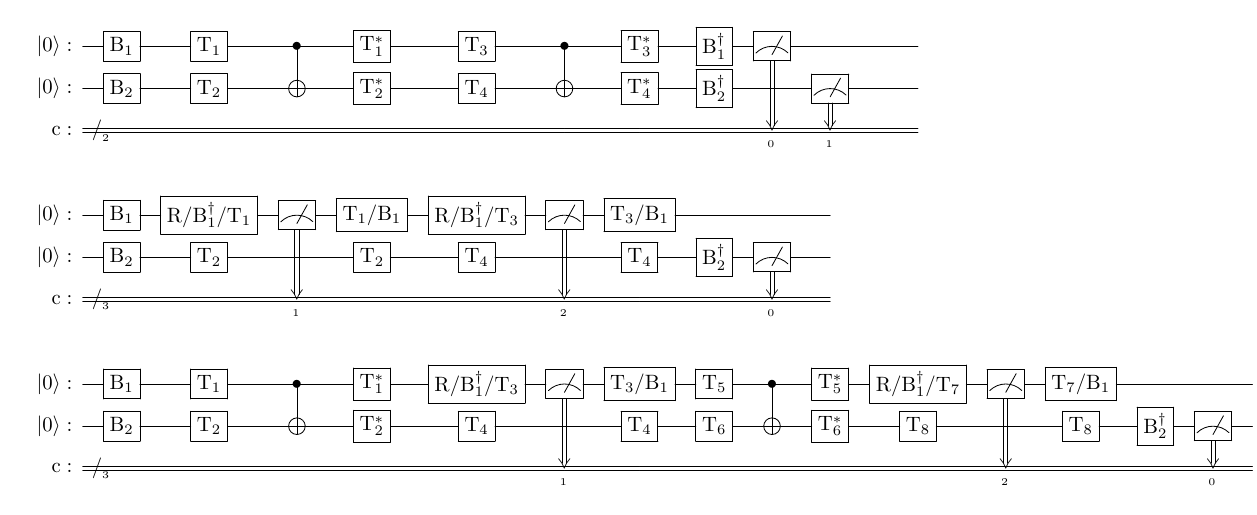}
    \caption{\justifying Three CB set for decoupling unlearnable coefficients using MCM noise with layer $L = 2$. The top to bottom circuits respectively represent conventional CB, CB for MCMs, and interleaving CNOT-MCMs. $B, T, R$ respectively denote basis transformation, Pauli twirling and single-qubit rotation. These circuits encompasses everything implemented in the two experiments in this work. }
    \label{fig:experimental circuits}
\end{figure*}

\begin{figure*}
    \centering
    \includegraphics[width=\linewidth]{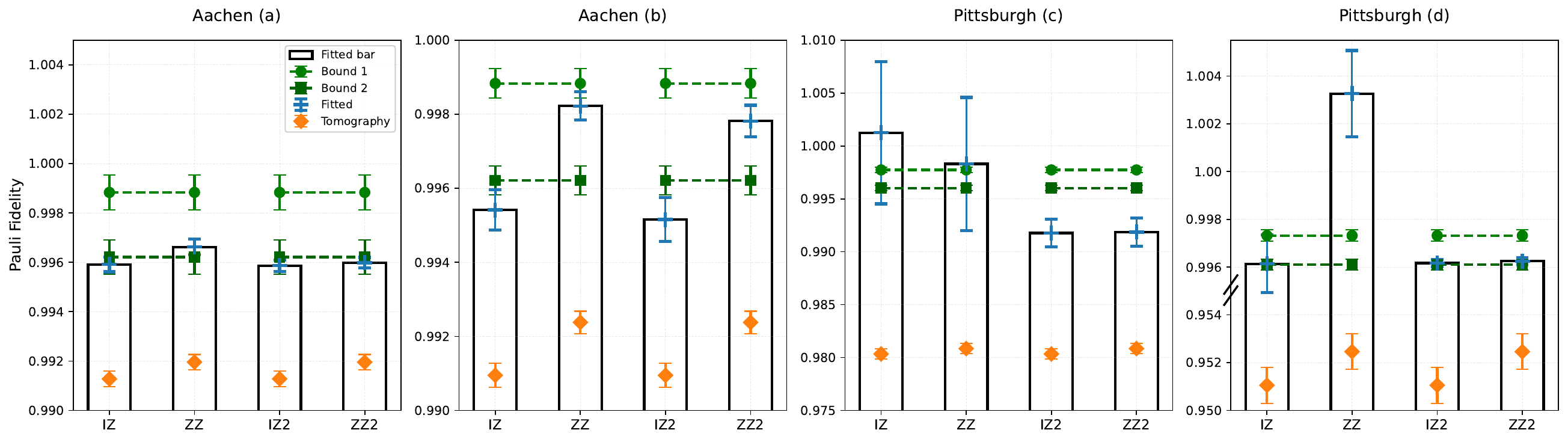}
    \caption{\justifying Cross quantum devices comparisons of CNOT Pauli noise learning utilizing MCMs-based CB (blue) with $R=64$ and $S=3000$ per cycle for the non-measuring neighbor of MCMs, and $R = 1920$ and $S = 100$ for CNOT gates. The total shot budget is $1.728\times 10^6$ for CNOT only CB, $1.344\times10^6$ for MCMs-based CB, and $1.344\times10^6$ for conventional tomography. All experiments have been benchmarked against -- the physical CPTP upper (lighter green) and lower (green) bounds, and regular fidelity tomography (orange). Data-set (a) and (b) represents the benchmarking results for different qubit pairs of \texttt{ibm-aachen}, and (c) and (d) of \texttt{ibm-pittsburgh}. In each subplot, noisy MCMs with increasing pre-measurement noise are used to study its effects on the unlearnable coefficients. Fidelity denoted as $\{$IZ, ZZ$\}$ are obtained from MCMs with single-qubit rotations. Meanwhile, $\{$IZ2, ZZ2$\}$ are obtained from MCMs without single-qubit rotation. Correctness of the results is assessed by how close it is to the CPTP bounds, and the relative difference between the unlearnable pairs compared to the width of the CPTP bounds.}
    \label{fig: Mid-circuit benchmarking}
\end{figure*}

\subsubsection{Unlearnable noise}

In the first experiment, the generalized CB procedure introduced in Section~\ref{section: algorithm} is applied directly to characterize CNOT gates, on two generations of quantum devices, \texttt{ibm-aachen} and \texttt{ibm-pittsburgh}. The task is to unravel the unlearnable noise pair $\{\lambda_{0z}, \lambda_{zz}\}$, which is only learnable as a product in standard CB. To assess the quality of MCM-based CB, three types of CB are implemented, summarized in Figure~\ref{fig:experimental circuits}, and regular fidelity tomography is performed. In the conventional CB benchmark, the CPTP conditions allow us to impose an upper and lower bound on the difference between $\{\lambda_{0z}, \lambda_{zz}\}$. We name these bounds the CPTP bounds. The cycle depth $[2, 18, 36, 52, 70, 88, 104, 122, 140]$ is chosen for \texttt{ibm-aachen} and $[2, 24, 46, 68, 90, 112, 134, 156, 180]$ for \texttt{ibm-pittsburgh}. Meanwhile, conventional tomography serves as the control of the experiment, where no cycles are implemented. This method assumes perfect state preparation, where measurement errors in the CNOT tomography are mitigated through a separate tomography on the measurement itself.

In this experimental background, we design the set-up for an MCMs-based CB protocol. Based on the fact that $\ket{1}$ is an order of magnitude less likely to appear compared to $\ket{0}$, we avoid twirling the MCMs of both \texttt{ibm-aachen} and \texttt{ibm-pittsburgh} to avoid exposure to the $\ket{1}$ state. Next, two types of MCMs are implemented, one with normal MCMs and MCMs preceded by a single-qubit rotation $\cos 12^\circ$. A single-qubit rotation that is followed by a measurement collapse is equivalent to adding a Pauli channel prior to a MCM, and the purpose is to show that pre-measurement errors can be mitigated in our characterization, according to the MCMs characterization algorithm of Subsection~\ref{subsection: subroutine 1}. Depending on the MCMs type, the choice of cycle depth varies. In this experiment, the cycle depth $[2, 21, 41, 61, 80, 100, 120]$ is chosen for MCMs without rotation and $[2, 18, 34, 51, 67, 83, 100]$ for MCMs with rotations. The final unlearnable pairs $\{\lambda_{0z}, \lambda_{zz}\}$ are extracted directly from data of MCMs-based CB, compared to the CPTP bounds and conventional tomography that are \textit{independently characterized}. 

The results are summarized in Figure~\ref{fig: Mid-circuit benchmarking}, with the data of the exponential fits available in Appendix~\ref{appendix: Exponential Fits of MCMs-based CB}. The progression of the figures, from left to right, represents the CB on qubit pairs $[103, 104]$ and $[64, 63]$ of \texttt{ibm-aachen}, and $[107, 108]$, and $[37, 45]$ of \texttt{ibm-pittsburgh}. These pairs are chosen from the set of qubits with the minimal measurement read-out errors, based on the IBM calibration on the day. Within each figure, we display on the left the learned Pauli fidelity through normal MCMs, and on the right through MCMs with single-qubit rotation. The final quality of the benchmarks, with the exception of (c) and the second bar of unlearnable noise in (d), agrees with the CPTP bounds and the fidelity difference extracted from conventional tomography. Although conventional tomography, illustrated in orange, produces fidelities that have a consistent relative difference compared to the width of the CPTP bound, overall MCMs-based CB outperform conventional tomography in terms of its closeness to the CPTP bound. 

Meanwhile, the large error bars observed in (c) correspond to the fact that the observed data points do not follow the exponential trend due to fluctuation. In the fittings of (c), some data points were removed from the fitting to obtain consistent fidelity estimations. Figure~\ref{fig: pittsburgh panel top} in Appendix~\ref{appendix: Exponential Fits of MCMs-based CB} visualizes the fittings that are employed to extracted the unlearnable fidelities. Meanwhile, the over-estimation in the second bar of unlearnable noise in (d) originates from the deviation of the exponential decay of expectation values $\langle ZZ\rangle$ toward $0$ as the number of cycles increases, leading to an under-estimation of the MCMs noise. Panel (h) of Figure~\ref{fig: pittsburgh panel bottom} in Appendix~\ref{appendix: Exponential Fits of MCMs-based CB} illustrates this effect, in contrast to (d) of Figure~\ref{fig: pittsburgh panel bottom} when no single-qubit rotations is implemented prior to MCM. 

Taken together, these results demonstrate that state-of-the-art superconducting MCMs, when calibrated, are robust against quantum noise for the purpose of computing Pauli fidelities. A comparison with conventional tomography shows that MCM-based CB exhibits partial SPAM immunity relative to the naive noise tomography discussed below Theorem~\ref{theorem: main result}. Given this accuracy, MCM-based CB is a practical tool for approximating the Pauli fidelities of $n$-qubit Clifford gates, at least in regimes where CPTP-like constraints are no longer available.

\subsubsection{Classical Read-out Error}
In our characterization protocol, classical read-out errors can be learnt as the statistical difference between post-processed and non-post-processed data. They can be learned independently from the MCMs CB set and the CNOT-MCMs set. Although the detected classical read-out errors appear to be small in our experiment, they are important to eliminate additional noise effects on the unlearnable noise fidelity that leads to the results in Figure~\ref{fig: Mid-circuit benchmarking}.

\subsection{Binomial Analysis}
\label{subsection: experimental binomial analysis}

\subsubsection{Binomial distribution}

In the second experiment, we apply the binomial analysis to detect non-Markovian quantum effects, a key challenge in quantum devices. Instead of performing binomial expansion on the unlearnable pairs $\{\lambda_{0z}, \lambda_{zz}\}$, in this study, we perform binomial analysis directly on the \textit{learnable} fidelity $\lambda_{z0}$, by expanding $\lambda_{z0}^N$ according to Eq.~(\ref{eqn: binom}). To characterize the possible types of non-Markovian effects, we perform an additional MCMs-based CB on \texttt{ibm-aachen} with/without twirling on the noisy MCMs.

\begin{figure*}[htbp]
\centering
\includegraphics[width=\linewidth]{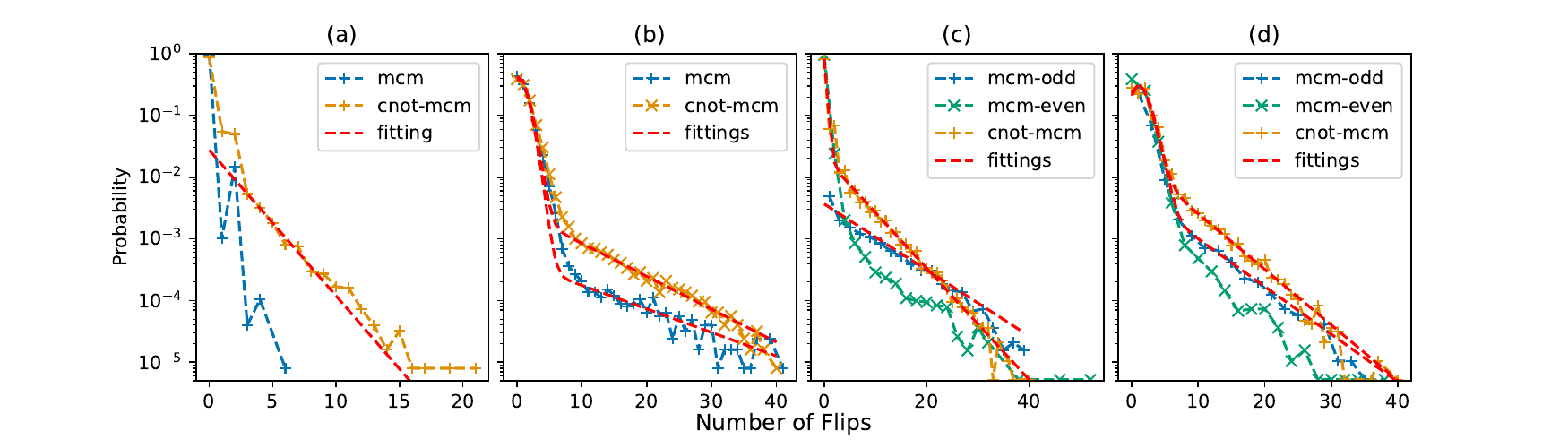}
\caption{\justifying
Binomial analysis on \texttt{ibm-aachen} and \texttt{ibm-pittsburgh} that demonstrates multiple modes of non-Markovian noise contributions. (a) and (b) represent binomial analysis of \texttt{ibm-aachen} with/without twirling with $R = 250$ and $S = 100$. (c) and (d) represent untwirled binomial analysis of $\texttt{ibm-pittsburgh}$ with/without single-qubit rotation with $R = 64$ and $S = 3000$ for the MCMs and $R = 1920$ and $S = 100$ for the CNOT gates. In (a) and (b) of \texttt{ibm-aachen}, the non-twirled result reveals a non-Markovian signature originating from classical read-out errors in the MCMs-only characterization, while the twirled result reveals a non-Markovian contribution that is attributed to leakage errors. In (c) and (d) of \texttt{ibm-pittsburgh}, the absence of twirling reveals how leakage dynamics bias toward the $\ket{1}$ state when CNOT is not involved, while the bias is eliminated when CNOT is applied, suggesting a non-Markovian effect from cross-talk rather than leakage from the Josephson junction alone. A single-qubit rotation is applied to (d) that induces a binomial peak, to demonstrate the failure of the binomial model in explaining the anomaly.}
\label{fig:non-Markovian-plot}
\end{figure*}

\begin{table*}[htbp]
\centering
\renewcommand{\arraystretch}{1.2}
\setlength{\tabcolsep}{4pt}

\caption{\justifying Fitted decay parameters for the non-Markovian signatures. The data correspond to the fitting in Figure~\ref{fig:non-Markovian-plot}. We observe the signature of a binomial statistic $p$, and non-Markovian contributions. The magnitude $A$ provides the order of magnitude of the deviation from the expected binomial statistics. Based on the model of classical read-out given leakage events, $\lambda_d$ provides the decay rate from the collective leakage state back to the computational basis.}
\label{tab:fitted_params}

\resizebox{\linewidth}{!}{%
  \begin{tabular}{lcccc}
    \toprule
    \textbf{Parameters} & \texttt{ibm-aachen} & \texttt{ibm-aachen} Twirled &
    \texttt{ibm-pittsburgh} ($\theta=0$) & \texttt{ibm-pittsburgh} ($\theta=12$) \\
    \midrule
    Binomial $p$ & -- & $0.988 \pm 2.65 \times 10^{-4}$ & $0.9991 \pm 9.14\times 10^{-5}$ & $0.986\pm 6.86\times 10^{-4}$ \\
    Decay Rate $\lambda_d$ & $0.580\pm 0.008$ & $0.884 \pm 3.94\times 10^{-3}$ & $0.812 \pm 9.31\times 10^{-3}$ & $0.813 \pm 5.12\times 10^{-3}$ \\
    Magnitude $A$ & $0.027 \pm 0.001$ & $2.90\times 10^{-3} \pm 1.80 \times 10^{-4}$ & $2.13\times 10^{-2}\pm 2.97\times 10^{-3}$ & $2.02\times 10^{-2}\pm 1.56\times 10^{-3}$ \\
    \midrule
    Binomial $p$ & -- & $0.989 \pm 2.06\times 10^{-4}$ & -- & $0.985 \pm 6.17\times 10^{-4}$ \\
    Decay Rate $\lambda_d$ & -- & $0.916 \pm 1.00\times 10^{-2}$ & $0.884 \pm 3.93\times 10^{-3}$ & $0.837 \pm 1.11\times 10^{-2}$ \\
    Magnitude $A$ & -- & $4.26\times 10^{-4} \pm 7.09\times 10^{-5}$ & $3.68\times 10^{-3} \pm 2.55\times 10^{-4}$ & $5.92 \times 10^{-3}\pm 1.12\times 10^{-3}$ \\
    \bottomrule
  \end{tabular}%
}
\end{table*}

In Figure~\ref{fig:non-Markovian-plot}, we collectively illustrate the binomial analysis of the generalized CB running on \texttt{ibm-aachen} on qubit pairs $[106, 107]$ and \texttt{ibm-pittsburgh} on $[107,108]$. These figures compare the types of binomial deviations, demonstrating the difference between with/without twirling in \texttt{ibm-aachen} and with/without single-qubit rotation in \texttt{ibm-pittsburgh}. In the data set of \texttt{ibm-aachen}, the binomial expansion of $\lambda_{z0}^{75}$ is illustrated in both (a) and (b); in \texttt{ibm-pittsburgh}, on the other hand, the binomial expansion $\lambda_{z0}^{120}$ is illustrated in (c) and $\lambda_{z0}^{100}$ in (d).

In panels (b) and (d), the binomial peak is observed: the twirling of \texttt{ibm-aachen} and the single-qubit rotation applied to \texttt{ibm-pittsburgh}, both in between the $10^{-1}$--$10^{-3}$ scales. Although both effectively implement a Markovian channel that dominates the noise process, the mechanisms behind these phenomena are different. The contrast between (b) and (a) without twirling verifies the fact that Pauli twirling converts a noise model to an effective Pauli channel. Meanwhile, the contrast between (d) and (c) without single-qubit rotation demonstrates that the dominating noise effect has been shifted from a non-Markovian source to a Markovian one, the single-qubit rotation. Nonetheless, in both cases, the quantum effects must obey the separability condition in Eq.~(\ref{eqn: separable 1}) and (\ref{eqn: separable 2}) to produce a binomial distribution.

Meanwhile, a zig-zag-like statistical composition is instead observed for the untwirled result in panel both (a) and (c) when the number of flips is less than or equal to $6$. Particularly for (a), this phenomenon is captured in our simulation in Figure~\ref{fig:asymmetry}. Below the $10^{-3}$ scale of panels (b), (c) and (d), we observe abnormal tails across, especially in the presence of the CNOT gate. According to the numerical illustration Figure~\ref{fig:asymmetry}, a binomial distribution does not predict a measurable contribution when $f$ is large, so the experimental observation is necessarily attributed to other quantum effects. The existence of the asymmetry and tails are violation of Markovianity. These anomalies motivates a further analysis of non-Markovian contributions in CB.

\subsubsection{Asymmetries}
\begin{figure}[htbp]
    \centering
    \includegraphics[width=0.9\linewidth]{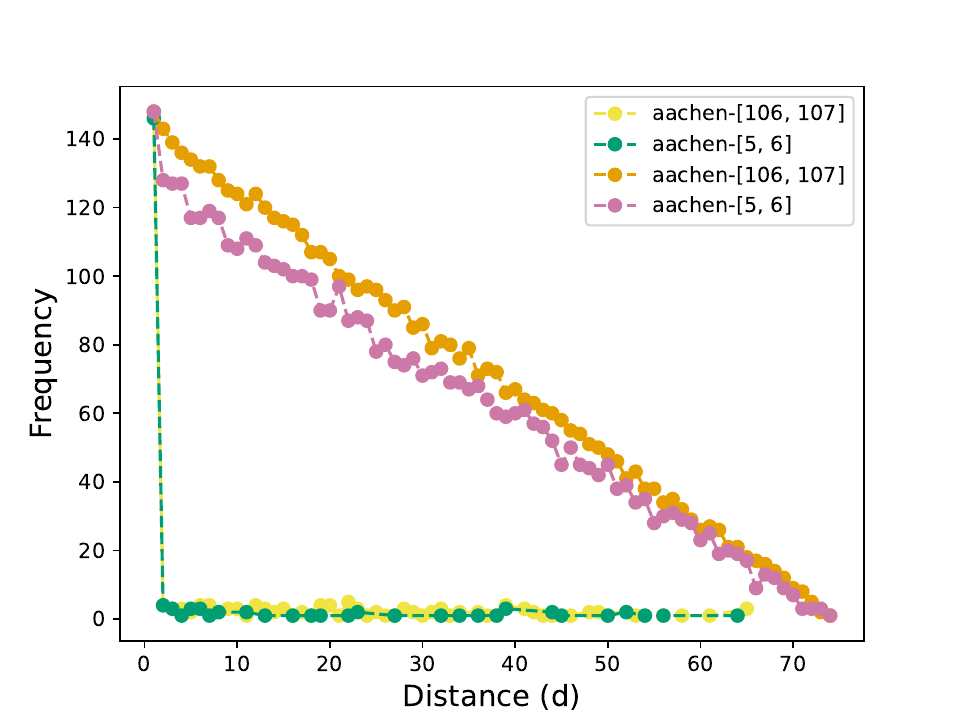}
    \caption{\justifying Length distribution of the MCMs read-out sequences that contribute to the statistic at $f = 2$ for \texttt{ibm-aachen} with/without twirling. $f=2$ is chosen since the number of classical read-out errors is constrained to one, if it ever occurs. This allows for a better isolation of the effect of read-out errors. The distribution of the twirled data set (orange and pink) displays an expected linear decay from binomial statistics. However, in the data set without twirling (yellow and green), $d=1$ dominates the distribution, suggesting the presence of classical read-out errors.}
    \label{fig:even-bias}
\end{figure}

In Figure~\ref{fig:non-Markovian-plot}, we illustrate two types of asymmetry that emerge in panels (a), (c), and (d), when CB is performed on MCMs only. The first type of asymmetry is the aforementioned zig-zag distribution that appears in panels (a) and (c). It is characterized by its bias toward the even flips. In Figure~\ref{fig:even-bias}, we simulate this effect by adding classical read-out errors to each of the MCMs. To understand the emergence of this zig-zag pattern, let us consider the flip statistical $f = 0$ given by a sequence of MCMs read-outs:
$$[0,0,0, ...,0].$$ Any read-out error on any one of the bits will simultaneously contribute to the transition of $0 \to 1$ and $1 \to 0$ in the middle of the read-out, such as:
$$[0,1,0, ...,0].$$
For $n$ events of independent read-out errors, post-processing will add an erroneous even contribution $2k$, upper bounded by $2n$, to the original read-out of $f = 0$. Therefore, some $f$ events are translated to $f+2k$. Since $f = 0$ and $f = 1$ effectively dominate the measurement statistics in the binomial distribution Eq.~(\ref{eqn: binom}) when $(1-\lambda_{P_p})/2$ is small, the zig-zag pattern emerges when part of $f = 0$ and $f = 1$ probabilities are translated to $f = 2k$ and $f = 2k+1$.

To verify that independent read-out errors contribute to the zig-zag pattern of the probability distribution, we further analyze the composition of read-out events that contribute to the probability statistics at $f$. In a read-out error dominated situation, the read-out sequence of MCMs will display disconnected values of $1$s due to statistical independence; while for a binomial distribution, the read-out sequence will display connected sequences of $1$s and $0$s due to the transition rules. By counting the length distributions of $1$s that make up the probability statistics at $f$, we pinpoint the effect that contributes to this zig-zag statistic. In Figure~\ref{fig:even-bias}, we illustrate the length composition of probability at $f=2$ of the binomial analysis for \texttt{ibm-aachen} in panels (a) and (b) of Figure~\ref{fig:non-Markovian-plot}. In contrast to the expected linear decay of length in the twirled case that generates the binomial distribution, the untwirled case exhibits an abrupt peak of the length distribution at $d=1$, verifying the presence of classical read-out error events.

A closer look at the MCMs benchmark in panels (c) and (d) of Figure~\ref{fig:non-Markovian-plot} reveals the second type of asymmetry that favors $\ket{1}$ in the MCM read-out over $\ket{0}$. To emphasize the bias, we intentionally colored the odd flips and even flips differently to contrast the different dynamics of the even and odd flip statistics. In panel (c), a transition of asymmetries takes place between $f = 5$ and $f = 6$, suggesting a shift of the dominating noise effect from classical read-out errors to a different type noise. The bias toward $\ket{1}$ strongly suggests the existence of a mechanism that induces excitations on $\ket{1}$ rather than $\ket{0}$. This is consistent with the current superconducting error phenomenology where the $\ket{1}$ state cross-talks with its neighboring qubits and leaks into the non-computational bosonic modes~\cite{Bengtsson2024willowmcm}. The bias persists in panel (d) despite implementing a single-qubit rotation $R_y(12^\circ)$ prior to each measurement, suggesting that this emerging pattern of bias cannot be explained away with a standard Markovian noise model in the computational state. In contrast, when we insert a CNOT gate, the time correlations remain, while the bias toward the $\ket{1}$ outcome is strongly suppressed. This points to a regime in which the non-Markovian correlation most likely does not originate from the Josephson junction alone, but rather is an effect of cross-talk with the neighboring qubits.

A key distinction between the first and second type of asymmetry is that the first type of asymmetry does not affect the expectation value computation. Taking Eq.~(\ref{eqn: binom}) as an example, to compute $\lambda_{I\otimes P_p}$, we perform the following summation:
\begin{equation}
    \lambda_{I\otimes P_p}=\sum_{f}(-1)^f\binom{N}{f}\left(\frac{1+\lambda_{I\otimes P_p}}{2}\right)^{N-f}\left(\frac{1-\lambda_{I\otimes P_p}}{2}\right)^f.
\end{equation}
but the effect of classical read-out errors $(-1)^{f+2k}$ is canceled under parity. However, the second asymmetry directly deviates from the binomial probability model. Unavoidably, this leads to a fundamental accuracy floor of the fidelity estimate.

\subsubsection{Non-Markovian Contributions}


We propose a leakage model that explains the abnormal exponential tail observed in the binomial analysis on the quantum gates of \texttt{ibm-aachen} and \texttt{ibm-pittsburgh}. Based on previous studies \cite{PhysRevLett.134.100601}, the emergence of exponential tails finds a coherent explanation with the read-out dynamics when a leakage event takes place. With superconducting qubits, the measurement output is a classical analog signal with two quadratures, so it is plotted in an IQ plane. The outcomes for $|0\rangle$ and $|1\rangle$ form two clouds, and then one chooses a decision boundary---a ``cut'' or threshold---to classify each shot as $0$ or $1$. Suppose that the leakage event induces an enlarged Markov process with probability $\lambda_d$ that the superconducting state remains in leakage, and $1-\lambda_d$ that it returns to the computational state. When the leakage event occurs, the IQ plane can register additional readout clusters between the $|0\rangle$ and $|1\rangle$ decision regions, which can cause substantial random bit-flips in the classification. Since the duration of the random bit-flip is proportional to the time when the quantum state remains in the leaked state, the length of the random bit-flip is proportional to the duration -- which occurs with probability proportional to $\lambda_d^f$:
\begin{equation}
    A^\prime\lambda_d^f(1-\lambda_{d}),
\end{equation}
with $A^\prime$ a normalization of the probability distribution.

When combining this empirical model with the signature binomial distribution, a new probability model emerges that describes the non-Markovianity:
\begin{align}
    p_{\mathrm{mod}}(f) &= (1-\epsilon)\binom{M}{f}\left(\frac{1-\lambda_{Z_p}}{2}\right)^f\left(\frac{1+\lambda_{Z_p}}{2}\right)^{M-f}\\&+ \epsilon A\lambda_d^f,
    \label{eqn: mod}
\end{align}
where $\epsilon$ represents the probability that a time-correlated event takes place, $A$ is a normalization, and $\lambda_d$ is the exponential decay rate. A linear sum between the binomial distribution and the deviation captures the fact that these two represent independent mechanisms of noise sources.

With this model, non-Markovian effects such as the decay rate and magnitude of leakage events become characterizable. The resulting fitting parameters are demonstrated in Table~\ref{tab:fitted_params}. The data show consistent relative increases of magnitude in the presence of the CNOT gate, but also an increase in the decay rate $\lambda_d$ with respect to the number of flips. These observations, combined with the asymmetries, suggest two types of non-Markovian events taking place depending on the presence of the CNOT gate. However, in both cases, the exponential behavior is preserved. The observation of a consistent non-Markovian behavior implies an accuracy floor for the measured Pauli fidelities. The benchmarks show that the deviations are relatively large and can lead to non-negligible errors in deep circuits. Potentially, these results imply new consistent error mitigation strategies that would make allowance for this effect and improve such computations. High-quality fidelity estimation is necessary for error mitigation and error correction. From the perspective of hardware optimization, the deviation from a binomial distribution is simple to measure and may be employed as a means to minimize non-Markovian effects under different settings. In Appendix~\ref{appendix: time correlation}, we further analyze time correlations between the flip statistics, to verify the existence of time correlation in the process of CB.

\section{Conclusion} 
In this paper, we present a framework for generalized CB using MCMs, addressing a fundamental limitation of standard CB: the inability to resolve Pauli fidelities that are coupled across Clifford-induced cycles in the Pauli transfer matrix. The framework operates at three levels of increasing generality, each corresponding to a different assumption about the quality of the MCM resource available on hardware.

At the first level, when MCMs realize high-fidelity $z$-basis projections with negligible post-projection quantum error, Theorem~\ref{theorem: exact} establishes that all previously unlearnable Pauli fidelities of both the MCM and the interleaved Clifford gate become exactly learnable. The key insight is that all classical feed-forward can be deferred entirely to post-processing via the QICF construction, requiring no real-time quantum control. The state preparation cost of breaking a given Clifford cycle is precisely quantified by Theorem~\ref{theorem: main result} in terms of the Hamming weight of the Pauli weight-pattern mismatch $\delta(P_\alpha, P_\beta)$, which simultaneously specifies the MCM insertion locations and measures the overhead relative to standard CB. Importantly, classical read-out assignment error does not enter the learnability condition --- it is a purely classical mislabelling that does not couple noise parameters across the projection, and can be characterized and divided out independently. When only $\ket{0}$ preparation is reliably noiseless, as is typical for superconducting platforms, Theorem~\ref{theorem: approx} shows that the unlearnable Pauli fidelities remain approximable for CNOT gate via a geometric-arithmetic mean decomposition enabled by post-selection on the $\ket{0}$ outcome.

At the second level, when the projection quality of the MCMs cannot be guaranteed to satisfy either learnability condition, the QICF framework remains useful through its compatibility with the consistent Pauli noise learning framework~\cite{chen2025disambiguatingpaulinoisequantum}. Because MCMs with classical feed-forward are operationally equivalent to Clifford gates at the level of Pauli noise propagation, they can be incorporated into the gate set tomography design matrix without modification. The gate set is thereby reduced to its minimal meaningful unit --- a Clifford gate paired with its associated MCM layer --- and the unified characterization ensures that residual unlearnable gauge ambiguity does not propagate into error mitigation, even on devices where projection quality is imperfect. This compatibility means that the generalized CB protocol is not restricted to idealized hardware conditions but extends to the full class of dynamic circuits available on current quantum processors.

At the third level, the introduction of MCMs into the CB protocol opens a direct diagnostic window onto noise dynamics that are invisible to final-measurement-only protocols. Under the Markovian separable USI model, the flip statistics of repeated MCM readouts follow a clean binomial distribution parameterized solely by $\lambda_{Z_p}$; any deviation constitutes direct evidence of temporal correlations beyond the Markovian assumption. Experimentally, this diagnostic reveals two distinct non-Markovian signatures on \texttt{ibm-aachen} and \texttt{ibm-pittsburgh}: a classical readout asymmetry favoring even flip counts, attributable to single-round readout errors and harmless to fidelity estimation, and a universal exponential tail at large flip counts, attributable to leakage-induced memory whose mean duration is directly captured by the decay rate $\lambda_d$ of the phenomenological model of Eq.~(\ref{eqn: mod}). The magnitude of this non-Markovian noise floor is non-negligible on current devices and sets a precision floor on any Pauli noise learning technique that rests on the Markovian assumption. In the best case --- the \texttt{ibm-pittsburgh} device --- the extracted unlearnable fidelity pair achieves within-CPTP-bound characterization, demonstrating that state-of-the-art superconducting MCMs already carry sufficient projection quality for physically consistent noise learning with partial SPAM immunity.

Overall, this work demonstrates that MCMs are not merely a passive diagnostic tool but an active resource for noise characterization: they resolve unlearnable gauge degrees of freedom in Clifford CB, enable consistent noise learning when projection quality is imperfect, and expose non-Markovian dynamics that standard CB cannot access. While the current benchmarks are limited to superconducting devices, the framework is hardware-agnostic. Extending it to trapped-ion, neutral-atom, and semiconductor platforms --- where the dominant non-Markovian mechanisms differ fundamentally from the leakage-dominated superconducting case --- is a natural next step, and one where the higher MCM fidelities available on trapped-ion hardware may enable exact noise learning under Theorem~\ref{theorem: exact} more reliably than current superconducting devices permit.

\bibliography{noise}

\appendix

\section{Mathematical Foundations of MCM Noise Learning}
\label{appendix: all proofs}

In this section, we establish key results in MCM noise learning for general $n$-qubit Clifford gates using Markovian noisy MCMs.

\subsection{Pauli Propagation of Stochastic CNOT}
\begingroup
\def\thetheorem{\ref{theorem: 1}}
\begin{theorem}[Pauli Propagation for Stochastic CNOT]
A stochastic CNOT, formed by tracing out the classical register of the QICF, replicates the Pauli transformation of a coherent CNOT gate on $Z$-type observables, up to a scalar prefactor from the classical readout error. Formally, tracing out the classical register gives:
\begin{equation}
    \text{CX}_{\text{sto}, p, q}(\cdot) = \sum_{r, k\in \mathbb{Z}_2}q_r\,\mathcal{CX}_{k,r; p,q}(\cdot).
\end{equation}
For $c \in \mathbb{Z}_2$, the stochastic CNOT preserves the Pauli transformation structure of a coherent CNOT in the $z$-basis:
\begin{align}
    \text{CX}_{\text{sto}, p, q}(Z_p^cZ_q) &= (1-2q_{1})\,\omega_{Z_p^{c+1}Z_q}\,Z_p^{c+1}Z_q,
    \\
    \text{CX}_{\text{sto}, p, q}(\mathcal{X}_p\mathcal{X}^c_q(\cdot)) &= \mathcal{X}_p\mathcal{X}^{c+1}_q\bigl(\text{CX}_{\text{sto}, p, q}(\cdot)\bigr),
\end{align}
where the classical readout error $(1-2q_1)$ factors out independently of the Pauli structure, and:
\begin{align}
    &\omega_{Z_p^{c+1}Z_q} = p_{00}\lambda_{00}+(-1)^cp_{10}\lambda_{10}
    -\bigl((-1)^{c}p_{01}\lambda_{01}+p_{11}\lambda_{11}\bigr), \nonumber\\
    &\Lambda_{ab}(Z^c_pZ_q) = \lambda_{ab}Z^c_pZ_q. \nonumber
\end{align}
\end{theorem}
\addtocounter{theorem}{-1}
\endgroup

\begin{proof}
A stochastic CNOT channel is defined as a quantum instrument with classical feed-forward
on the neighboring qubit:
\begin{equation}
    \mathcal{CX}_{k, r;p, q}(\cdot) = \sum_{r, a, b \in \mathbb{Z}_2}q_rp_{ab}\Lambda_{ab}\mathcal{X}^{k+r}_q
    \otimes (\mathcal{X}^b\superproj{k}\mathcal{X}^a)_p(\cdot).
\end{equation}
Note that $\Lambda_{ab}(\cdot)$ is a general $(n-1)$-qubit channel with dependence on the measurement errors $a, b$ that occur at qubit $p$. We suppress the qubit indices $p, q$ when the context is clear. We define the
marginalized stochastic CNOT channel as:
\begin{equation}
    CX_{\mathrm{sto}}(\cdot) = \sum_{k,r\in\mathbb{Z}_2}\mathcal{CX}_{k,r}(\cdot),
\end{equation}
which averages over the classical readout $k$. We proceed to characterize the action of
$CX_{\mathrm{sto}}$ on the Pauli basis, before generalizing to arbitrary Pauli channels.

\textbf{Step 1: Action on $Z_p^cZ_q$.} Consider the action of $CX_{\mathrm{sto}}$
on $Z_p^cZ_q$:
\begin{align}
    CX_{\mathrm{sto}}(Z_p^cZ_q)
    &= \sum_{k,r, a, b \in \mathbb{Z}_2}q_rp_{ab}\left[\Lambda_{ab}\mathcal{X}_q^{k+r}\left[
    \mathcal{X}_p^b\superproj{k}\mathcal{X}_p^a(Z_p^cZ_q)\right]\right],\nonumber\\
    &= \sum_{r, a, b\in\mathbb{Z}_2}(-1)^{ac+b(c+1)}p_{ab}\left[\Lambda_{ab}q_r\mathcal{X}^r_q(Z_p^{c+1}Z_q)\right],
    \nonumber \\
    &= (1-2q_1)\omega_{Z_p^{c+1}Z_q}\, Z_p^{c+1}Z_q,
\end{align}
where the eigenvalue $\omega_{Z_p^{c+1}Z_q}$ is given by:
\begin{align}
    &\omega_{Z_p^{c+1}Z_q} = p_{00}\lambda_{00}+(-1)^cp_{10}\lambda_{10}-((-1)^{c}p_{01}\lambda_{01}+p_{11}\lambda_{11})\nonumber, \\
     &\Lambda_{ab}(Z^c_pZ_q)= \lambda_{ab}Z^c_pZ_q. \nonumber
\end{align}
The process is stochastic rather than coherent because $CX_{\mathrm{sto}}(PX_p) = 0$ for any Pauli $P$, meaning that quantum information in non-$z$-basis directions is destroyed
by the projective measurement.

\textbf{Step 2: Sequential application.} We now consider the behavior of
$CX_{\mathrm{sto}}$ under sequential composition. Statistical independence of successive
measurement outcomes, which is a consequence of Markovianity, implies:
\begin{equation}
    \sum_{k_{p, 1}, k_2\in \mathbb{Z}_2}\mathcal{CX}_{k_{p, 1}, p, q}
    \left(\mathcal{CX}_{k_2, p, q}(\cdot)\right)
    = CX_{\mathrm{sto}, p, q}\left(CX_{\mathrm{sto}, p, q}(\cdot)\right).
    \label{eqn: stats independence}
\end{equation}
When $CX_{\mathrm{sto}}$ is applied sequentially to the same pair of qubits for $2n$
cycles, it follows by induction that:
\begin{equation}
    CX_{\mathrm{sto}}^{2n}(Z_p^cZ_q) =
    (1-2q_1)^{2n}\left(\omega_{Z_p^cZ_q}\omega_{Z_p^{c+1}Z_q}\right)^n Z_p^cZ_q.
\end{equation}

\textbf{Step 3: Interleaving with a general quantum channel.} We now consider the case
where $CX_{\mathrm{sto}}$ is interleaved with a general quantum channel $\Lambda$. By
linearity of $\Lambda$:
\begin{align}
    \Lambda\left(\mathcal{CX}_{0}(\cdot)\right)+\Lambda\left(\mathcal{CX}_{1}(\cdot)\right)
    &= \Lambda\left(CX_{\mathrm{sto}}(\cdot)\right),\\
    \mathcal{CX}_{0}\left(\Lambda(\cdot)\right)+\mathcal{CX}_{1}\left(\Lambda(\cdot)\right)
    &= CX_{\mathrm{sto}}\left(\Lambda(\cdot)\right).
\end{align}
For Pauli operators $P \in \{I, X, Y, Z\}^{\otimes 2}$ not prepared in the $z$-basis,
the second identity vanishes due to the projective measurement. As a concrete example,
suppose that $\Lambda$ is a coherent CNOT gate satisfying $CX(Z_p^cZ_q) =
\gamma_{Z_p^{c+1}Z_q}Z_p^{c+1}Z_q$. Then, for $n$ interleaved cycles of
$CX_{\mathrm{sto}}$ and $CX$:
\begin{align}
    &\sum_{k_{p, 1},\ldots, k_n \in \mathbb{Z}_2}\mathcal{CX}_{k_{p, 1}}\left(CX\left(\cdots
    \mathcal{CX}_{k_n}\left(CX(Z_p^cZ_q)\right)\right)\right)\nonumber\\
    &= CX_{\mathrm{sto}}\left(CX\left(\cdots CX_{\mathrm{sto}}\left(CX(Z_p^cZ_q)
    \right)\right)\right)\nonumber\\
    &= (1-2q_1)^n\left(\omega_{Z_p^cZ_q}\gamma_{Z_p^{c+1}Z_q}\right)^n
    Z_p^cZ_q.
    \label{eqn: new coupling}
\end{align}

\textbf{Step 4: Action on a general Pauli channel.} Finally, we show that
$CX_{\mathrm{sto}}$ transforms a general two-qubit Pauli channel
$(X_q^NZ_q^n)\otimes(X_p^oZ_p^r)$ as follows:
\begin{align}
    &CX_{\mathrm{sto}}\circ \left(\mathcal{X}_q^{m}\mathcal{Z}_q^n\mathcal{X}_p^o\right)(\cdot)\nonumber \\
    &= \sum_{a, b \in \mathbb{Z}_2}p_{ab}(\Lambda_{ab}\mathcal{X}_q^{k}(\mathcal{X}_q^N\mathcal{Z}_q^n))
    \otimes (\mathcal{X}_p^b\ket{k}\bra{k}\mathcal{X}_p^a(\mathcal{X}_p^o\mathcal{Z}_p^r))(\cdot)\nonumber \\
    &= \sum_{a, b \in \mathbb{Z}_2}p_{ab}((\mathcal{X}_q^N\mathcal{Z}_q^n)\Lambda_{ab}\mathcal{X}_q^{k})
    \otimes (\mathcal{X}_p^b\ket{k}\bra{k+o}\mathcal{X}_p^a)(\cdot)\nonumber\\
    &= \sum_{a, b \in \mathbb{Z}_2}p_{ab}((\mathcal{X}_q^N\mathcal{Z}_q^n)\Lambda_{ab}\mathcal{X}_q^{k+o})
    \otimes (\mathcal{X}_p^b\ket{k+o}\bra{k}\mathcal{X}_p^a)(\cdot)\nonumber\\
    &= \left(\mathcal{X}_q^{m+o}\mathcal{Z}_q^n\mathcal{X}_p^o\right)\circ CX_{\mathrm{sto}}(\cdot).
    \label{eqn: Clifford transform}
\end{align}
This result shows that $CX_{\mathrm{sto}}$ preserves the transformation
$\mathcal{X}_q^N \mathcal{X}_p^o \rightarrow \mathcal{X}_q^{m+o}\mathcal{X}_p^o$ when acting on Pauli channels in the
$X$-direction, consistent with the coherent CNOT transformation. However, $Z$-string
information $(Z_p^r)$ is lost due to the projective measurement basis, and the
transformation of $Z$-type Paulis does not follow the coherent CNOT rule.

Combining Steps 1--4, we conclude that $CX_{\mathrm{sto}}$ partially preserves the
Clifford transformation structure of a coherent CNOT gate: it faithfully reproduces the
coherent CNOT action on $Z_p^cZ_q$ as a Pauli basis element and on
$\mathcal{X}_q\mathcal{X}_p^c$ as a Pauli channel, while the remaining transformation properties are
collapsed by the choice of projective measurement basis.

\textbf{Step 5: Generalization.} Since $\Lambda$ is a multi-qubit Pauli channel, adding Pauli operators that are not subject to the transformation of the stochastic CNOT does not alter the form of the results, namely:
\begin{align}
    &\text{CX}_{\mathrm{sto}, p, q}(Z_p^cZ_q\otimes P_{\gamma}) = (1-2q_{1})\omega_{Z_p^{c+1}Z_q\otimes P_{\gamma}}Z_p^{c+1}Z_q\otimes P_{\gamma},\\
    &\text{CX}_{\mathrm{sto}, p, q}(\mathcal{X}_p\mathcal{X}^c_q\otimes \mathcal{P}_{\gamma}(\cdot)) = \mathcal{X}_p\mathcal{X}^{c+1}_q\otimes \mathcal{P}_{\gamma}(\text{CX}_{\mathrm{sto}, p, q}(\cdot)).
\end{align}
where:
\begin{align}
    \omega_{Z_p^{c+1}Z_q\otimes P_{\gamma}} &= p_{00}\lambda_{00, P_{\gamma}}+(-1)^cp_{10}\lambda_{10, P_{\gamma}}\nonumber \\
    &-((-1)^{c}p_{01, }\lambda_{01, P_{\gamma}}+p_{11}\lambda_{11,P_{\gamma}})\nonumber, \\
     \Lambda_{ab}(Z^c_pZ_q\otimes P_\gamma) &= \lambda_{ab, P_{\gamma}}Z^c_pZ_q\otimes P_\gamma. \nonumber
\end{align}
\end{proof}
\label{subpendix: pauli propagation}
\begingroup
\def\thetheorem{\ref{theorem: USI pauli propagation.}}
\begin{theorem}[Pauli Propagation for USIs]
A uniform stochastic instrument (USI) measuring a single qubit $p$, when averaged over all measurement outcomes, acts as a Pauli channel on $Z$-type observables and commutes with bit-flip operations. Formally, define $\mathcal{M}(\cdot) = \sum_{k\in \mathbb{Z}_2}\mathcal{M}_{k}(\cdot)$ as the outcome-averaged USI. For $c\in \mathbb{Z}_2$:
\begin{align}
    \mathcal{M}(Z_p^cZ_q) &= \delta_{Z_p^cZ_q}\,Z_p^cZ_q,
    \\
    \mathcal{M}(\mathcal{X}_p^{c_1}\mathcal{X}_q^{c_2}(\cdot)) &= \mathcal{X}_p^{c_1}\mathcal{X}_q^{c_2}(\mathcal{M}(\cdot)),
\end{align}
where the fidelity coefficient is:
\begin{align}
    \delta_{Z_p^cZ_q} &= p_{00}\lambda_{00}+(-1)^cp_{10}\lambda_{10}+(-1)^{c}p_{01}\lambda_{01}+p_{11}\lambda_{11}, \nonumber\\
    \Lambda_{ab}(Z^c_pZ_q) &= \lambda_{ab}Z^c_pZ_q. \nonumber
\end{align}
\end{theorem}
\addtocounter{theorem}{-1}
\endgroup
\begin{proof}
    \textbf{Step 1: Follow-through.} From Eq.~(\ref{eqn: standard USI}) in the main text, by following the same steps of derivation in the proof of Theorem~\ref{theorem: 1} with the exception of removing $\mathcal{X}_q^{k+r}(\cdot)$ from the process, the result above follows immediately. In particular, $(1-2q_1)$ vanishes due to the absence of classical feed-forward. Employing the same $z$-basis linearity argument as in the previous proof, the quantum instrument acts exactly like a Pauli noise channel with $z$-basis input in the measurement.

    \textbf{Step 2: Generalization.} Similarly to the proof of Theorem~\ref{theorem: 1}, the generalized version of the theorem holds without changing the form of the expression:
    \begin{align}
        \mathcal{M}(Z_p^c\otimes P_\gamma) &=\delta_{Z_p^cZ_q}Z_p^c\otimes P_\gamma,\\
        \mathcal{M}(\mathcal{X}_p^{c_1}\otimes \mathcal{P}_\gamma(\cdot)) &=\mathcal{X}_p^{c_1}\otimes\mathcal{P}_\gamma(\mathcal{M}(\cdot)).
    \end{align}
    where:
    \begin{align}
        \delta_{Z_p^c\otimes P_{\gamma}} &= p_{00}\lambda_{00, P_{\gamma}}+(-1)^cp_{10}\lambda_{10,  P_{\gamma}}\nonumber \\
        &+(-1)^{c}p_{01}\lambda_{01,P_{\gamma}}+p_{11}\lambda_{11,P_{\gamma}}\nonumber,\\
        \Lambda_{ab}(Z^c_p\otimes P_{\gamma})&= \lambda_{ab, P_{\gamma}}Z^c_p\otimes P_{\gamma}. \nonumber
    \end{align}
\end{proof}

\textbf{Remarks:} The choice of measurement basis naturally induces a hierarchy between Pauli operators. In this example, the Pauli-$X$ operator and its transformation are promoted to a bit-flip channel, while the Pauli-$Z$ operator becomes the information basis for information bits. Intuitively, the above result states that in a fixed choice of basis, quantum probability behaves like classical probability.

\subsection{Learnable Degrees of Freedom}
\label{subpendix: learnable dof}
\begin{theorem}
    For cycles of quantum instruments, the classical read-out error is learnable SPAM-free if the pre-measurement and post-measurement errors are separable.
    \label{theorem: classical read-out error}
\end{theorem}
\begin{proof}
    Based on Theorems~\ref{theorem: USI pauli propagation.} and~\ref{theorem: 1}, the following can be learned via CB, isolating SPAM errors ($p$ and $q$ are dropped given the clear context):
    \begin{align}
    f_1 &= A(\omega_{Z^cZ}\omega_{Z^{c+1}Z})^n+B,\\ f_2 &= C\delta_{Z^cZ}^{2n}+D,
    \end{align}
    in one characterization circuit for $2n$ cycles. Note that in the proofs of those theorems, the structural forms of these Pauli fidelities still hold when we extend from two-qubit fidelities to $n$-qubit fidelities, provided that they do not interfere with the transformation.

    All learnable information on the Pauli fidelities of $IZ$ and $ZZ$ from both the quantum instruments with and without classical post-processing, in the $z$-basis, are listed as follows:
    \begin{align}
        \lambda_r\omega_{IZ}\omega_{ZZ} &= (1-2q_1)^2((p_{00}\lambda_{00}-p_{11}\lambda_{11})^2
        \label{eqn: constraint 1}
    \nonumber\\
    &-(p_{10}\lambda_{10}-p_{01}\lambda_{01})^2),\\
        (\delta_{IZ}+\delta_{ZZ})/2 &= p_{00}\lambda_{00}+p_{11}\lambda_{11},
        \label{eqn: constraint 2}\\
        (\delta_{IZ}-\delta_{ZZ})/2  &= p_{10}\lambda_{10}+p_{01}\lambda_{01}.
        \label{eqn: constraint 3}
    \end{align}
    We treat $p_{ab}\lambda_{ab}$ as one unknown for the reconstruction of the quantum instrument, as it suffices to reconstruct all of $p_{ab}\Lambda_{ab}$, where $p_{ab}$ is computable from the trace-preservation of $\Lambda_{ab}$. Effectively, we have three constraints for five unknowns, contributing to unlearnable coefficients in MCMs generalized CB. By treating the pre- and post-measurement errors as separable, we assume the following holds:
    \begin{equation}
        p_{ab} = p_a^1p_b^2, \quad \lambda_{ab} = \lambda_a^1\lambda_b^2.
    \end{equation}
    with the effective Pauli fidelity:
    \begin{align}
        \omega_{Z^{c+1}Z} &= (p^1_0\lambda ^1_{0}+(-1)^cp^1_1\lambda ^1_{1})(p^2_0\lambda ^2_{0}+(-1)^{c+1}p^2_1\lambda^2_{1}),\\
        \delta_{Z^cZ} &=(p^1_0\lambda ^1_{0}+(-1)^cp^1_1\lambda ^1_{1})(p^2_0\lambda ^2_{0}+(-1)^cp^2_1\lambda^2_{1}).
    \end{align}
    It immediately follows that:
    \begin{equation}
        (1-2q_1)\frac{\omega_{IZ}\omega_{ZZ}}{\delta_{IZ}\delta_{ZZ}} = (1-2q_1).
    \end{equation}
\end{proof}
\textbf{Remarks:} The separable case eliminates the classical read-out degrees of freedom, but it also reduces one constraint as the product of $\omega_{Z^cZ}$ is now the same as the product of $\delta_{Z^cZ}$. Since the degrees of freedom for a quantum instrument exceed those of the constraints by 2, approximation is required.

\begingroup
\def\thetheorem{\ref{theorem: exact}}
\begin{theorem}[Exact Noise Extraction Principle]
All Pauli fidelities of a single-qubit USI are exactly learnable from measurement data, provided that no quantum errors occur \emph{after} the measurement projection. Formally, given $\mathcal{M}(\cdot)$ that obeys Eq.~(\ref{eqn: stochastic Pauli channel})--(\ref{eqn: stochastic channel prop}), all Pauli fidelities are learnable if:
\begin{equation}
    p_{01} = p_{11} = 0.
\end{equation}
Classical readout assignment errors $q_{\vec{r}}$ do not affect learnability.
\end{theorem}
\addtocounter{theorem}{-1}
\endgroup
\begin{proof}
    First, vanishing post-measurement bit-flip implies separable pre-measurement and post-measurement errors. Employing Theorem~\ref{theorem: classical read-out error}, the classical read-out error can be isolated. Meanwhile, the following also holds:
    \begin{equation}
        \omega_{Z^{c+1}Z} = \delta_{Z^cZ}.
    \end{equation}
From Eq.~(\ref{eqn: new coupling}), the pairs $\{\omega_{Z^cZ},\gamma_{Z^{c+1}Z}\}$ can be decoupled, isolating $\gamma_{Z^{c+1}Z}$, the unlearnable coefficients.
\end{proof}

\textbf{Remarks:} The condition of perfect state preparation, which matches properties in trapped ions or neutral atoms by projection, is only one possible condition for reducing the degrees of freedom appearing in the learnability constraints Eq.~(\ref{eqn: constraint 1}), (\ref{eqn: constraint 2}), and (\ref{eqn: constraint 3}). The real question boils down to which approximation is most suitable for modeling noise phenomena for approximate unlearnable noise extraction. For example, if quantum instruments have non-separable noise, one may instead wish to make the following approximation: $q_1 \approx0, p_{01}\approx0$. In this paper, we focus on learnability based on the quality of state preparation for the sake of practicality.

\subsection{Deferred Feed-Forward Principle for Generalized CB}
\label{subpendix: deferred feed-forward principle}

Building on the learnability results for Markovian MCM errors, we establish the second key result in this work --- that single-qubit Clifford gates, combined with MCMs, are sufficient to decouple all unlearnable degrees of freedom for arbitrary $n$-qubit Clifford gates.

\begingroup
\def\thelemma{\ref{lemma: generalized post-processing}}
\begin{lemma}[Deferred Feed-Forward Principle for CB]
In generalized CB interleaved with a Clifford gate and a QICF, no real-time quantum feed-forward is needed. All classical feed-forward operations that manipulate Pauli weights of the propagating Paulis can be deferred to classical post-processing without changing the measurement statistics.
\end{lemma}
\addtocounter{lemma}{-1}
\endgroup
\begin{proof}
    Consider the following $n$-qubit quantum circuit in the superoperator formalism:
    \begin{equation}
        \prod_{i=1}^N \text{Ctrl-}M_i\left[\Lambda_i\circ\mathcal{G}_i(\cdot)\right].
    \end{equation}
    consisting of an $n$-qubit noisy Clifford gate $\Lambda\circ \mathcal{G}_i(\cdot)$ at the $i$-th layer, and an $(n+1)$-qubit QICF $\text{Ctrl-}M_i$ defined in Eq.~(\ref{eqn: QICF with cerr}) that acts on an ancillary qubit named qubit $s$. The purpose of qubit $s$ is to receive all conditional Pauli-$X$ operations from the measurement read-outs on the $n$-qubit system. Specifically, each $M_i$ can be represented as an $l_i$-dimensional vector $\vec{M}_i$:
    \begin{equation}
        \vec{M}_i = [1,1,\ldots, 1].
    \end{equation}
    Each classical read-out per cycle, denoted $\vec{r}$, results in the conditional Pauli-$X$ operation on qubit $s$:
    \begin{equation}
        \mathcal{X}^{\vec{M}_i\cdot\vec{r}}_s(\cdot) = \prod_{c = 1}^{l_i} \mathcal{X}^{\vec{r}[c]}(\cdot). \label{eqn: X accumulation}
    \end{equation}
    In this expression, the conditional operation is decomposed into products of $\mathcal{X}^{\vec{r}[c]}$ operations. This step decouples measurement read-out operations into products of stochastic CNOT operations, where the transformation rule introduced in Theorem~\ref{theorem: 1} can be leveraged.

    With this in mind, let us consider an $l$-dimensional Pauli operator, as a subset of the $n$-qubit system, whose Pauli weight we wish to change:
    $$P_{i_1}^{v_1}\otimes P_{i_2}^{v_2}\otimes \cdots\otimes P_{i_l}^{v_l},$$
    where $i_k \in\{1,2, 3\}$ denotes the Pauli-$X$, $Y$, and $Z$ operators, and ${v}_k$ is the integer power of the Paulis. The ancillary qubit $s$ extends the Pauli operator to:
    $$Z_s\otimes P_{i_1}^{v_1}\otimes P_{i_2}^{v_2}\otimes \cdots\otimes P_{i_l}^{v_l},$$
    where the sequence of stochastic CNOTs acts between $Z_s$ and $P_{i_k}$ whose Pauli weight is altered. In Eq.~(\ref{eqn: ZZ propagation}), the measured qubit is where the change of Pauli weight occurs. With the appropriate choice of single-qubit Clifford gates, the Pauli weight can be manipulated statistically as follows:
    \begin{equation}
        Z_s\otimes P_{i_1}^{v_1+1}\otimes P_{i_2}^{v_2+1}\otimes \cdots\otimes P_{i_l}^{v_l+1}.
        \label{eqn: Pauli weight change}
    \end{equation}
    A QICF with single-qubit Clifford gates can induce arbitrary Pauli weight changes at the qubit locations where we perform the MCMs, and importantly all Pauli-$X$ operations are performed on qubit $s$. Since qubit $s$ does not interfere with the quantum circuit, we may replace qubit $s$ with a classical bit $s$. Since Eq.~(\ref{eqn: X accumulation}) consists of commuting Pauli-$X$ operations, we can implement all conditional Pauli-$X$ operations in post-processing, rather than as sequences of stochastic CNOTs, even in the setting of many cycles. This defers QICF to post-processing.
\end{proof}

\begingroup
\def\thetheorem{\ref{theorem: main result}}
\begin{theorem}[State Preparation Cost from Pauli Weight-Pattern Mismatch]
Let $\mathcal{G}$ be an $n$-qubit Clifford gate, and let $P_\alpha$, $P_\beta$
be a pair of Paulis satisfying $\mathcal{G}(P_\alpha) = P_\beta$.
Then the Pauli fidelity $\lambda_{P_\beta}$ is learnable provided that USIs
satisfying Theorem~\ref{theorem: exact} are inserted between successive applications
of $\mathcal{G}$ in the CB protocol at the qubit locations specified by
\begin{equation}
    \delta(P_\alpha, P_\beta).
\end{equation}
The number of MCMs required equals the Hamming weight $|\delta(P_\alpha, P_\beta)|$,
which also quantifies the overhead of generalized CB relative to standard CB.
\end{theorem}
\addtocounter{theorem}{-1}
\endgroup

\begin{proof}
    \textbf{Pauli Weight Reversal:}
    Lemma~\ref{lemma: generalized post-processing} shows that a stochastic manipulation of the Pauli weight is possible with the insertion of MCMs at chosen locations in a Pauli basis of one's choice. Following from the machinery of the previous proof, a generalized CB satisfies the set-up
    $$\mathcal{G}_i = \mathcal{G},$$ and
    $$\text{Ctrl-}M_i = \text{Ctrl-}M.$$
    Suppose that we transform $P_\alpha\rightarrow P_\beta$ through $\mathcal{G}$. Consider the collection of Paulis of $P_\alpha$ where the Pauli weight changes: $P_{i_k}\rightarrow I$ and $I\rightarrow P_{i_j}$. The locations where Pauli weights are altered are given by the Pauli weight difference between $P_\alpha$ and $P_\beta$:
    $$\delta(P_\alpha, P_\beta).$$
    By Eq.~(\ref{eqn: Pauli weight change}), a Pauli weight reversal transformation can be implemented by measuring the selected set of Paulis. Finally, single-qubit Clifford gates are sufficient for reversing the Pauli-type transformation induced by $\mathcal{G}$. Therefore, it is sufficient to construct the reversal $\mathcal{R}_{\delta(P_\alpha, P_\beta)}: P_\beta\rightarrow P_\alpha$ with MCMs and single-qubit Clifford gates.

    \textbf{MCMs Noise Learning:} The next task is to show that MCMs noise parameters can be learned, so that the MCMs noise parameters and the Clifford gate noise can be decoupled. In Section~\ref{section: algorithm}, we gave the procedure for the extraction of the pre-measurement noise and classical read-out errors. Following through that procedure, and given the guarantee of state preparation quality in Theorem~\ref{theorem: exact}, MCMs noise becomes isolable, and therefore so does the noise of a general Clifford gate.
\end{proof}

\textbf{Remarks:} In quantum circuits consisting of Clifford gates and MCMs, the rules of classical feed-forward can be deferred to post-processing, through which Pauli weight is manipulated. Since measurement and reset are effectively a control-$X$ on the same qubit, it is possible to implement measurement and reset in post-processing. The correspondence between feed-forward and post-processing provides a guide to the implementation of dynamic circuits on quantum hardware.

\section{Algorithm 2: Error Suppression}
\label{appendix: scheme 2 proof}

\begingroup
\def\thetheorem{\ref{theorem: approx}}
\begin{theorem}
    With high-fidelity preparation of $\ket{0}$ alone, the Pauli fidelities of CNOT gate can still be approximated. Formally, given a USI that obeys Eq.~(\ref{eqn: stochastic Pauli channel})--(\ref{eqn: stochastic channel prop}) with $p_{01} \approx 0$, the Pauli fidelity pairs of CNOT gate are approximable via post-selection on the $\ket{0}$ outcome.
\end{theorem}
\addtocounter{theorem}{-1}

\begin{proof}
We start with a noisy MCMs system with an unusable $\ket{1}$ state, classical read-out assignment error rate $\epsilon_0$, and the remaining errors from MCMs defined as: $\lambda_{\text{post}}$ for the error fidelity of $\ket{0}$ state preparation, and $\lambda_{\text{pre}}^{01}, \lambda_{\text{pre}}^{10}, \lambda_{\text{pre}}^{11}$ for the noise coefficients from the pre-measurement two-qubit error channel of the Pauli operators $P_iI, IZ, P_iZ$ respectively. The classical read-out error corresponding to $0\rightarrow 1$ events is denoted $\epsilon_0$. Meanwhile, a control parameter:
\begin{equation}
    \lambda_{\text{Reg}}(\theta) = \cos(\theta),
\end{equation}
is injected into the circuit prior to each MCM to control the read-out fidelity. For the noise learning of the CNOT gate, we consider the circuit set-up in Figure~\ref{fig: Mid-circuit benchmarking} in the main text, and derive their corresponding learned coefficients accordingly. When we are guaranteed only the $\ket{0}$ state, we prefer to post-select the quantum statistics with sequential MCMs read-outs of $\ket{0}$. This corresponds to the sequential operation of the $m = l = 1$ quantum instrument $\mathcal{M}_0(\cdot)$ on the two-qubit Pauli operator. Define the pair of coupling Pauli operators as $\{\lambda_{i0}, \lambda_{ij}\}$, and the learnable coefficient as $\lambda_{0j}$. With alternating CNOT and noiseless MCMs corresponding to the third circuit of Figure~\ref{fig:experimental circuits}, the unnormalized density matrix of the unmeasured qubit upon the post-selected CB is (Appendix~\ref{appendix: Read-out error derivation}):
\begin{equation}
    \rho_{\text{noiseless}}=\frac{(1+\lambda_{0k})^NI+(\lambda_{i0}+\lambda_{\mathrm{Reg}}(\theta)\lambda_{ij})^NP_i}{2^N}.
\end{equation}
Meanwhile, if instead a noisy quantum instrument $\mathcal{M}_0$ is applied to the quantum state, the terms such as $\lambda_{i0}+\lambda_{\mathrm{Reg}}(\theta)\lambda_{ij}$ receive additional noise contributions:
$$\gamma_1(\theta) =\lambda_{i0}\lambda_{\text{pre}}^{10}\lambda_{\text{post}}+\lambda_{\text{Reg}}(\theta)\lambda_{ij}\lambda_{\text{pre}}^{11}.$$
If a noisy CNOT gate and $\mathcal{M}_0$ are applied alternately, it results in the following density matrix at the end of the quantum circuit:
\begin{equation}
    \rho_{\text{noisy}} = \left(I+\left(\frac{\gamma_1(\theta)}{p_1/(1-\epsilon_0)}\right)^N P_i\right)/2,
    \label{eqn: noisy M0}
\end{equation}
with probability
\begin{equation}
    p_1^N= \left(\frac{1+\lambda_{\text{Reg}}(\theta)\lambda_{0k}\lambda_{\text{pre}}^{01}\lambda_{\text{post}}}{2}\right)^N(1-\epsilon_0)^N.
\end{equation}
The additional $(1-\epsilon_0)^N$ term emerges due to the fact that some $\ket{0}$ states are lost from classical bit-flip errors.

Next, we consider the benchmarking of purely noisy MCMs, the second circuit in Figure~\ref{fig:experimental circuits}. Similarly to the previous argument, the expectation value of $P_i$ is expected to be $\frac{1+\lambda_{\text{Reg}}(\theta)}{2}$ given perfect MCMs, but $\mathcal{M}_0$ contributes additional noise terms to the final measurements:
$$\gamma_0(\theta) = \lambda^{10}_{\text{pre}}+\lambda_{\text{Reg}}(\theta)\lambda_{\text{pre}}^{11}\lambda_{\text{post}}.$$
A sequential operation of noisy MCMs thus results in the following effective density matrix:
\begin{equation}
    \rho_{\text{mcm}} = \left(I+\left(\frac{\gamma_0(\theta)}{p_0/(1-\epsilon_0)}\right)^NP_i\right)/2,
    \label{eqn: MCMs M0}
\end{equation}
with probability
\begin{equation}
    p_0^N = \left(\frac{1+\lambda_{\text{Reg}}(\theta)\lambda_{\text{pre}}^{01}\lambda_{\text{post}}}{2}\right)^N(1-\epsilon_0)^N.
\end{equation}
At first glance, Eq.~(\ref{eqn: noisy M0}) and (\ref{eqn: MCMs M0}) appear to be irrelevant to the learning of the arithmetic mean $\lambda_{\text{am}}(\theta) = \frac{\lambda_{i0}+\lambda_{\text{Reg}}(\theta)\lambda_{ij}}{2}$. However, there is a remarkable approximation through which $\lambda_{\text{am}}$ can be estimated to high precision. Given $m-n = \mathcal{O}(\delta)$, the following approximation holds:
\begin{equation}
    \frac{1+m(1-\delta)}{1+n(1-\delta)} = \frac{1+m}{1+n}+\mathcal{O}(\delta^2),
    \label{eqn: suppression}
\end{equation}
which leads to the approximation by combining Eq.~(\ref{eqn: noisy M0}) and Eq.~(\ref{eqn: MCMs M0}):
\begin{equation}
\frac{\gamma_1(\theta)(1+\lambda_{\text{Reg}}(\theta))}{\gamma_0(\theta)}=\lambda_{\text{am}}(\theta)+\mathcal{O}(\delta^2).
\end{equation}
Here we assume that all MCMs error fidelities are of the same order of magnitude $1-\mathcal{O}(\delta)$. This leads to first-order error cancellation of all pre-measurement errors $\lambda_{\text{pre}}^{01}, \lambda_{\text{pre}}^{10}, \lambda_{\text{pre}}^{11}$ and robustness to the post-measurement error $\lambda_{\text{post}}$. The classical read-out error $\epsilon_0$ is also canceled during the division. Combined with the geometric mean $\lambda_{\text{gm}}(\theta) = \sqrt{\lambda_{i0}\lambda_{ij}\lambda_{\text{Reg}}(\theta)}$, we can estimate the unlearnable coefficients:
\begin{align}
    \lambda_{i0} &= \lambda_{\text{am}}(\theta)+\sqrt{\lambda_{\text{am}}(\theta)^2-\lambda_{\text{gm}}(\theta)^2},\\
    \lambda_{ij} &= \left(\lambda_{\text{am}}(\theta)-\sqrt{\lambda_{\text{am}}(\theta)^2-\lambda_{\text{gm}}(\theta)^2}\right)/\lambda_{\text{Reg}}(\theta).
\end{align}
The concept of a regulator (angle $\theta$) is central to this characterization. Without it, the difference between $\lambda_{\text{am}}(0)$ and $\lambda_{\text{gm}}(0)$ is too small to be distinguishable. Using the scaling difference, the term $\lambda_{\text{am}}(\theta)^2-\lambda_{\text{gm}}(\theta)^2$ is magnified as the regulator angle $\theta$ increases. The resolution lower bound of our approximation scheme is shown in Figure~\ref{fig:regulator}. Since $\lambda_{\text{post}}$ has a strong impact on the fidelity difference crossover, an estimate of the state preparation fidelity allows one to prepare the ideal $\lambda_{\text{Reg}}(\theta)$ value for high-precision unlearnable noise learning.
\end{proof}

\begin{figure}
    \centering
    \includegraphics[width=\linewidth]{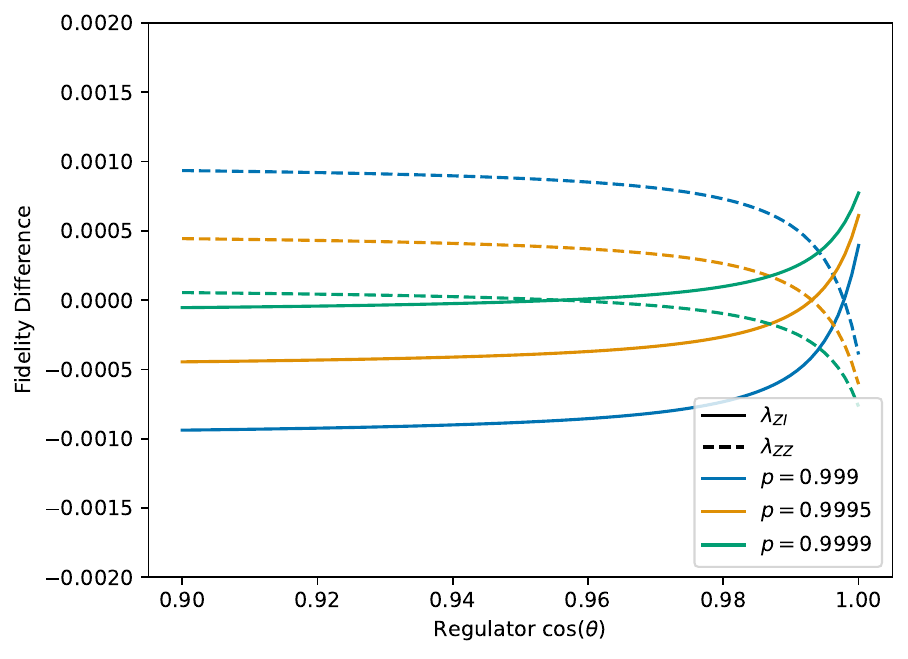}
    \caption{\justifying Fidelity difference between the approximation calculated analytically and the exact solution of Pauli fidelities. The unlearnable Pauli fidelities are set as $\lambda_{ZI} = 0.995, \lambda_{ZZ} = 0.991$. The fidelity difference decreases as the state preparation fidelity $p = \lambda_{\rm post}$ of the noisy MCM improves.}
    \label{fig:regulator}
\end{figure}

\section{Binomial Analysis and Classical Read-out Error from MCMs}
\label{appendix: Read-out error derivation}

We algebraically derive unlearnable noise extraction using perfect MCMs, without directly referring to the stochastic CNOT gate in Eq.~(\ref{eqn: stochastic CNOT}). We then further expand the derivation to the accumulation of classical read-out error:
\begin{align}
    \lambda_{r} &= (1-2q_1),\\
    q_1 &= \frac{\epsilon_0+\epsilon_1}{2},
\end{align}
for $\epsilon_{0/1}$ the read-out assignment error rate of $\ket{0/1}$.

\subsection{Binomial Distribution}

Consider a more restricted form of USI --- a separable $l = 1$ USI, acting on qubit $p$:
\begin{equation}
\mathcal{N}_{k}(\cdot) = \sum_{a,b \in \mathbb{Z}_2}p_{a}p_{ b}^\prime\left[\Lambda_{a}\Lambda_{b}^\prime\otimes \mathcal{X}^{b}_p\superproj{k}_p\mathcal{X}^{a}_p\mathcal{K}\right](\cdot).
\end{equation}
We adopt this representation to isolate a single-qubit measurement event from a parallel operation of MCMs governed by Theorem~\ref{theorem: main result}. In this definition, the events $a$ and $b$ are separable. The notation $\mathcal{K}$ denotes the set-up consisting of the Clifford gate $\mathcal{G}$ and the reversal channel $\mathcal{R}_{\delta(P_\alpha, P_\beta)}$ (see Section~\ref{subpendix: deferred feed-forward principle}), except that we leave one of the stochastic CNOTs unprocessed. We implicitly assume that $\mathcal{G}$ is prepared in the $z$-basis at qubit $p$ through single-qubit Clifford gates.

To analyze Markovianity, we sequentially apply $\mathcal{N}_{k_p}$ to a quantum circuit and use $P_\alpha$ as a probe to detect time correlations:
\begin{equation}
    \sum_{k_{p, 1},\cdots,k_{p,N}\in \mathbb{Z}_2} \prod_{i=1}^N\mathcal{N}_{k_{p,i}}(P_\alpha),
    \label{eqn: binomial inital Pauli}
\end{equation}
where $P_\alpha$ is implicitly rotated to the $z$-basis at qubit $p$. Due to the sequential operation, if we account for the finite systematic error that is accounted for by the initial and final SPAM, we may effectively move the previous post-measurement error $p_b\Lambda_b$ of $\mathcal{N}_{k_{i-1}}$ to the next channel $\mathcal{N}_{k_{i}}$, resulting in the new definition:

\begin{align}
    \mathcal{N}_{k}^\prime(\cdot) &= \superproj{k}_p\sum_{a,b \in \mathbb{Z}_2} \left[(p_{a}\Lambda_{a}\otimes \mathcal{X}^{a}_p)\mathcal{K}(p^\prime_b\Lambda^\prime_b\otimes\mathcal{X}^{b}_p)\right](\cdot),\\
    \mathcal{N}_{k}^\prime(\cdot) &= \superproj{k}_p\mathcal{L}(\cdot).
\end{align}
That is, the definition of a separable USI allows us to rearrange the channel into a combination of a noisy Clifford gate $\mathcal{K}$ contaminated with Pauli channels $p_a\Lambda_a\otimes\mathcal{X}^{a}_p$ and $p^\prime_b\Lambda^\prime_b\otimes\mathcal{X}^{b}_p$, and a perfect measurement $\superproj{k}_p$. In this form, we may treat $\mathcal{L}$ as a whole Markovian quantum channel, whose Markovian properties may now be studied with respect to the projection $\superproj{k}_p$.

In Eq.~(\ref{eqn: binomial inital Pauli}), $P_\beta$ is the initial Pauli that propagates through the CB of $\mathcal{N}_{k_i}$. When all MCMs are processed as QICF, we expect both the Clifford gate and the reversal to yield $\mathcal{R}\circ\mathcal{G}(P_\beta) = \lambda_{P_\beta}P_\beta$. However, we have left one MCM unprocessed at qubit $p$. As mentioned in Section~\ref{subsection: binomial analysis}, this results in the pair of Paulis $\{P_\gamma\otimes I_p, P_\gamma\otimes Z_p\}$ that are coupled to each other. As we will show below, the existence of such a pair leads to a binomial distribution. Consider the following sum of operators that will propagate through the Markovian channel:
\begin{equation}
    J_{k_{p, i}} = P_\gamma\otimes \frac{(1+(-1)^{k_{p, i}}Z)_p}{2}.
\end{equation}

Provided that:
\begin{equation}
    \mathcal{L}(I\otimes Z_p)= I\otimes Z_p,
    \label{eqn: control gate restriction}
\end{equation}
the Pauli propagation rule of $\mathcal{L}$ has the following property:
\begin{align}
    \mathcal{L}(P_\gamma\otimes I_p) &= \lambda_{P_\gamma\otimes Z_p}P_\gamma\otimes Z_p,\\
    \mathcal{L}(P_\gamma\otimes Z_p) &= \lambda_{P_\gamma\otimes I_p}P_\gamma\otimes I_p.
\end{align}
The restriction of Eq.~(\ref{eqn: control gate restriction}) is common in control gates and Pauli channels, which means binomial analysis is applicable for the characterization of both Pauli channels and control quantum gates.

Upon one application of $\mathcal{N}^\prime_{k_{p, 1}}(J_{0, k_{p, 0}})$, we obtain:
\begin{align}
    \mathcal{N}^\prime_{k_{p, 1}}(J_{k_{p, 0}}) &= \lambda_{P_\gamma}(k_{p, 0},k_{p, 1})J_{k_{p, 1}},\\
    \lambda_{P_\gamma}(k_{p, 0},k_{p, 1}) &= (-1)^{k_{p, 0}}\frac{\lambda_{P_\gamma\otimes I_p}+(-1)^{k_{p,0}+k_{p, 1}}\lambda_{P_\gamma\otimes Z_p}}{2}, \label{eqn: binomial eigenvalue}
\end{align}
where we use $k_{p,i}$ to denote the $i$-th read-out of MCMs at qubit $p$. Recalling Eq.~(\ref{eqn: the flip variable}), where the flip variable $f_{p,i} = k_{p, i-1}\oplus k_{p, i}\in \mathbb{Z}_2$ is introduced, and given a sequential operation of Eq.~(\ref{eqn: binomial inital Pauli}) without the summation and $f = \sum_{i}f_{p,i}$, we can establish the following identity:
\begin{equation}
    \prod_{i=1}^N\mathcal{N}_{k_{p,i}}(J_{0}) = (-1)^{N_{\vec{k}_p}}(\lambda_{P_{\gamma}}(0, 1))^f(\lambda_{P_{\gamma}}(0,0))^{N-f}J_{k_{p, N-1}}.
\end{equation}
where $N_{\vec{k}_p} = \sum_{i = 0}^{N-1}k_{p,i}$ is the sum of the signs of the values in Eq.~(\ref{eqn: binomial eigenvalue}). Define $\lambda_{P_\gamma}(f)$ as the expectation value measured when $f$ flips are detected in the MCM read-out. Since there are $\binom{N}{f}$ equivalent read-outs for which the event may be registered as an $f$-flip, expanding the distribution of expectation values with respect to $f$ yields the following result:
\begin{equation}
    \lambda_{P_\gamma}(f) = (-1)^{N_{\vec{k}_p}}\binom{N}{f}(\lambda_{P_{\gamma}}(0, 1))^f(\lambda_{P_{\gamma}}(0,0))^{N-f}.
    \label{eqn: generalized binomial distribution}
\end{equation}
When we finally post-process the MCMs read-out with the stochastic CNOT, the sign $(-1)^{N_{\vec{k}_p}}$ can be removed, because at each iteration the sign $(-1)^{k_{p,0}}$ of Eq.~(\ref{eqn: binomial eigenvalue}) can be canceled with the sign in $J_{k_{p,1}}$. Interestingly, the following identity holds:
\begin{equation}
    (-1)^{\sum_{i}f_{p, i}} = (-1)^f =(-1)^{k_{p,N-1}}.
\end{equation}
Therefore, measuring the following quantity and averaging over the read-outs recovers the desired Pauli fidelities:
\begin{align}
    \sum_{f}\mathrm{Tr}\{|\lambda_{P_\gamma}(f)|J_{k_{p,N-1}}P_{\gamma}\otimes I_p\} &= \lambda_{P_\gamma\otimes I_p}P_{\gamma}\otimes I_p,\label{eqn: binomial fidelity 1}\\
    \sum_{f}\mathrm{Tr}\{|\lambda_{P_\gamma}(f)|J_{k_{p,N-1}}P_{\gamma}\otimes Z_p\} &= \lambda_{P_\gamma\otimes Z_p}P_{\gamma}\otimes Z_p.
    \label{eqn: binomial fidelity 2}
\end{align}
These are precisely the fidelities of $P_\beta$. As a special case when $P_\gamma = I^{\otimes(n-1)}$, we recover the binomial distribution.

Therefore, MCMs provide a powerful tool that probes the binomial spectrum given Markovian Clifford gates and separable USI. Each Pauli fidelity learned in the quantum hardware admits a binomial expansion, so that we may read out the Markovian contributions to the fidelities. Conversely, when a deviation from the binomial statistics is observed, we can also faithfully detect it via this method. Some deviations, such as the classical read-out errors in Figure~\ref{fig:asymmetry}, do not affect the final computation of the fidelities. However, others, such as the results in Figure~\ref{fig:non-Markovian-plot}, reveal an accuracy bound where the Markovian assumption is replaced by new non-Markovian dynamics. This analysis is particularly useful in the sense that the deviations provide a new tool for assessing the accuracy lower bound of current Pauli noise learning techniques that rely on the Markovian assumption.

\subsection{Classical Read-out Error}

In Eq.~(\ref{eqn: stochastic CNOT}), it is already assumed that the read-out errors are averaged due to randomized compilation. However, we expect that the classical read-out assignment rate differs depending on the collapsed quantum basis state $0$ or $1$. In this section, we derive the effective read-out error fidelity due to randomized compilation.

Let us define the bit-flip probability $1\rightarrow 0$ as $\epsilon_0$ and $0\rightarrow 1$ as $\epsilon_1$. When considering noise learning with classical feed-forward, the asymmetry $\epsilon_0/\epsilon_1$ is naturally correlated with each classical read-out $(-1)^{k_{p, i}}$ of the noisy MCMs, and this contributes to a non-linear noise effect. This leads to a non-systematic characterization of classical read-out errors that depends on the read-out. Via randomized compilation, the classical read-out errors are decoupled from the classical feed-forward, since the physical read-outs are randomized while the digital read-outs are not. In the following, we show how randomized compilation allows the extraction of classical read-out errors from the classical feed-forward.

Consider the form of Eq.~(\ref{eqn: generalized binomial distribution}). When we compute the final Pauli fidelity of Eq.~(\ref{eqn: binomial fidelity 1}) or (\ref{eqn: binomial fidelity 2}), we would like to sample the following expectation values with equal weight:
\begin{equation}
    S_1 = \{|\lambda_{P_\gamma}(0)|, |\lambda_{P_\gamma}(1)|,\ldots, |\lambda_{P_\gamma}(N)|\}.
\end{equation}
The sum or the alternating sum will then give us the desired fidelities. However, imperfect classical feed-forward also forces us to sample from the opposite set:
\begin{equation}
    S_2 = \{-|\lambda_{P_\gamma}(0)|, -|\lambda_{P_\gamma}(1)|,\ldots, -|\lambda_{P_\gamma}(N)|\}.
\end{equation}
This set subsequently computes the negative of Eq.~(\ref{eqn: binomial fidelity 1}) or (\ref{eqn: binomial fidelity 2}). Suppose that, in the untwirled case, the number of $1$s in the read-out is the exponent of the sign $N_1 =N_{\vec{k}_p}$, and the number of $0$s is $N_0 = N-N_{\vec{k}_p}$. The probability, after all errors in the classical feed-forward are accounted for, that the sampling remains in $S_1$ rather than flipping to the undesirable $S_2$ is:
\begin{align}
    p_{S_1}(N_0) 
    &= \left[\sum_{k \text{ even}}\binom{N_0}{k}\epsilon_0^k(1-\epsilon_0)^{N_0-k}\right]\nonumber\\
    &\quad\times\left[\sum_{k \text{ even}}\binom{N_1}{k}\epsilon_1^k(1-\epsilon_1)^{N_1-k}\right]\nonumber\\
    &\quad+\sum_{k \text{ odd}}(\dots)\times \sum_{k \text{ odd}}(\dots)\nonumber\\[6pt]
    &= \frac{(1+(1-2\epsilon_0)^{N_0})(1+(1-2\epsilon_1)^{N_1})}{4}\nonumber\\
    &\quad+\frac{(1-(1-2\epsilon_0)^{N_0})(1-(1-2\epsilon_1)^{N_1})}{4}\nonumber\\[6pt]
    &= \frac{1+(1-2\epsilon_0)^{N_0}(1-2\epsilon_1)^{N_1}}{2}.
\end{align}

Next, we invoke bit-flip averaging. This allows $N_0$ to be randomly scrambled in the physical qubits while the logical read-out remains $N_0$. The averaging process removes the dependence on $N_0$ and leads to the effective read-out probability:
\begin{align}
    p_{S_1} &= \frac{1}{2^N}\sum_{N_0}\binom{N}{N_0}p_{S_1}(N_0),\nonumber\\
    &=\frac{1+(1-\epsilon_0-\epsilon_1)^N}{2}.
\end{align}
In post-processing, the effective error channel becomes:
\begin{align}
    \mathcal{C}(\cdot) &= \frac{1+(1-\epsilon_0-\epsilon_1)^N}{2}I(\cdot)I\nonumber\\
    &+\frac{1-(1-\epsilon_0-\epsilon_1)^N}{2}Q_p(\cdot)Q_p.
\end{align}
acting on $P_\gamma$ with $\{P_\gamma, Q_p\} = 0$ (due to the classical feed-forward).
Since we sample from the set of eigenvalues of $S_1$ with probability $p_{S_1}$ and from $S_2$ with probability $1-p_{S_1}$, the new effective density matrix with average error $q_1 = \frac{\epsilon_0+\epsilon_1}{2}$ between $S_1$ and $S_2$ is given by:

\begin{align}
    \mathcal{C}\left(\sum_{\lambda_k\in S_1}\lambda_kP_\gamma\otimes I\right) &= p_{S_1}\sum_{\lambda_k\in S_1}\lambda_kP_\gamma\otimes I\nonumber \\
    &-p_{S_2}\sum_{\lambda_k\in S_1}\lambda_kP_\gamma\otimes I\nonumber, \\
    &= (1-2q_1)^N\lambda_{P_\gamma\otimes I}P_\gamma\otimes I.
    \label{eqn:corrected-res}
\end{align}
If we define
\begin{align}
    S_1^\prime = &\{|\lambda_{P_\gamma}(0)|, -|\lambda_{P_\gamma}(1)|,\ldots, \nonumber\\
    &(-1)^f|\lambda_{P_\gamma}(f)|,\ldots, (-1)^N|\lambda_{P_\gamma}(N)|\}
\end{align}
as the alternating sum of the unnormalized binomial distribution, we instead derive:
\begin{align}
    \mathcal{C}\left(\sum_{\lambda_k\in S_1^\prime}\lambda_kP_\gamma\otimes Z\right) &= p_{S_1}\sum_{\lambda_k\in S_1^\prime}\lambda_kP_\gamma\otimes Z\nonumber \\
    &-p_{S_2}\sum_{\lambda_k\in S_1^\prime}\lambda_kP_\gamma\otimes Z,\nonumber \\
    &= (1-2q_1)^N\lambda_{P_\gamma\otimes Z}P_\gamma\otimes Z.
\end{align}
Therefore, twirling averages out classical read-out errors.

\subsection{Impact on Binomial Analysis}
In the previous sections, we derived the binomial expression for binomial analysis directly from propagating the Paulis through the channels, and computed the impact of classical read-out errors on the fidelity. Here, the effect of classical read-out errors on the binomial distribution is further explored. From Section~\ref{section: Experiment}, asymmetry is observed in the binomial analysis, and Figure~\ref{fig:even-bias} points to the fact that the majority of the asymmetry contributions originate from single read-out error events from $f = 0$ events that populate the $f = 2$ events. In Figure~\ref{fig:asymmetry}, we simulate, with the USI model, the emergence of asymmetry in the presence of classical read-out assignment errors, verifying the experimental observation.

\section{Exponential Fits of MCMs-based CB}
\label{appendix: Exponential Fits of MCMs-based CB}

Figures~\ref{fig: aachen panel top}, \ref{fig: aachen panel bottom}, \ref{fig: pittsburgh panel top}, and \ref{fig: pittsburgh panel bottom} present the exponential interpolations of the MCM-based CB corresponding to panels (a)--(d) of Figure~\ref{fig: Mid-circuit benchmarking}, respectively. The fitting and post-processing techniques are taken directly from the exact unlearnable noise learning method described in Subsection~\ref{subsection: CNOT gate}, and the resulting fits provide additional information on fitting quality that is reflected in Figure~\ref{fig: Mid-circuit benchmarking}. The blue dashed lines show CB data without post-processing, while the orange lines show data obtained from the deferred measurement principle of Theorem~\ref{lemma: generalized post-processing}. The panel layout is consistent across all four figures: left half of the columns correspond to CB on CNOT gates interleaved with MCMs, right halves to CB on MCMs alone; top rows correspond to benchmarking without single-qubit rotations, bottom rows to benchmarking with single-qubit rotations applied prior to the MCM. Post-processing on the top rows yields the unlearnable fidelity pairs labelled $\{IZ2, ZZ2\}$ in Figure~\ref{fig: Mid-circuit benchmarking}, while post-processing on the bottom rows yields the pairs labelled $\{IZ, ZZ\}$. Sub-panels (a), (c), (e), and (g) provide the data needed to extract the Pauli fidelity $\lambda_{IZ}$, while (b), (d), (f), and (h) provide the data for $\lambda_{ZZ}$.

The quality of the exponential interpolation can be used to infer the quality of the unlearnable noise extraction. For example, the irregular fluctuation of the data points in the lower panels of Figure~\ref{fig: pittsburgh panel top} accounts for the large error bar in panel (c) of Figure~\ref{fig: Mid-circuit benchmarking}. Similarly, the deviation from the linear trend in panel (h) of Figure~\ref{fig: pittsburgh panel bottom} leads to an underestimation of the MCM noise, which in turn inflates the estimated $\lambda_{ZZ}$ shown in the second bar of panel (d) in Figure~\ref{fig: Mid-circuit benchmarking} and pushes it outside the range allowed by the CPTP bound.

\begin{figure*}
    \centering
    \includegraphics[width=1\linewidth]{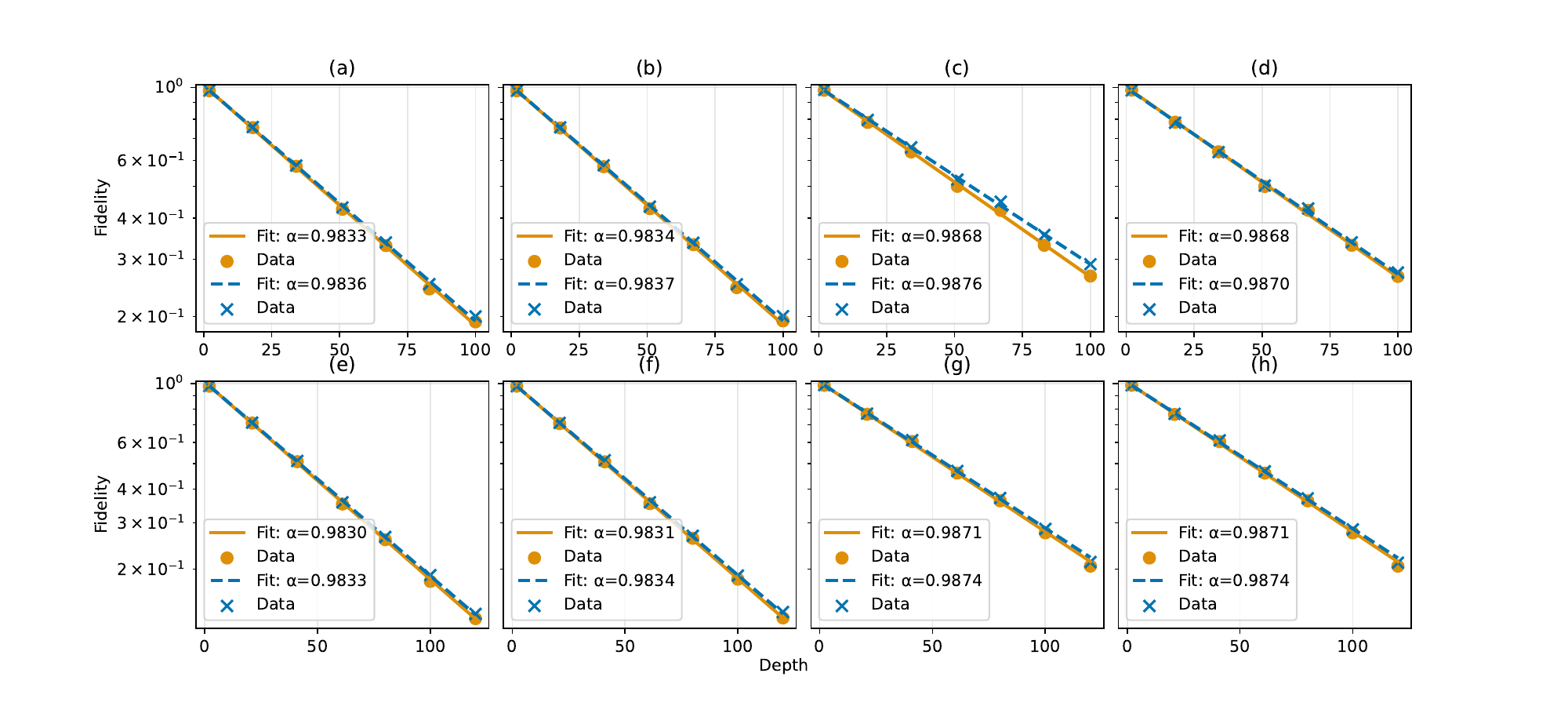}
    \caption{MCM-based CB interpolation for unlearnable fidelities of qubit pair [103, 204] of \texttt{ibm-aachen}. Rows from the top to the bottom respectivly represent data that is learnt from MCMs with/without single-qubit rotations. Each row contains all necessary information for extracting classical read-out errors and the unlearnable fidelities. Within each row, the left/right half of the panels are CB over MCMs with/without interleaving CNOT. (a)/(b), (c)/(d), (e)/(f), and (g)/(h) correspond to the expectation values measured for $IZ$/$ZZ$ at the end of the quantum circuit, through which the unlearnable pairs are extracted.}
    \label{fig: aachen panel top}
\end{figure*}
\begin{figure*}
    \centering
    \includegraphics[width=1\linewidth]{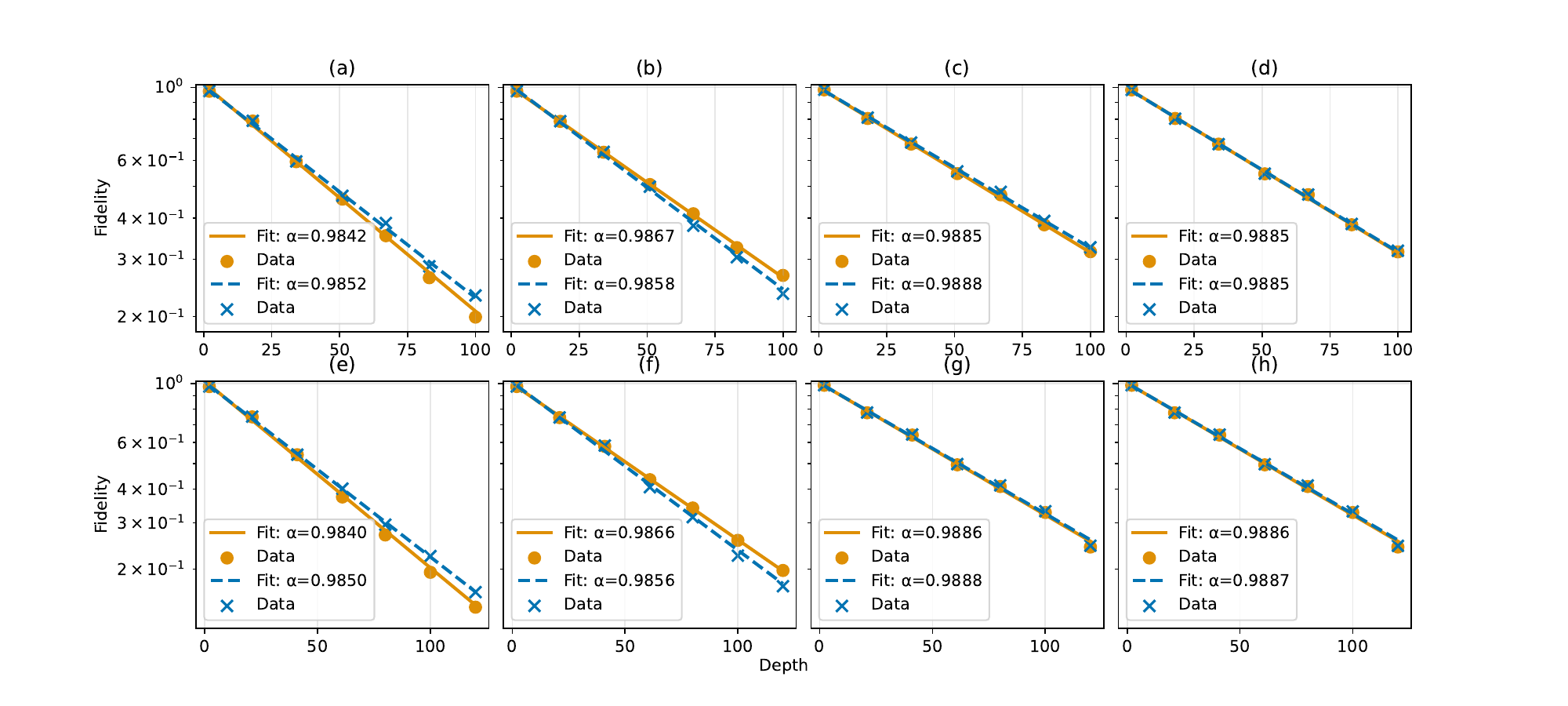}
    \caption{MCM-based CB interpolation for unlearnable fidelities of qubit pair [64, 63] of \texttt{ibm-aachen}. Panel layout as in Figure~\ref{fig: aachen panel top}.}
    \label{fig: aachen panel bottom}
\end{figure*}

\begin{figure*}
    \centering
    \includegraphics[width=1\linewidth]{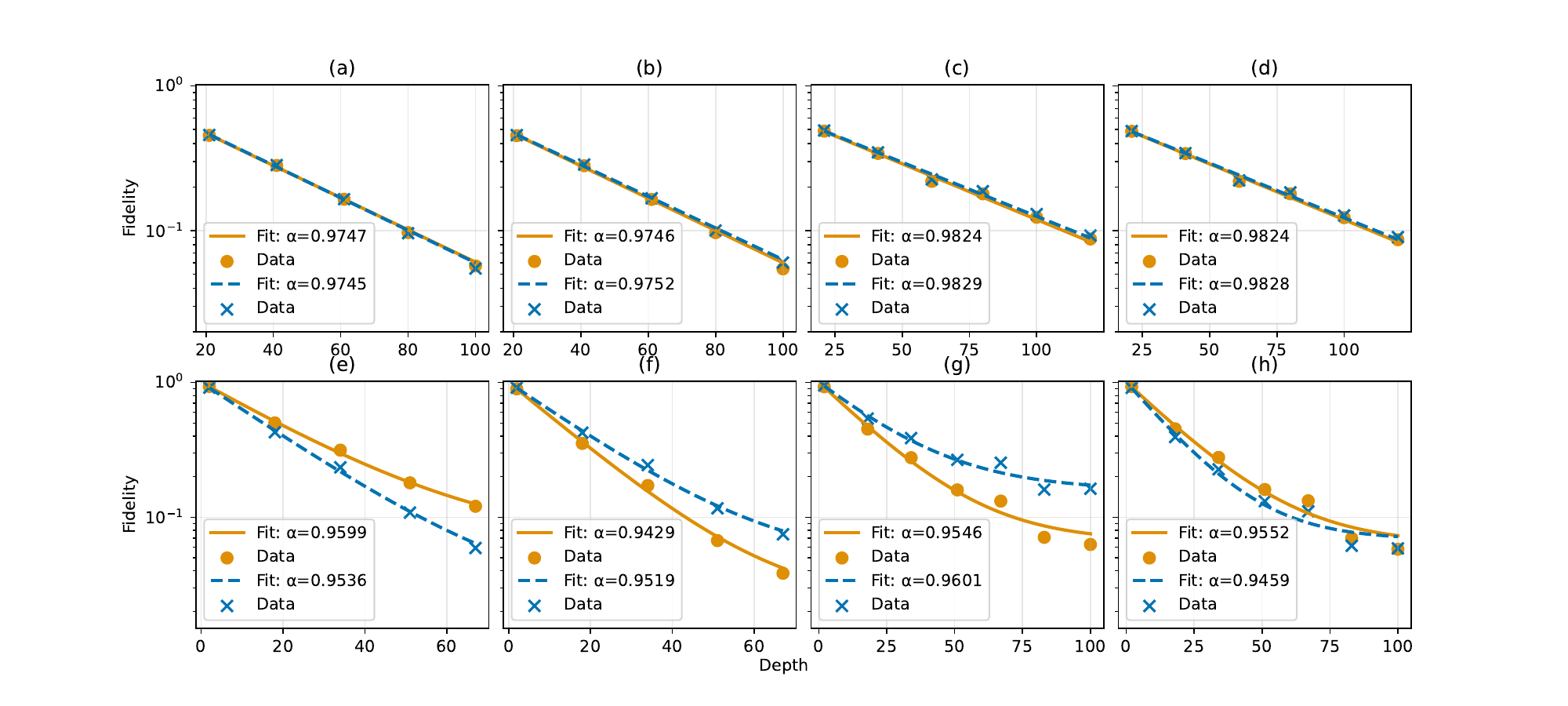}
    \caption{MCM-based CB interpolation for unlearnable fidelities of qubit pair [107, 108] of \texttt{ibm-pittsburgh}. Panel layout as in Figure~\ref{fig: aachen panel top}. Some data points were removed in the fitting process to obtain a consistent characterization. }
    \label{fig: pittsburgh panel top}
\end{figure*}

\begin{figure*}
    \centering
    \includegraphics[width=1\linewidth]{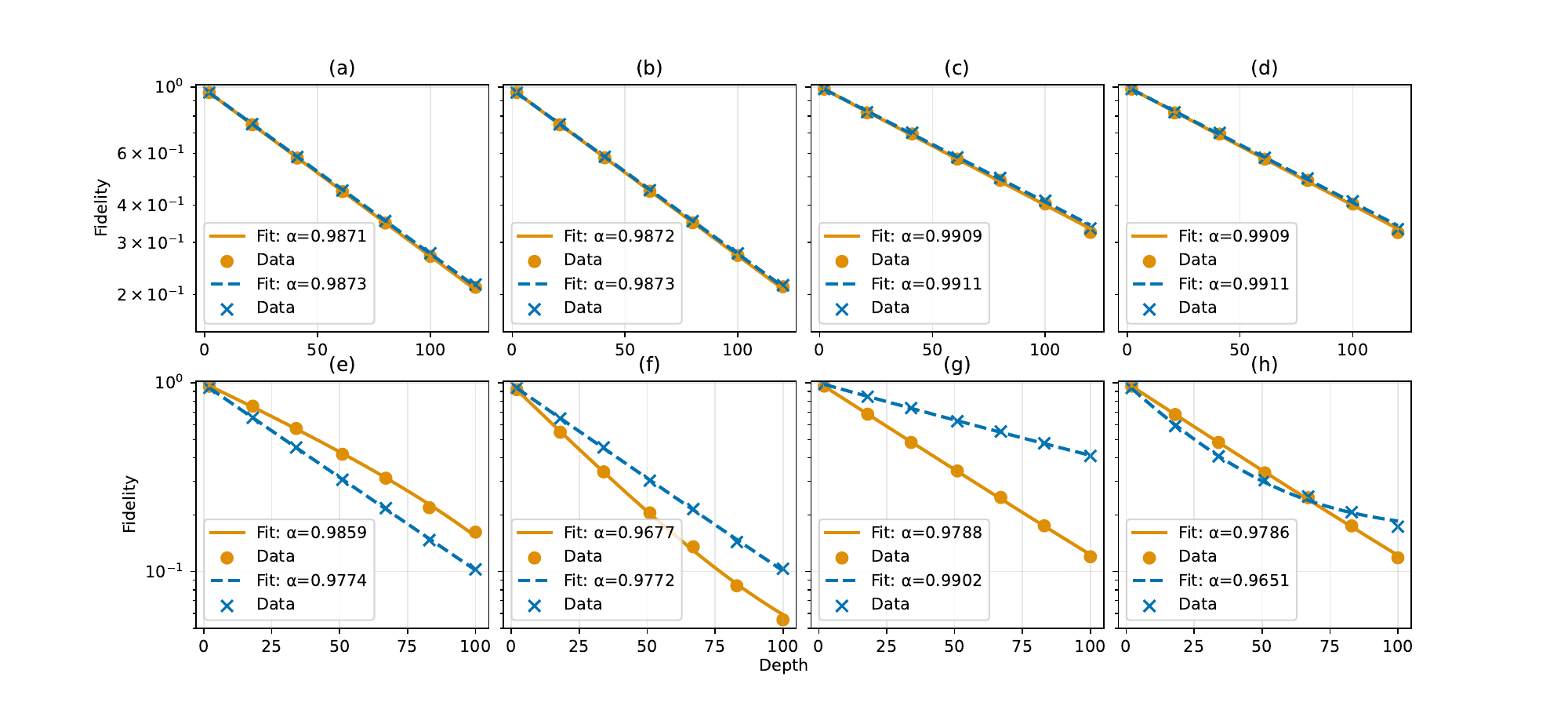}
    \caption{MCM-based CB interpolation for unlearnable fidelities of qubit pair [37, 45] of \texttt{ibm-pittsburgh}. Panel layout as in Figure~\ref{fig: aachen panel top}.}
    \label{fig: pittsburgh panel bottom}
\end{figure*}

\section{Non-Markovian Time Correlations}
\label{appendix: time correlation}
\begin{figure*}
    \centering
    \includegraphics[width=1\linewidth]{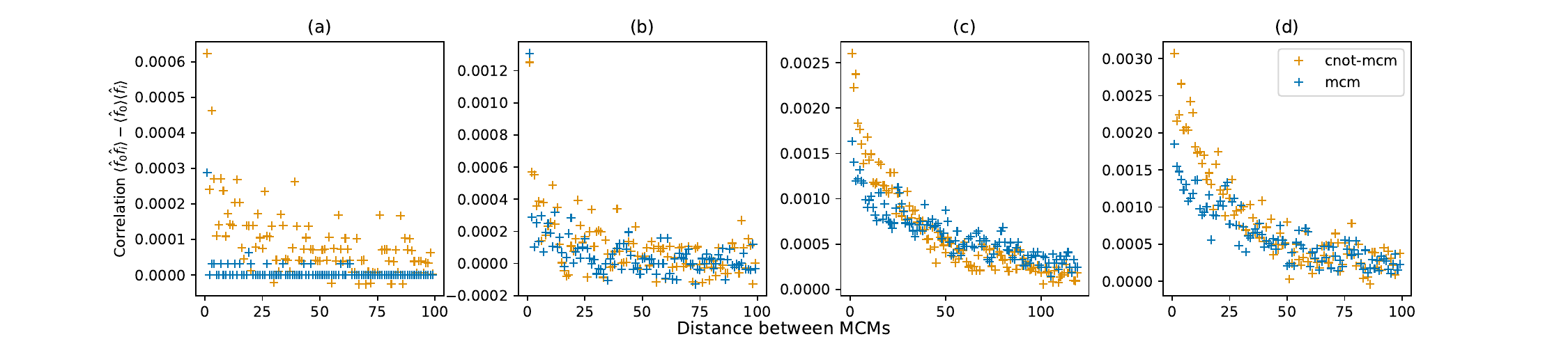}
    \caption{\justifying Time correlations between flips $\hat{f}_i = Z_{i-1}Z_i$ in mid-circuit measurements, corresponding to the binomial data set in Figure~\ref{fig:non-Markovian-plot} in the main text. Time correlation decay is captured in \texttt{ibm-pittsburgh} in (c) and (d) as the distance from the nearest MCMs to state preparation increases. Time correlation decay is also captured when the CNOT gate is employed in \texttt{ibm-aachen} in (a) and (b), but otherwise remains constant.}
    \label{fig:time-correlate}
\end{figure*}

Non-Markovianity is demonstrated in the binomial analysis of \texttt{ibm-aachen} and \texttt{ibm-pittsburgh} in Figure~\ref{fig:non-Markovian-plot}. Recall the flip variable introduced in Eq.~(\ref{eqn: flip observable}). We further compute the time correlations between each flip observable, $\langle \hat{f}_i\hat{f}_j\rangle-\langle \hat{f}_i\rangle\langle \hat{f}_j\rangle$, to verify the existence of time correlations in the experiment. The final result is illustrated in Figure~\ref{fig:time-correlate}.

Similarly to Figure~\ref{fig:non-Markovian-plot}, the progression of this figure from left to right represents the experiment on \texttt{ibm-aachen} without twirling, \texttt{ibm-aachen} with twirling, \texttt{ibm-pittsburgh} without rotation, and \texttt{ibm-pittsburgh} with rotation. The time correlations observed in the figure agree with the analysis conducted in the main text. Across all plots, the first pattern that emerges is that the presence of the CNOT gate \textit{increases} the time correlations. This is consistent with the increase in magnitude of the exponential decays in Figure~\ref{fig:non-Markovian-plot} with CNOT. Meanwhile, in panel (a), the observation of near-zero correlation in the MCMs experiment is consistent with the fact that classical read-out errors that induce asymmetries are time-independent events. Compared to panel (b), the time correlation in (c) and (d) is more pronounced, consistent with the idea that the magnitude of non-Markovianity in (c) and (d) is greater, as shown in Table~\ref{tab:fitted_params}.
\end{document}